\RequirePackage{luatex85}

\documentclass[11pt]{article}
\usepackage{iftex}

\ifPDFTeX
\usepackage[utf8]{inputenc}
\usepackage[T1]{fontenc}
\fi

\usepackage{amsmath}
\usepackage{amssymb}
\usepackage{amsthm}
\usepackage{thmtools}

\ifPDFTeX

\usepackage{calrsfs}
\DeclareMathAlphabet{\pazocal}{OMS}{zplm}{m}{n}

\usepackage{libertine}
\usepackage[scaled=0.96]{zi4} 
\else
\usepackage{fontspec}
\usepackage{unicode-math}
\fi

\usepackage[left=1in,top=1in,right=1in,bottom=1in]{geometry}

\usepackage{array}
\usepackage{paralist}
\usepackage{color}
\usepackage{bm,enumitem}
\usepackage{soul}
\usepackage{amsmath}
\usepackage{amsthm}

\usepackage{comment}
\usepackage{algorithmic}
\usepackage{amsfonts}
\usepackage{amssymb}
\usepackage[algo2e,ruled,linesnumbered,vlined,procnumbered]{algorithm2e} 
\makeatletter
\usepackage{xcolor}
\usepackage{amsthm, thmtools}
\usepackage{nameref}
\definecolor{ForestGreen}{rgb}{0.1333,0.5451,0.1333}
\definecolor{DarkRed}{rgb}{0.8,0,0}
\definecolor{Red}{rgb}{1,0,0}
\usepackage[linktocpage=true,
pagebackref=true,colorlinks,
linkcolor=DarkRed,citecolor=ForestGreen,
bookmarks,bookmarksopen,bookmarksnumbered]
{hyperref}

\usepackage[libertine]{newtxmath} 
\usepackage{cleveref}

\declaretheorem[numberwithin=section,refname={Theorem,Theorems},Refname={Theorem,Theorems},name={Theorem}]{theorem}
\declaretheorem[numberlike=theorem,refname={Lemma,Lemmas},Refname={Lemma,Lemmas},name={Lemma}]{lemma}
\declaretheorem[numberlike=theorem,refname={Corollary,Corollaries},Refname={Corollary,Corollaries},name={Corollary}]{corollary}

\declaretheorem[numberlike=theorem,refname={Fact,Facts},Refname={Fact,Facts},name={Fact}]{fact}

\declaretheorem[numberlike=theorem,refname={Proposition,Propositions},Refname={Proposition,Propositions},name={Proposition}]{proposition}

\declaretheorem[numberlike=theorem,refname={Definition,Definitions},Refname={Definition,Definitions},name={Definition}]{definition}

\declaretheorem[numberlike=theorem,refname={Claim,Claims},Refname={Claim,Claims}]{claim}

\def \eps {{\epsilon}}

\newcommand{\ww}{w}

\newcommand{\Property}{\mathcal{P}}
\newcommand{\poly}{\ensuremath{\mathrm{poly}}}

\def\Real{\mathbb{R}}
\newcommand\Otil{\tilde{\mathit{O}}}
\def\B{\mathbf{B}}
\def\C{\mathbf{C}}
\def\L{\mathbf{L}}

\def\WW{\mathcal{W}}

\def\poly{\operatorname{poly}}
\def\Real{\mathbb{R}}

\def\CHI{\boldsymbol{\chi}}
\def\SC{\mathbf{SC}}
\def\ER{\mathcal{R}_{\mathrm{eff}}}
\def\Pr{\mathbf{Pr}}

\newcommand\Htil{\widetilde{\mathit{H}}}
\newcommand\Ctil{\widetilde{\mathit{C}}}
\newcommand\bw{\mathbf{w}}

\newcommand\bx{\mathbf{x}}

\newcommand\blbd{\boldsymbol{\lambda}}
\newcommand\beps{\boldsymbol{\epsilon}}

\newcommand\CP{\mathcal{P}}

\newcommand\CG{\mathcal{G}}

\newcommand{\norm}[1]{\left\lVert#1\right\rVert}
\newcommand{\T}[1]{#1 ^\intercal}
\newcommand{\pinv}[1]{#1 ^\dagger}

\DeclareMathOperator{\maxflow}{max-flow}

\DeclareMathOperator{\dist}{dist}
\DeclareMathOperator{\expec}{\mathbb{E}}
\global\long\def\cutsp{\textsc{VertexSparsify}}

\makeatother

\begin{document}

\global\long\def\sep{\mathtt{Sep}}%
\global\long\def\P{\mathcal{P}}%
\global\long\def\D{\mathcal{D}}%

\title{Fast Dynamic Cuts, Distances and Effective Resistances via Vertex Sparsifiers}

\author{
		Li Chen\thanks{Georgia Institute of Technology, USA}
		\and
        Gramoz Goranci\thanks{
                University of Toronto, Canada} 
        \and 
        Monika Henzinger\thanks{
                University of Vienna, Austria}
        \and
        Richard Peng\thanks{Georgia Institute of Technology, USA}
        \and
        Thatchaphol Saranurak\thanks{
                Toyota Technological Institute at Chicago, USA}
}

\maketitle
\begin{abstract}

We present a general framework of designing efficient dynamic approximate algorithms for optimization on undirected graphs. In particular, we develop a technique that, given any problem that admits a certain notion of vertex sparsifiers, gives data structures that maintain approximate solutions in sub-linear update and query time. We illustrate the applicability of our paradigm to the following problems.

(1) A fully-dynamic algorithm that approximates all-pair
maximum-flows/minimum-cuts up to a nearly logarithmic factor in $\tilde{O}(n^{2/3})$ amortized time against an oblivious adversary, and $\tilde{O}(m^{3/4})$ time against an adaptive adversary.

(2) An incremental data structure that maintains $O(1)$-approximate shortest path in $n^{o(1)}$ time per operation, as well as fully dynamic approximate all-pair shortest path and transshipment in $\tilde{O}(n^{2/3+o(1)})$ amortized time per operation.

(3) A fully-dynamic algorithm that approximates all-pair effective resistance up to an $(1+\eps)$ factor in $\tilde{O}(n^{2/3+o(1)} \epsilon^{-O(1)})$ amortized update time per operation.

The key tool behind result (1) is the dynamic maintenance of an algorithmic construction due to Madry [FOCS' 10], which partitions a graph into a collection of simpler graph structures (known as j-trees) and approximately captures the cut-flow and metric structure of the graph. The $O(1)$-approximation guarantee of (2) is by adapting the distance oracles by [Thorup-Zwick JACM `05]. Result (3) is obtained by invoking the random-walk based spectral vertex sparsifier by [Durfee et al. STOC `19] in a hierarchical manner, while carefully keeping track of the recourse among levels in the hierarchy.

\end{abstract}
\newpage

\section{Introduction}
In the study of graph algorithms,
there are long-standing gaps in the performances
of static and dynamic algorithms.
A dynamic graph algorithm is a data structure that maintains a property
of a graph that undergoes edge insertions and deletions,
with the goal of minimizing the time per update and query operation.
Due to the prevalence of large evolving graph data in practice,
dynamic graph algorithms have natural connections with
network science~\cite{BorgwardtKW06,ParanjapeBL17},
and databases~\cite{Angles08,RobinsonWE13:book}.
However, compared to the wealth of tools available for static graphs,
it has proven to be much more difficult to develop algorithms for
dynamic graphs, especially fully dynamic ones undergoing both
edge insertions and deletions.
Even maintaining connectivity undirected graphs has witnessed 35 years
of continuous progress~\cite{NanongkaiS17,Wulff-Nilsen17,NanongkaiSW17}.
The directed version, fully dynamic transitive closure, has seen even
less progress~\cite{Sankowski04,RodittyZ08,BrandNS19}, and is one of the
best reflections of the difficulties of designing dynamic graph algorithms,
especially in practice~\cite{HanaulerHS20:arxiv}.

Over the past decade,
dynamic graph algorithms and their lower bounds
have been studied extensively.
These results led to significantly improved understandings of maintaining
many basic graph properties such as connectivity, maximal matching,
shortest paths, and transitive closure.
However, for many of these results there are linear or polynomial conditional lower
bounds for maintaining them exactly~\cite{AbboudW14,HenzingerKNS15,AbboudD16,Dahlgaard16}.
This shifted the focus to maintaining \emph{approximate} solutions to these problems,
and/or restricting the update operations to only insertions
(known as the {\em incremental} setting)
or only deletions (known as the {\em decremental} setting).

While this approach has led to much recent progress on shortest
path algorithms~\cite{Sankowski05,Thorup05,Bernstein09,RodittyZ12,AbrahamCT14,Bernstein16,AbrahamCK17,Chechik18}, there has been comparatively little
development in the maintenance of flows.
Flows and their associated dual labels, cuts,
are widely used in network analysis due to their ability
to track multiple paths and more global
information~\cite{HanPK11:book,BoykovVZ01,Zhu05}.
For example, the $st$-maximum flow problem asks for the maximum
number of edge disjoint paths between a pair of vertices~\cite{GoldbergT14},
while electrical flow minimizes a congestion measure related
to the sums of squares of the flow values along edges~\cite{DoyleS84}.
This need to track multiple paths has motivated the development
of new dynamic tools that eschew the tree-like structures typically
associated with problems such as connectivity and single-source
shortest paths~\cite{Goranci19:thesis}.
Such tools were recently used to give the first sublinear time
data structures for maintaining $(1 + \epsilon)$-approximate
electrical flows / effective resistances, which raised the
optimistic possibility that all flow related problems can be maintained
with $(1 + \epsilon)$-approximation factors in subpolynomial time~\cite{DurfeeGGP19}.

Motivated by interest in better understanding these problems, in this paper we present a general framework for designing efficient dynamic approximate algorithms for graph-based optimization problems in undirected graphs.
In particular, we develop a technique that reduces these problems to
finding a data-structure notion of vertex sparsifiers.
We then utilize this framework to study dynamic graph algorithms for flows,
with focus on obtaining the best approximation ratios possible,
but with sub-linear time per update/query. We achieve the following results:
\begin{enumerate}
\item Fully dynamic all pair max-flow/min-cut, shortest path,
and transshipment: $O(\log n \log \log n)$-approx. with $\Otil(n^{2/3 + o(1)})$\footnote{The $\tilde{O}(\cdot)$ notation is used in this paper to hide poly-logarithmic factors.} amortized edge update time and query time~(Theorems~\ref{theorem:DynamicMinCut} and \ref{theorem:DynamicAPSP}) against an oblivious adversary. Our dynamic max-flow algorithm can be extended to work against an adaptive adversary while increasing the update and query time to $\tilde{O}(m^{3/4})$~(Theorem~\ref{thm:dynamicMaxFlowAdaptive}).

\item Incremental all pair shortest path: $(2r-1)^t$-approximation with $\Otil(m^{1/(t+1)}n^{t/r})$ worst-case update and query time, where $t,r \geq 1$~(Theorem~\ref{thm: IncrementalApproximateAPSP}).
\item Fully dynamic all pair effective resistance for general weighted graphs:
$(1 + \eps)$-approximation with $\Otil(n^{2/3+o(1)} \epsilon^{-O(1)})$ amortized edge update time and query time,
improving upon the previous running time of $\Otil(n^{5/6} \epsilon^{-6})$~(Theorem~\ref{theorem:DynammicER}).
\end{enumerate}
In each of these three cases, our approximation ratios obtained
match up to constants the current best known approximation ratios
of oracles versions of these problems on static graphs,
namely oblivious routings~\cite{Racke08},
distance oracles~\cite{ThorupZ05},
and static computations of effective resistances~\cite{SpielmanS08:journal}.
This focus on approximation ratio is by choice: we believe just as with
static approximation and optimization algorithms, approximation ratios
should be prioritized over running times.
However, because all current efficient construction of edge sparsifiers
that preserve flows/cuts and resistances
with constant or better ($1 + \epsilon$) approximation are randomized,
all above algorithms except (2) are randomized, and their guarantees are
only provable against oblivious adversary
(who determines the hidden sequence of updates/queries beforehand).
We believe the design of more robust variants of our results
hinge upon the development of more robust edge sparsification tools,
which are interesting questions on their own.

%

Our techniques also extend to the offline dynamic setting, where the whole sequence of updates (edge insertions and deletions) and queries is given an advance.
In other words, the algorithm needs to output information about the graphs  at various points in this given update sequence.
Specifically, we show that for graph properties that admit efficient constructions of \emph{static} vertex sparsifiers, there are offline fully dynamic approximation algorithms with sub-linear average update and query time. We achieve the following results:

\begin{enumerate}
\item Offline fully dynamic all pair max-flow/min-cut: $O(\log^{4t} n)$-approximation with $\tilde{O}(m^{1/t+1})$ average update and query time, where $t \geq 1$~(Theorem~\ref{thm: offlineMaxFlow}).
\item Offline fully dynamic all shortest-path: $(2r-1)^{t}$-approximation with $\tilde{O}(m^{1/t+1}n^{2/r})$ average update and query time, where $t, r \geq 1$~(Theorem~\ref{thm: offlineShortestPaths}).
\end{enumerate}

Although the offline setting is a weaker than the standard dynamic
setting, it is interesting for two reasons. First, offline algorithms are used to obtain fast static algorithms (e.g. \cite{BringmannKN19,LPYZ18}).
Second, many conditional lower bounds (e.g. \cite{AbboudW14,AbboudD16,Dahlgaard16})
for the standard dynamic setting also hold for the offline dynamic setting.
Thus, giving an efficient algorithm for the offline dynamic setting
shows that no such conditional lower bound is possible. Moreover, for 
certain applications (e.g. computing ``sensitivity information'' for cretain graph properties) the sequence of updates is also known beforehand.



\subsection{Related work}

\paragraph{Previous results on dynamic flow/cuts}
Despite the fact that all pair max-flow/min-cut is one of the cornerstone problems combinatorial optimization and has been extensively studied in the static setting,
there are essentially no fast algorithms in the dynamic setting.
Using previous techniques, it is possible to get dynamic algorithms with $\tilde{O}(1)$
worst-case update time and $\tilde{O}(n)$ query time under the assumption 
that the adversary is oblivious.\footnote{We maintain a dynamic cut-sparsifier (against oblivious adversary)
of size $\tilde{O}(n)$ due to \cite{AbrahamDKKP16} with $\tilde{O}(1)$
update time, and when given a query, we execute the fastest static approximation algorithms on the sparsifier in $\tilde{O}(n)$ time (using, for example,
\cite{Peng16} for $(1+\eps)$-approximate max flow, \cite{Sherman17}
for $(1+\eps)$-approximate multi-commodity concurrent flow, and \cite{Sherman09}
for $O(\sqrt{\log n})$-approximate sparsest cuts).} To the best of our knowledge, there is no previous algorithm with both $o(n)$
update and query time, even when we are content with only amortized guarantees. 

Perhaps the closest work to this paper is the dynamic
algorithm due to \cite{CheungKL13} for explicitly maintaining the values
of all-pairs min-cuts in $\tilde{O}(m^{2})$ update time. For $s$-$t$ max flow where $s$ and $t$ are fixed, there is an incremental algorithm with $O(n)$ amortized update time \cite{GuptaK18}. If we restrict to bipartite graphs with a certain specific structure, there is a $(1+\eps)$-approximation fully dynamic algorithm \cite{AbrahamDKKP16} with polylogarithmic worst-case update time. From the lower bound perspective, Dahlgaard \cite{Dahlgaard16} showed
a conditional lower bound of $\Omega(n^{1-o(1)})$ on the amortized update
time for maintaining exact incremental $s$-$t$ max flow in \emph{weighted
undirected} graphs. This shows that approximation is necessary to
achieve sublinear running times. 

The \emph{global} minimum cut problem has been much better understood from the perspective of dynamic graphs. This is closely related to a similar phenomenon in the static setting, where in contrast to the $s-t$ min-cut problem, its global counterpart admits arguably simpler and easier algorithms. The best-known fully-dynamic algorithm is due to Thorup~\cite{Thorup07}, who maintains a $(1+o(1))$-approximation to the value of global minimum cut using $\tilde{O}(\sqrt{n})$ update and query time. When the graph undergoing updates remains planar, Lacki and Sankowski~\cite{LackiS11} showed an exact fully-dynamic algorithm with $\tilde{O}(n^{5/6})$ update and query time. Recently, Goranci, Henzinger and Thorup~\cite{GoranciHT18} designed an exact incremental algorithm with $O(\log^3 n \log \log^2 n)$ update time and $O(1)$ query time. 

\paragraph{Previous graph sparsification in the dynamic setting.}

Many previous works in dynamic graph algorithms are based on \emph{edge sparsification}. This usually allows algorithms to assume that an underlying dynamic graphs is always sparse and hence speed up the running time. To the best of our knowledge, the first paper that applies edge sparsification in the dynamic setting is by Eppstein et al. \cite{EppsteinGIN97}. This work has proven useful for several fundamental problems including dynamic minimum spanning forest and different variants of edge/vertex connectivity~(e.g. \cite{Thorup07,NanongkaiS17,Wulff-Nilsen17,NanongkaiSW17}). Edge sparsification has been also a key technique in dynamic shortest paths problems (e.g. \cite{BernsteinR11} maintains distances on top of spanners, \cite{BernsteinC16,BernsteinC17} replace “dense parts” of graphs with sparser graphs). Recently, there are works that study edge sparsification for matching problems~ \cite{BernsteinS15,BernsteinS16,Solomon18}. In fact, the core component of several dynamic matching algorithms is only to maintain such sparsifiers \cite{BernsteinS15,BernsteinS16}.

There are also previous developments in dynamic graph algorithms based on \emph{vertex sparsification} which allow algorithms to work on graphs with smaller number of vertices. This usually offers a more significant speed up than edge sparsification. Earlier works \cite{EppsteinGIS96,EppsteinGIS98,FakcharoenpholR01} that utilize vertex sparsification in the dynamic setting are restricted to planar graphs and exploit the fact that this class of graphs admit small separators. Similar techniques are used and generalized in \cite{GoranciHP17a,GoranciHP18} but none of these works extend to general graphs. Several previous \emph{offline} dynamic algorithms exploit vertex sparsification for maintaining minimum spanning forests~\cite{Eppstein91}, small edge/vertex connectivity~\cite{PengSS19}, and effective resistance~\cite{LiPYZ19}.

\paragraph{Recent Results on Dynamic Vertex Sparsification} Very recently, Goranci et 
al.~\cite{GoranciRST20} give a fully dynamic algorithm for 
maintaining a tree flow sparsifier based on a new notion of expander 
decomposition. One of their applications is a fully dynamic algorithm for $s$-$t$ 
maximum flow and minimum cuts. Their algorithm is deterministic, has $n^{o(1)}$ 
worst-case update time and $O(\log^{1/6}n)$ query time, but their approximation ratio 
is $2^{O(\log^{5/6} n)} = n^{o(1)}$. Our algorithms from
Theorems~\ref{theorem:DynamicMinCut}~and~\ref{thm:dynamicMaxFlowAdaptive}
guarantees a much better  approximation ratio of $O(\log n (\log\log n)^{O(1)})$.
However, our update and query times are slower, and are randomized.

Concurrent to our result there have also been several recent developments
on utilizing vertex sparsifiers to maintain $c$-edge connectivity for small values
of $c$~\cite{PengSS19,ChalermsookDLV19:arxiv,LiuPS19:arxiv,JinS20:arxiv}. 

\subsection{Technical Overview}
We start by discussing an incremental version of our meta theorem, which is key to our incremental all pair shortest path algorithm. The main algorithmic tool behind our construction is a data-structure version of the well-studied notion of \emph{vertex sparsifier}~\cite{Moitra09,leighton,charikar,juliasteiner,mm10,EnglertGKRTT14}, which we refer to as \emph{incremental vertex sparsifier}. To better convey our intuition, we start with a slightly weaker definition of such a sparsifier, which already leads to non-trivial guarantees. We then discuss the generalization and its implications. 

Let $G=(V,E)$ be an $n$-vertex graph and, for each $u,v\in V$, let $\Property(u,v,G)$
denote a \emph{property }between $u$ and $v$ in $G$\footnote{Our approach also works for graph properties with a number of parameters that is different from 2.}. For example,
$\Property(u,v,G)$ could be the distance between $u$ and $v$ in $G$. Let
$T\subseteq V$ be a set of nodes called \emph{terminals}. Given a parameter $\alpha$, an $\alpha$-\emph{vertex sparsifier}
of $G$ w.r.t. $T$ is a graph $H=(V',E')$ such that 1) $V' \supseteq T, |V'|\approx |T|$
and 2) $\Property(u,v,H) \approx_{\alpha} \Property(u,v,G)$ for all $u,v\in T$. That is,
$H$ has size close to $T$ but still approximately preserves
the property $\Property$ between all terminal nodes up to a factor of $\alpha$. 

Given a graph $G=(V,E)$ and terminals $T \subseteq V(G)$, an $\alpha$-\emph{incremental vertex sparsifier}~(IVS) of $G$ is a data structure that maintains an $\alpha$-vertex sparsifier $H_T$ and supports the following operations: 
\begin{itemize} 
	\item $\textsc{Preprocess}(G,\alpha)$: preprocess the graph $G$,
\item \textsc{AddTerminal}$(u)$: let $T' \gets T \cup \{u\}$ and update $H_{T'}$ to an $\alpha$-vertex sparsifier of $G$ w.r.t. $T'$. 
\end{itemize}

An \emph{efficient} $(\alpha,f(n),g(n))$-IVS of $G$ is an $\alpha$-IVS of $G$ that supports the preprocessing and terminal addition operations in $O(|E|f(n))$ and $O(g(n))$ time, respectively. 

The advantage of such an efficient sparsifier is that it immediately leads to a simple two-level incremental algorithm. Concretely, given an initial graph $G=(V,E)$ and an approximation parameter $\alpha \geq 1$, assume we want to design an incremental algorithm that maintains some property $\mathcal{P}(s,t,G)$
that can be computed in time $O(|E|h(n))$ on a static graph $G=(V,E)$.
To achieve this, our data-structure maintains
(1) an efficient $(\alpha, f(n), g(n))$-IVS of $G$ and 
(2) a set of terminals $T$, which is initially set to empty.
We initialize our data-structure using the \textsc{Preprocess}$(G,\alpha)$ operation of the efficient IVS and rebuild from scratch every $\beta$ operations, for some parameter $\beta \geq 0$. Note that after a rebuild, $H_T$ is empty. We next describe the implementation of insertions and queries. 
Upon insertion of a new edge $e=(u,v)$ in $G$, we invoke \textsc{AddTerminal}$(u)$ and \textsc{AddTerminal}$(v)$, and add $e$ to $H_T$. For answering $(s,t)$ queries, we invoke \textsc{AddTerminal}$(s)$ and \textsc{AddTerminal}$(t)$, and run a static algorithm on $H_T$ that computes property $\mathcal{P}(s,t, H_T)$ and return the result as an answer.

As $H_T$ is an $\alpha$-vertex sparsifier of the current graph $G$ and $T \supseteq \{s,t\}$ by construction, we have that $\mathcal{P}(s,t, H_T)$ approximates property $\mathcal{P}(s,t,G)$ up to an $\alpha$ factor. The update time consists of (1) the cost for rebuilding every $\beta$ operations, which is $O(mf(n)/\beta)$ and (2) the cost for adding endpoints of $\beta$ edges 
as terminals, which is $O(\beta g(n))$. By construction, $|T| = O(\beta)$ at any time, resulting in a size of $O(\beta g(n))$ for $H_T$ (since we start with an empty $H_T$ after a rebuild). As the static algorithm
on $H_T$ takes time $O(|H_T| h(n))$, the query time is bounded by $O(\beta g(n) h(n))$.

Combining the above bounds on the update and query time, we obtain the following expression
\[
O \left(\left(\frac{m}{\beta}\right)f(n) + \beta g(n)h(n)\right)
\] 
which bounds the amortized update time and worst-case query time.

A challenge we face to design a multi-level \emph{incremental} algorithm using the above approach is the large flexibility allowed in the \textsc{AddTerminal} operation. More concretely, given a sparsifier $H_T$ w.r.t. $T$, whenever we add a new terminal to $T$ and update $H_T$ to $H_{T'}$, the implementation of the operation could potentially both delete and insert vertices/edges to $H_T$. Ideally we would like that the $H_{T'}$ is constructed by \emph{only} adding new vertices/edges to $H_T$, which in turn would allow us to keep the incremental nature of the problem. To address this, we  modify the operation of adding terminals in the definition of $\alpha$-IVS as follows:
\begin{itemize}
	\item \textsc{AddTerminal}$(u)$: let $T'$ be $T \cup \{u\}$ and update $H_T$ to $H_{T'}$ such that
	\begin{itemize}
		\item $H_{T'}$ is an $\alpha$-vertex sparsifier of $G$ w.r.t. $T'$.
		\item $H_{T} \subseteq H_{T'}$.
	\end{itemize}
\end{itemize}

An important measure related to this new definition is the  notion of \emph{recourse}, which is the number of changes performed to the old sparsifier $H_T$, i.e., $|H_{T'} \setminus H_{T}|$. While it is straightforward to bound the recourse by the time needed to support the addition of terminals, there are scenarios where recourse can be much smaller. Equipped with the new definition of incremental vertex sparsifier and the notion of recourse, we immediately get a multi-level hierarchy for designing incremental algorithms, which is formally stated in Theorem~\ref{thm: metaTheorem} of Section~\ref{sec: incrementalMeta}.

We demonstrate the applicability our meta theorem to the incremental all-pair shortest path problem by showing that an efficient IVS can be constructed using a deterministic variant of the distance oracle due to Thorup and Zwick~\cite{ThorupZ05}. At a high level, the oracle preprocessing works as follows: (1) it constructs a hierarchy of centers, (2) for each vertex it finds the closest center at every level of the hierarchy and (3) for each vertex $u \in V$, it defines the notion of \emph{bunch} $B(u)$, which is the union of all centers of $u$ that we found in (2). Our key observation is that this construction leads to an efficient IVS with bounded recourse: (1) we preprocess the graph using the oracle preprocessing, and (2) implement the terminal addition of a vertex $u$ by simply adding its bunch $B(u)$ to the current vertex sparsifier maintained by the data structure. This construction leads to an efficient $(\tilde{O}(n^{1/r}), \tilde{O}(n^{1/r}),2r-1)$-IVS for $G$. The correctness of our data-structure heavily relies on the fact $\cup_{u \in T} B(u)$ is an $\alpha$-vertex (distance) sparsifier of $G$ w.r.t. $T$.

To extend our meta theorem to fully-dynamic graphs, we simply go back to the old implementation of the \textsc{AddTerminal} operation (as now we can support insertions/deletions of edges) and then augment our data-structure with the operation \textsc{Delete}$(e)$, which allows deleting $e$ from the underlying graph $G$ and updates the maintained vertex sparsifier with respect to this deletion. Similarly to the above, bounding the recourse and performing recursive invocations of the vertex sparsifier data-structure leads to a meta theorem for fully dynamic algorithms, which is formally stated in Theorem~\ref{metaThM: P1} of Section~\ref{sec: fullyDynamicMetaTHm}.

Our fully-dynamic results on flow/cuts, distances and electrical flow/effective resistances, which are based on our generic meta-theorem, are obtained by adapting:

\begin{enumerate}
	\item the $j$-tree decomposition ofgraphs by Madry~\cite{Madry10},
	which is in-turn based on the oblivious routing scheme by Racke~\cite{Racke08},
	\item the random-walk interpretation of electrical flow preserving
	vertex sparsifiers (also known as Schur complements) that are
	also at the core of~\cite{DurfeeGGP19}.
\end{enumerate}

In each of these cases, the approximation ratios obtained by our
data structures match the current best bounds of static variants
of these problems, which are themselves well-studied.
Specifically, our approximation ratios for these three problems are
identical to those of sublinear time query oracles for answering
(multi source/sink) min-cut queries\footnote{The Gomory-Hu tree
	on the other hand provides a much more efficient query oracle
	for answering $s$-$t$ min-cut queries, but we are not aware of
	generalizations of it to small sets of source/sink vertices},
distances, and effective resistances in static graphs.
As a result, we believe our results represent natural starting points
for more efficient versions of these data structures, and hope
they will motivate further work on the static query versions
of flows/cuts and shortest paths.

\section{Incremental Algorithms via Incremental Vertex Sparsifiers}
\label{sec: incrementalMeta}

Let $G=(V,E)$ be a graph. Let $\P$ be a property of graphs. Properties can be either (i) \emph{global}, if $\P$ is meant for the entire graph or  (ii) \emph{local}, if $\P$ is defined with respect to a particular pair of vertices $(u,v)$ in the graph. For example, $\P(u,v)$ is the distance between $u$ and $v$ in $G$ or the size of the $u-v$ minimum cut. To simplify presentation, we assume throughout the rest of this paper that $\P$ takes three parameters, two vertices, followed by the graph. Let $T \subseteq V$ be a set of nodes called \emph{terminals}. An $\alpha-$\emph{vertex sparsifier} of $G$ for $\P$ w.r.t. $T$ is a graph $H=(V',E')$ such that i) $V' \supseteq T, |V'| \approx |T|$,  and ii) $\P(u,v,H) \approx_{\alpha} \P(u,v,G)$ for all $u,v \in T$. That is, $H$ is close to $T$ and at the same time approximately preserves the property $\P$ between pairs of terminal nodes up to a factor of $\alpha$.

The notion of \emph{recourse} is used for measuring the number of changes in the target graph maintained by a data structure.

\begin{definition}[Incremental Vertex Sparsifer]  \label{def: IVS}
Let $G=(V,E)$ be a graph, $T \subseteq V$ be the set of terminals, and $\alpha$ be an non-negative parameter. A data-structure $\D$ is an \emph{$\alpha-$Incremental Vertex Sparsifier} (abbr. \emph{$\alpha-$IVS}) of $G$ if $\D$ explicitly maintains an $\alpha-$vertex sparsifier $H_T$ of $G$ with respect to $T$ and supports following operations: 
\begin{enumerate}
\setlength{\itemsep}{0em}
\item \textsc{Preprocess}$(G, T)$: preprocess $G$ in time $t_p$
\item \textsc{AddTerminal}$(u)$: let $T'$ be $T \cup \{u\}$ and update $H_T$ to $H_{T'}$ in time $t_u$ such that i) $H_{T'}$ is an $\alpha-$vertex sparsifier of $G$ with respect to $T'$, ii) $H_T \subseteq H_{T'}$ and iii) \emph{recourse}, number of edge changes from $H_T$ to $H_{T'}$, is at most $r_u$.
This operation should return the set of edges $H_{T'} \setminus H_T$.
\end{enumerate}
Also, the size of the vertex sparsifier, $H_T$, should be $O(|T|r_u)$.
\end{definition} 

The simplest example of such \emph{Incremental Vertex Sparsifer} can be made as follows:
$\D$ maintains $G$ as the vertex sparsifier.
Such $\D$ is a $1$-IVS of $G$ with preprocessing time $t_p$ update time $t_u$ and recourse $r_u$ being $O(|E(G)|)$.

Given a $\alpha-$IVS $\D$ of $G$, we can support edge insertion by first adding 2 endpoints to the terminal set and then add the edge directly to the vertex sparsifier $H_T$.
The correctness comes from the decomposability.
To specify the cost, the following corollary is presented to formalized the approach.
\begin{lemma}
\label{lem:IVS1levelEdgeInsertion}
Let $G=(V,E)$ be a graph and $\D$ be an $\alpha-$IVS of $G$.
$\D$ also maintains an $\alpha-$vertex sparsifier $H_T$ of $G$ with respect to $T$ subject to the following operation:
\begin{itemize}
\setlength{\itemsep}{0em}
\item \textsc{Insert}$(e=uv)$: insert edge $e$ to $G$ and update $T$ and $H_T$. It is done in $O(t_u)$-time with $O(r_u)$-recourse.
\end{itemize}
\end{lemma}
\begin{proof}

\begin{algorithm2e}
\label{algo:IVS1levelEdgeInsertion}
\caption{$\D.\textsc{Insert}(e=uv)$}
    $E_u \gets \D.\textsc{AddTerminal}(u)$. \\
    $E_v \gets \D.\textsc{AddTerminal}(v)$. \\
    Add edge $e$ to $H_T$. \\
    \Return $E_u \cup E_v \cup \{e\}$. \\
\end{algorithm2e}

The algorithm for edge insertion is presented as Algorithm~\ref{algo:IVS1levelEdgeInsertion}.
The correctness comes from the decomposability of the desired graph property.
The time bound is $2t_u + O(1)$ for 2 $\textsc{AddTerminal}$ operations and edge insertion in a graph.
\end{proof}

Since our ultimate goal is to design incremental algorithms that support updates and queries in sub-linear time, we will focus on building incremental vertex sparsifiers whose update time and recourse are sub-linear in $n$. This requirement is made precise in the following definition.

\begin{definition}[Efficient IVS] \label{def: efficientIVS}
Let $G=(V,E)$ be a graph, $\alpha$ be an non-negative parameter, and $f(n)$, $g(n)$, $r(n) \geq 1$ be functions. We say that $\D$ is an $(\alpha,f(n),g(n),r(n))-$\emph{efficient IVS} of $G$ if $\D$ is an $\alpha-$\emph{IVS} of $G$ with \emph{preprocessing} time $t_p=O(m \cdot f(n))$, \emph{addTerminal} time $t_u = O(g(n))$, and recourse $r_u = O(r(n))$.
\end{definition}


\label{sec: incrementalVertexSparsifier}
Next we show how to use efficient incremental vertex sparsifier to design online (approximate) incremental algorithms for problems with certain properties while achieving fast amortized update and query time.

\begin{theorem} \label{thm: metaTheorem}
Let $G=(V,E)$ be a graph, and for any $u,v \in V$, let $\mathcal{P}(u,v,G)$ be a solution to a minimization problem between $u$ and $v$ in $G$. Let $f(n),g(n),r(n),h(n) \geq 1$ be functions, $\alpha\geq 1$ be the approximation factor, $\ell\geq 1$ be the depth of the data structure, and let $\gamma,\mu_0,\mu_1,\ldots, \mu_{\ell}$ with $\mu_0 = m$ be parameters associated with the running time. Assume the following properties are satisfied
\begin{enumerate}
\setlength{\itemsep}{0em}
\label{metaThM: P1}\item $G$ admits an $(\alpha, f(n),g(n),r(n))-$\emph{efficient IVS}
 \item The property $\mathcal{P}(u,v,G)$ can be computed in $O(m h(n))$ time in a static graph with $m$ edges and $n$ vertices.
\end{enumerate}
Then there is an incremental (approximate) dynamic algorithm that maintains for every pair of nodes $u$ and $v$, an estimate $\delta(u,v)$, such that
\begin{equation} \label{eq: approxMeta}
        \mathcal{P}(u,v,G) \leq \delta(u,v) \leq \alpha^{\ell} \cdot \mathcal{P}(u,v,G),
\end{equation}
with worst-case update time of
\[
        T_u = O\left(\sum_{i=1}^{\ell}\left(\frac{\sum_{j=i}^{\ell}\mu_{j-1}f(\mu_{j-1})}{\mu_i} + g(\mu_{i-1})\right)\prod_{j=0}^{i-1}r(\mu_j)c^{i-1}\right),
\]
and worst-case query time of
\[
        T_q = O\left(\ell T_u + \mu_\ell r(\mu_\ell) h(\mu_\ell) \right),
\]
where $c < 3$ is an universal constant.
\end{theorem}

We declare this section to prove the Theorem~\ref{thm: metaTheorem}

\paragraph*{Data Structure.}

Consider some integer parameter $\ell \geq 1$ and parameters $\mu_0 \geq \ldots \geq \mu_{\ell}$, with $\mu_0 = m$.
Our data structure maintains
\begin{enumerate}[noitemsep]
\item a hierarchy of graphs $\{G_i\}_{0 \le i \le \ell}$,
\item a hierarchy of terminal sets $\{T_i\}_{1 \le i \le \ell}$, each associated with the parameters $\{\mu_{i}\}_{1 \leq i \leq \ell}$, and
\item a hierarchy of $(\alpha, f(n), g(n), r(n))$-efficient IVSs $\{\D_i\}_{1 \le i \le \ell}$, each associated with the graph $G_{i-1}$ and the terminal set $T_i$.
\end{enumerate}

The data structure is initialized recursively.
First, initialize $T_i \gets \emptyset$ for $1 \le i \le \ell$ and $G_0 \gets G$.
For $1 \le i \le \ell$, construct an $(\alpha, f(n), g(n), r(n))$-efficient IVS, $\D_i$, of $G_{i-1}$ with respect to terminal set $T_i$ and set $G_i$ be the sparsifier maintained by $\D_i$.
That is, every $G_i$ is an $\alpha$-vertex sparsifier of $G_{i-1}$.
By the transitivity, we know $G_{\ell}$ is an $\alpha^\ell$-vertex sparsifier of $G_\ell$.
Thus, we compute the estimate $\delta(u,v)$ in $G_{\ell-1}$, which should be the smallest graph we have.

To deal with edge insertion, we add both endpoints of the new edge to the terminal set of every IVS $\D_i$.
Note that an edge insertion in $\D_i$ produces $O(r(\mu_{i-1}))$ more insertions in the next level of IVS.
Also, to bound the sparsifier size, which is related to the size of the terminal set, we have to rebuild $\D_i$ whenever $\D_i$ has handled $\mu_i$ updates since the last rebuild.
A counter $c_i$ is used for keeping track of the number of updates handled by $\D_i$.
The rebuild cost is amortized over $\mu_i$ updates.
To get worst-case cost, we use the standard reduction of creating copies of the data structure in the background.

\paragraph*{Handling Insertions.}

Consider the insertion of edge $e=uv$ in $G$, which is $G_0$ with an $\alpha$-IVS $\D_1$.
We insert $e$ to $\D_1$ via Corollary~\ref{lem:IVS1levelEdgeInsertion}.
Each edge insertion to $G_i$ handled by $\D_{i+1}$ creates $O(r(\mu_i))$ edge insertions in the resulting vertex sparsifier, $G_{i+1}$.
Thus one edge insertion in $G$ creates $O(c^i\prod_{j=0}^{i-1}r(\mu_j))$ edge insertions in $G_{i}$.

To bound the size of the terminal sets, we rebuild each $\D_i$ every $2\mu_i$ updates to the graph $G_{i-1}$ and also rebuild every $\D_j, i < j$ that gets affected.
When rebuilding some $\D_i$, we rebuild it with respect to the terminal set $T_i$ containing the latest-added $\mu_i$ vertices.

This algorithm is depicted in Algorithm~\ref{alg: IncrementalInsert}.

\paragraph*{Handling Queries.}
To answer the query for the approximate property $\mathcal{P}(s,t,G)$ between any pair of vertices $s$ and $t$ in $G$ we proceed as follows.
We first add $s$ and $t$ to every terminal set of $\{\D_i\}$.
The algorithm for adding terminal in our hierarchical data structure is presented in Algorithm~\ref{alg: IncrementalAddTerminal}.

Then we run the algorithm from Theorem~\ref{thm: metaTheorem}~Part~2 on $G_{\ell}$, which is maintained by IVS $\D_\ell$, to calculate the property $\mathcal{P}(s,t,G_{\ell})$ between $s$ and $t$ in $G_{\ell}$.
The value $\mathcal{P}(s,t,G_{\ell})$ is returned as an estimate to $\mathcal{P}(s,t,G)$.
This algorithm is depicted in Algorithm~\ref{alg: IncrementalQuery}.

\begin{algorithm2e}[t!]
\caption{\textsc{Rebuild}$(i)$}
\label{alg: IncrementalRebuild}
\For{$j \in \{i, \ldots, \ell\}$}{
    Remove vertices from $T_j$ except the newest $\mu_{j}$ vertices \\
    $\D_j.\textsc{Initialize}(G_{j-1}, T_j)$ \\
    Set $G_j$ to be the vertex sparsifier maintained by $\D_j$ \\
    Set $c_j \gets 0$ \\
}
\end{algorithm2e}

\begin{algorithm2e}[t!]
\caption{\textsc{Insert}$(e)$}
\label{alg: IncrementalInsert}
\For{$i \in \{0, 1, \ldots, \ell\}$}{
    Set $E_{i+1} \gets \phi$ \\
    \For{$f \in E_i$}{
        Perform $\D_{i+1}.\textsc{Insert}(f)$ and add the changes of $G_{i+1}$ to $E_{i+1}$ \\
        Set $c_{i+1} \gets c_{i+1} + 1$ \\
        \If{$c_{i+1} \geq 2\mu_{i+1}$}{
            \textsc{Rebuild}($i+1$) \\
            Break the loop \\
        }
    }
}
\end{algorithm2e}

\begin{algorithm2e}[t!]
\caption{\textsc{AddTerminal}$(u)$}
\label{alg: IncrementalAddTerminal}
Set $E_0 \gets \phi$ \\
\For{$i \in \{0, 1, \ldots, \ell-1\}$}{
    Set $E_{i+1} \gets \phi$ \\
    \For{$f \in E_i$}{
        Perform $\D_{i+1}.\textsc{Insert}(f)$ and add the changes of $G_{i+1}$ to $E_{i+1}$ \\
        Set $c_{i+1} \gets c_{i+1} + 1$ \\
        \If{$c_{i+1} \geq 2\mu_{i+1}$}{
            \textsc{Rebuild}($i+1$) \\
            Set $E_{i+1} \gets \phi$ \\
            Break the loop \\
        }
    }
    Perform $\D_{i+1}.\textsc{AddTerminal}(u)$ and add the changes of $G_{i+1}$ to $E_{i+1}$ \\
    Set $c_{i+1} \gets c_{i+1} + 1$ \\
    \If{$c_{i+1} \geq 2\mu_{i+1}$}{
        \textsc{Rebuild}($i+1$) \\
        Set $E_{i+1} \gets \phi$ \\
    }
}
\end{algorithm2e}

\begin{algorithm2e}[t!]
\caption{\textsc{Query}$(s,t)$}
\label{alg: IncrementalQuery}
\textsc{AddTerminal}$(s)$ \\
\textsc{AddTerminal}$(t)$ \\
Run static algorithm on $G_\ell$ to compute $\mathcal{P}(s, t, G_\ell)$ \\
\Return $\mathcal{P}(s, t, G_\ell)$ \\
\end{algorithm2e}

\paragraph*{Correctness.}
Let $G$ be the current graph throughout the execution of the algorithm.
Via induction, we know every data structure $\D_i$ in the hierarchy is an $\alpha$-IVS of $G_{i-1}$ with respect to the terminal set $T_i$.
Hence, $G_i$, maintained by $\D_i$, is an $\alpha$-vertex sparsifier of $G_{i-1}$.
Therefore via transitivity of vertex sparsifier, we know $G_{\ell}$ is an $\alpha^{\ell}$-vertex sparsifier of $G_0$, which is the current graph $G$.

Thus, we have
\[
    \mathcal{P}(s,t,G) \le \mathcal{P}(s,t,G_{\ell-1})=\delta_{\D_\ell}(s,t) \le \alpha^{\ell} \mathcal{P}(s,t,G).
\]
Therefore the approximation claim in Theorem~\ref{thm: metaTheorem} is proved.

\paragraph*{Running Time.}

We first study the update time of our data structure.
Since rebuild is incurred every $\mu_i$ operations for the data structure $H_i$,
we can charge the rebuild cost among $\mu_i$ operations.
Note that $\D_i$ is an $\alpha$-IVS of $G_{i-1}$, which is a graph on $\mu_{i-1}$-vertices.
By the size bound in Definition~\ref{def: IVS}, rebuild cost is $O(\mu_{i-1}r(\mu_{i-1})f(\mu_{i-1}))$.
Also notice all $\D_j, i < j \le \ell$ are also rebuilt, each has rebuild cost $O(\mu_{j-1}r(\mu_{j-1})f(\mu_{j-1}))$.
By amortizing the rebuild cost, we know the time $\D_i$ spent on either \textsc{AddTerminal} or \textsc{Insert} is:
\begin{align*}
    O\left(\frac{\sum_{j=i}^{\ell}\mu_{j-1}f(\mu_{j-1})}{\mu_i} + g(\mu_{i-1})\right).
\end{align*}

Since an update in $\D_i$ creates $\le c r(\mu_i)$, $c$ being some positive universal constant, updates in $G_i$, which is handled by the data structure in the next level, $\D_{i+1}$, we have to incorporate such quantity into the analysis.
Note that when rebuilding $\D_i$, all $\D_j, i < j$ are also rebuilt.
Hence no recourse is made for rebuilding.
We can now analyze the number of updates handled by $\D_i$ when 1 edge insertion happens in $G$.
By simple induction, we know there will be $\le \prod_{j=1}^{i-1}{cr(\mu_j)}$ updates in $G_{i-1}$.
Thus for the data structure in $i$-th level, $\D_i$, there will be at most $\prod_{j=1}^{i-1}{cr(\mu_j)}$ updates.
Combining these 2 quantities, we can bound the amortized update time of our data structure:
\begin{align*}
    T_u
    &=
    O\left(\sum_{i=1}^{\ell}\left(\frac{\sum_{j=i}^{\ell}\mu_{j-1}f(\mu_{j-1})}{\mu_i} + g(\mu_{i-1})\right)\prod_{j=1}^{i-1}{cr(\mu_j)}\right)\\
    &=
    O\left(\sum_{i=1}^{\ell}\left(\frac{\sum_{j=i}^{\ell}\mu_{j-1}f(\mu_{j-1})}{\mu_i} + g(\mu_{i-1})\right)\prod_{j=1}^{i-1}{r(\mu_j)}c^{i-1}\right).
\end{align*}

We next study the query time of our data-structure.
When answering a query, we add $s$ and $t$ to each layer of the terminal set.
When adding terminals to each layer of data structure $\D_i$, edge changes also propagate to lower levels.
As for the analysis of update time, we have to take the recourse into account.
The time can be bounded by $O(\ell T_u)$ where $T_u$ is the update time of our data structure.

Then we compute $\mathcal{P}(s, t, G_{\ell})$ in $G_{\ell}$, which has $\mu_\ell r(\mu_\ell)$ edges as guaranteed by the definition of IVS.
Since we have an $O(mh(n))$ algorithm for computing $\mathcal{P}(s, t)$ in an $m$-edge $n$-vertex graph, $\mathcal{P}(s, t, G_{\ell})$ can be computed in $\mu_\ell r(\mu_\ell)h(\mu_\ell)$ time.
Combining these 2 bounds, we can bound the query time by:
\begin{align*}
    T_q = O\left(\ell T_u + \mu_\ell r(\mu_\ell) h(\mu_\ell) \right).
\end{align*}

\begin{lemma} \label{lem: optimalTradeoff}
Let $\{\mu_i\}_{0 \leq i \leq \ell}$ be a family of parameters with $\mu_0 = m$.
Also, all parameters regarding efficient IVS, $f(n), g(n), r(n)$, and $h(n)$, are of order $n^{o(1)}$.
If we set \[ \mu_i = m^{1-i/(\ell+1)},~\text{ where }1 \leq i \leq \ell, \] then the update time is
\begin{equation} \label{eq: updateMinimized}
    T_u =
    O\left(\sum_{i=1}^{\ell}\left(\frac{\sum_{j=i}^{\ell}\mu_{j-1}f(\mu_{j-1})}{\mu_i} + g(\mu_{i-1})\right)\prod_{j=1}^{i-1}{r(\mu_j)}c^{i-1}\right) =
    O\left(\ell c^\ell m^{1 / (\ell+1) + o(\ell)}\right),
\end{equation}
and the query time is
\begin{equation} \label{eq: queryMinimized}
    T_q = O\left(\ell T_u + \mu_\ell r(\mu_\ell) h(\mu_\ell) \right) =
    O\left(\ell^2 c^\ell m^{1 / (\ell+1) + o(\ell)}\right).
\end{equation}
\end{lemma}
\begin{proof}
Plug in the choice of $\mu_i$, we have
\begin{align*}
    \frac{\sum_{j=i}^{\ell}\mu_{j-1}}{\mu_i}
    &=
    m^{i/(\ell+1) - 1}\left(\sum_{j=i}^{\ell}{1 - m^{(j-1)/(\ell+1)}}\right)
    = 
    \sum_{j=i}^{\ell}{m^{(i-j+1)/(\ell+1)}}
    \le 2m^{1/(\ell+1)}.
\end{align*}

Since $r(n), c \geq 1$ and $r(n)$ is an increasing function, $\prod_{j=1}^{i-1}{r(\mu_j)}c^{i-1} = O(r(n)^\ell c^\ell) = O(m^{o(\ell)} c^\ell)$ holds.
Combining these 2 inequalities, we can bound the update time by
\begin{align*}
    T_u &=
    O\left(\sum_{i=1}^{\ell}\left(\frac{\sum_{j=i}^{\ell}\mu_{j-1}f(\mu_{j-1})}{\mu_i} + g(\mu_{i-1})\right)\prod_{j=1}^{i-1}{r(\mu_j)}c^{i-1}\right)\\
    &=
    O\left(\sum_{i=1}^{\ell}\left(m^{1/(\ell+1)}f(n)+g(n)\right)r(n)^\ell c^\ell\right)\\
    &= 
    O\left(\ell c^\ell m^{1 / (\ell+1) + o(\ell)}\right).
\end{align*}

The bound for query time is straightforward from the definition of $\mu_\ell$.
\end{proof}

\section{Incremental All-Pairs Shortest Paths}
\label{sec:distance_inc}

In this section we show how to use Theorem~\ref{thm: metaTheorem} to design an online incremental algorithm for the approximate All-Pair Shortest-Paths problem with fast worst-case update and query time.
Concretely, we will show that the assumption (1) in Theorem~\ref{thm: metaTheorem} is satisfied with certain parameters for shortest paths.
Note that (2) follows immediately with $h(n)=\tilde{O}(1)$ by any $\tilde{O}(m)$ time single-pair shortest path algorithm.
We have the following theorem.
\begin{theorem}
\label{thm: IncrementalApproximateAPSP}
Let $G=(V,E)$ be an undirected, weighted graph and $c$ be some positive constant.
For every $\ell,r \geq 1$, there is an incremental deterministic All-Pair Shortest-Paths algorithm that maintains for every pair of nodes $u$ and $v$, an estimate $\delta(u,v)$ that approximates the shortest path distance between $u$ and $v$ up to a factor of $(2r-1)^{\ell}$ while supporting updates and queries in
$O(\ell c^{\ell - 1} r^{\ell} m^{1 / (\ell + 1)} n^{\ell/r} \log^{\ell-1} n)$ and
$O(\ell^2 c^{\ell - 1} r^{\ell} m^{1 / (\ell + 1)} n^{\ell/r} \log^{\ell-1} n)$
worst-case time, respectively.
\end{theorem}

\begin{corollary}
\label{coro: IncrementalApproximateAPSP}
Let $G=(V,E)$ be an undirected, weighted graph.
For every $r \geq 1$, there is an incremental deterministic All-Pair Shortest-Paths algorithm that maintains for every pair of nodes $u$ and $v$, an estimate $\delta(u,v)$ that approximates the shortest path distance between $u$ and $v$ up to a factor of $(2r-1)$ while supporting updates and queries in
$O(m^{1/2}n^{1/r})$ and
$O(m^{1/2}n^{1/r})$
worst-case time, respectively.
\end{corollary}
\begin{proof}
        The argument holds trivially by replacing $\ell$ by $1$ in Theorem~\ref{thm: IncrementalApproximateAPSP}.
\end{proof}

We start by introducing the usual definitions of sparsifiers and vertex sparsifiers for distances.
Let $G=(V,E)$ be an undirected, weighted graph with a \emph{terminal} set $T \subseteq V$.
For $u,v \in V$, let $\dist_G(u,v)$ denote the length of the shortest path between $u$ and $v$ in $G$.

\begin{definition}[Sparsifiers for Distances]
Let $G=(V,E)$ be an undirected, weighted graph, and let $\alpha \geq 1$ be a \emph{stretch} parameter.
A graph $H=(V',E')$ with $V \subseteq V'$ is an $\alpha$-\emph{distance sparsifier} of $G$ iff for all $u,v \in V$, $\dist_G(u,v) \leq \dist_{H}(u,v) \leq \alpha \cdot \dist_G(u,v)$.
\end{definition}

\begin{definition}[Vertex Sparsifiers for Distances]
\label{def: vertexDistanceSparsifier}
Let $G=(V,E)$ be an undirected, weighted graph with a \emph{terminal} set $T \subseteq V$, and let $\alpha \geq 1$ be a \emph{stretch} parameter.
A graph $H=(V',E')$ with $T \subseteq V'$ is an $\alpha$-\emph{(vertex) distance sparsifier} of $G$ with respect to $T$ iff for all $u,v \in T$,  $\dist_G(u,v) \leq \dist_{H}(u,v) \leq \alpha \cdot \dist_G(u,v)$.
\end{definition}

In the remainder of this section, we will show that the distance property in graphs admits an efficient incremental sparsifier oracle with desirable guarantees.

\begin{lemma}[Efficient ISO's for Distances]
\label{lem: efficientDistanceIVS}
Given an undirected, weighted graph $G=(V,E)$, and a parameter $r \geq 1$, there is a \emph{deterministic} algorithm that constructs an $(2r-1, O(n^{1/r} \log^2 n),O(n^{1/r} \log^{2} n),O(n^{1/r} \log^{2} n))$-efficient (distance) IVS of $G$.
\end{lemma}

We achieve this by showing a deterministic variant of the distance oracle due to Thorup and Zwick~\cite{ThorupZ05}. While our construction closely follows the ideas presented in the deterministic oracle due to Roddity, Thorup and Zwick~\cite{RodittyTZ05}, we note that their work only bounds the \emph{total} size of the oracle, which is not sufficient for our purposes. Similar ideas have been employed by Lacki et al.~\cite{LackiOPSZ15} for constructing bipartite emulators and by Abraham et al.~\cite{AbrahamCT14} for designing approximation algorithms for the fully-dynamic APSP problem.

We start by reviewing the randomized algorithm for APSP due to Thorup and Zwick~\cite{ThorupZ05}~(which is depicted in Figure~\ref{alg: randomizedPreprocessing}), and then derandomize that algorithm and show how it can be used to solve the above problem.

\begin{enumerate}
\itemsep0em 
\item  Set $A=V$ and $A_r = \emptyset$, and for $1 \leq i \leq r-1$ obtain $A_i$ by picking each node from $A_{i-1}$ independently, with probability $n^{-1/r}$. 

\item For each $1 \leq i < r$, and for each vertex $v \in V$, find the vertex $p_i(v) \in A_i$ (also known as the $i$-th \emph{pivot}) that minimizes the distance to $v$, i.e., \[ p_i(v) := \arg \min_{u \in A_i} \dist_G(u,v), \] and its corresponding distance value \[\dist_G(A_i,v) := \min\{\delta(w,v) \mid w \in A_i\} = \dist_G(v,p_i(v)). \]

\item For each vertex $v \in V$, define the \emph{bunch} $B(v) = \cup_{i=0}^{r-1} B_i(v)$, where
\begin{align*}
B_i(v) & := \{w \in A_i \setminus A_{i+1} \mid \dist_G(w,v) < \dist_G(A_{i+1},v) \}.
\end{align*}
\end{enumerate}

\begin{algorithm2e}[t!]
\caption{\textsc{HierarchyConstruct}$(G, r)$}
\label{alg: randomizedPreprocessing}
$A_0 \gets V$ ; $A_r \gets \emptyset$ \\
\For{$i \gets 1$ to $r-1$} 
{
        \label{algStep: random} $A_i \gets \textsc{Sample} \left( A_{i-1}, |V|^{-1/r} \right)$ \\
} 

\For{every $v \in V$}
{
        \For{$i \gets 0$ to $r-1$}
        {
           Let $\dist_G(A_i,v) \gets \min \{ \dist_G(w,v) \mid w \in A_i \}$ \\
           Let $p_i(v) \in A_i$ be such that $\dist_G(p_i(v),v) = \dist_G(A_i,v)$ \\
    }    
    \BlankLine
    $\dist_G(A_r,v) \gets \infty$ \\
    \BlankLine
    Let $B(v) \gets \cup_{i=0}^{r-1} \{w \in A_i \setminus A_{i+1} \mid \dist_G(w,v) < \dist_G(A_{i+1},v)\}$  
    
}

\end{algorithm2e}

Thorup and Zwick~\cite{ThorupZ05} showed that using the hierarchy of sets $(A_i)_{0 \leq i \leq r}$ chosen as above, the expected size of a bunch $\expec [|B(v)|]$ is $O(r n^{1/r})$, for each vertex $v \in V$. We note that the only place where their construction uses randomization is when building the hierarchy of sets~(the \textbf{for} loop in Step~\ref{algStep: random} in Figure~\ref{alg: randomizedPreprocessing}). Therefore, to derandomize their algorithm it suffices to design a deterministic algorithm that efficiently computes a hierarchy of set $(A_i)_{0 \leq i \leq r}$ such that $|B(v)| \leq \tilde{O}(r n^{1/r})$, for each $v \in V$ (note that compared to the randomized construction, we are content with additional poly-log factors on the size of the bunches).

We present a deterministic algorithm for computing the hierarchy of sets that closely follows the ideas presented in the deterministic construction of Roditty, Thorup, and Zwick~\cite{RodittyTZ05}. The main two ingredients of the algorithm are the \emph{hitting set} problem, and the \emph{source detection} problem. For the sake of completeness, we next review their definitions and properties.

\begin{definition}[Hitting set] \label{def: hittingSet} Let $U$ be a set of elements, and let $\mathbb{S} = \{S_1,\ldots,S_p\}$ be a collection of subsets of $U$. We say that $T$ is a \emph{hitting set} of $U$ with respect to $\mathbb{S}$ if $T \subseteq U$, and $T$ has a non-empty intersection with every set of $\mathbb{S}$, i.e., $T \cap S_i \neq \emptyset$ for every $1 \leq i \leq p$. 
\end{definition}

It is known that computing a hitting set of minimum size is an NP-hard problem. In our setting however, it is sufficient to compute approximate hitting sets. Since our goal is to design a deterministic algorithm, one way to deterministically compute such sets is using a variant of the well-known greedy approximation algorithm: (1) Form the set $T$ by repeatedly adding to $T$ elements of $U$ that `hit' as many `unhit' sets as possible, until only $|U|/s$ sets are unhit, where $|S_i| \geq s$ for each $1 \leq i \leq p$ ; (2) add an element from each one of the unhit sets to $T$. The lemma below shows that this algorithm finds a reasonably sized hitting set in time linear in the size of $U$ the collection $\mathbb{S}$. 


\begin{lemma} \label{lem: hittingSetAlgo}
Let $U$ be a set of size $u$ and let $\mathbb{S} = \{S_1,\ldots,S_p\}$ be the collection of subset of U, each of size at least $s$, where $s \leq p$. Then the above deterministic greedy algorithm runs in $O(u + ps)$ time and finds a hitting set $T$ of $U$ with respect to $\mathbb{S}$, whose size is bounded by $|T| = (u/s)(1+ \ln p)$.
\end{lemma}

Note that the size of this hitting set is within $O(\log n)$ of the optimum size since in the worst case $T$ has size at least $u /s$. 

\begin{definition}[Source Detection] \label{def: sourceDetect} Let $G=(V,E)$ be an undirected, weighted graph, let $U \subseteq V$ be an arbitrary set of sources of size $u$, and let $q$ be a parameter with $1 \leq q \leq u$. For every $v \in V$, we let $U(v,q,G)$ be the set of the $q$ vertices of $U$ that are closest to $v$ in $G$.
\end{definition}

Roditty, Thorup, and Zwick~\cite{RodittyTZ05} showed that the set $U(v,q,G)$ can be computed using $q$ single-source shortest path computations. We review their result in the lemma below.
\begin{lemma}[\cite{RodittyTZ05}] \label{lem: sourceDetectAlgo} For every $v \in V$, the set $U(v,q,G)$ can be computed in time $O(q m \log n)$.
\end{lemma}

Our algorithm for constructing the hierarchy of sets $(A_i)_{0 \leq i \leq r}$, depicted in Figure~\ref{alg: DeterministicHieararhcy}, is as follows. Initially, we set $A_0 = V$ and $A_r = \emptyset$. To construct the set $A_{i+1}$, given the set $A_i$, for $0 \leq i \leq r-2$, we first find the set $A_i(v,q,G)$, where $q=\tilde{O}(n^{1/r})$, using the source detection algorithm from Lemma~\ref{lem: sourceDetectAlgo}. Then we observe that the collection of sets $\{A_i(v,q,G)\}_{v \in V}$ can be viewed as an instance of the minimum hitting set problem over the set (universe) $A_i$, i.e., we want to find a set $A_{i+1} \subseteq A_i$ of minimum size such that each set $A_i(v,q,G)$ in the collection contains at least one node of $A_{i+1}$. We construct $A_{i+1}$ by invoking the deterministic greedy algorithm from Lemma~\ref{lem: hittingSetAlgo}, which produces a hitting set whose size is within $O(\log n)$ of the optimum one. We next prove the constructed hierarchy produces bunches whose sizes are comparable to the randomized construction, and also show that our deterministic construction can be implemented efficiently.

\begin{algorithm2e}[t!]
\label{alg: DeterministicHieararhcy}
\caption{\textsc{DetHierarhcy}$(G,r)$}
\KwData{Undirected, weighted graph $G=(V,E)$, parameter $r \geq 1$}
\KwResult{Hierarchy of sets $(A_i)_{0 \leq i \leq r}$}
\BlankLine
$q = \lceil n^{1/r} (1+ \ln n) \rceil $\\
$A_0 \gets V$ ; $A_r \gets \emptyset$ \\
\label{step: ForLoop} \For{$i \gets 0$ to $r-2$} 
{
        Compute $A_i(v,q,G)$ for each $v \in V$ using the source detection algorithm~(Lemma~\ref{lem: sourceDetectAlgo}) \\
        Let $\{A_i(v,q,G)\}_{v \in V}$ be the resulting collection of sets \\
        Compute a hitting set $A_{i+1} \subseteq A_{i}$ with respect to $\{A_i(v,q,G)\}_{v \in V}$~(Lemma~\ref{lem: hittingSetAlgo}) \\
}

\Return $(A_i)_{0 \leq i \leq r}$
\end{algorithm2e}

\begin{lemma} \label{lem: deterministicHierarhcy} Given an undirected, weighted graph $G=(V,E)$, and a parameter $r \geq 1$, Algorithm~\ref{alg: DeterministicHieararhcy} computes deterministically, in $O(rmn^{1/r} \log n)$ time, a hierarchy of sets $(A_i)_{0 \leq i \leq r}$ such that for each $v \in V$,
\[
        |B(v)| = O(rn^{1/r} \log n).
\]
\end{lemma}
\begin{proof}
We start by showing the bound on the size of the bunches. To this end, we first prove by induction on $i$ that $|A_{i}| \leq n^{1-i/r}$ for all $0 \leq i \leq r-1$. For the base case, i.e., $i = 0$, the claim is true by construction since $A_0 = V$. We assume that $|A_{i}| \leq n^{1-i/r}$ for the induction hypothesis, and show that $|A_{i+1}| \leq n^{1-(i+1)/r}$ for the induction step. Note that by construction each set in the collection $\{A_i(v,q,G)\}_{i \in V}$ has size $q = \lceil n^{1/r} (1+\ln n) \rceil \geq n^{1/r} (1+\ln n)$. Invoking the greedy algorithm from Lemma~\ref{lem: hittingSetAlgo}, we get a hitting set $A_{i+1} \subseteq A_i$ of size at most 
\[
  \left(\frac{|A_i|}{q}\right) (1+\ln n)  \leq \left(\frac{n^{1-i/r}}{n^{1/r} (1+\ln n)} \right) (1+\ln n) = n^{1-(i+1)/r}.
\]

We next show that for each $v \in V$ and for each $0 \leq i \leq r-1$, $|B_i(v)| \leq O(n^{1/p} \log n)$, which in turn implies the claimed bound on the size of vertex bunches. Note that it suffices to show that $B_i(v) \subseteq A_i(v,q,G)$ since then $|B_i(v)| \leq |A_i(v,q,G)| \leq n^{1/p} (1+\ln n) = O(n^{1/p} \log n)$. Recall that for $1 \leq i \leq r-1$
\[
        B_i(v) = \{w \in A_i \setminus A_{i+1} \mid \dist_G(w,v) < \dist_G(A_{i+1},v) \}
\]

Now, by construction of $A_{i+1}$ we have that $A_{i+1} \cap A_i(v,q,G) \neq \emptyset$, which implies that $B_i(v) \subseteq A_i(v,q,G)$ by the definition of $B_i(v)$.

We finally analyze the running time. For $0 \leq i \leq r-2$, consider the sequence of steps in the $i$-th iteration of the \textbf{for} loop in Figure~\ref{alg: DeterministicHieararhcy}. By Lemma~\ref{lem: sourceDetectAlgo}, the time to construct the collection of sets $\{A_i(v,q,G)\}_{v \in V}$ is $O(mn^{1/r} \log n)$. Furthermore, since the size of each set in this collection is at least $q=O(n^{1/r} \log n)$, Lemma~\ref{lem: hittingSetAlgo} guarantees that the greedy algorithm for computing a hitting set $A_{i+1}$ takes $O(n^{1+1/r} \log n)$ time. Combining the above bounds, we get that the total time for the $i$-th iteration is $O(mn^{1/r} \log n)$. Since there are at most $r$ iterations, we conclude that the running time of the algorithm is $O(rmn^{1/r} \log n)$.
\end{proof}


We next show to implement the two operations of the efficient (distance) ISO from Lemma~\ref{lem: efficientDistanceIVS}, namely $\textsc{Preprocess}()$ and $\textsc{AddTerminal}()$.

In the preprocessing phase, depicted in Figure~\ref{alg: PreprocessEmeregency}, given the graph $G$ and the stretch parameter $(2r-1)$,  we first invoke \textsc{HierarchyConstruct$(G,r)$} in Figure~\ref{alg: randomizedPreprocessing}, where Steps 1-3 are replaced by the deterministic algorithm for computing the hierarchy of sets \textsc{DetHierarchy$(G,r)$}.
Note that this modification ensures that our preprocessing algorithm is deterministic.
Next, for each vertex $v \in V$, we store its bunch $B(v)$ in a balanced binary search tree, where each vertex $w \in B(v)$ has as key the value $\dist_{G_0}(w,v)$
~(this step could have been implemented differently, but as we will shortly see, it is  useful in the subsequent applications of our algorithm),
where $G_0$ is the graph $G$ at the moment of preprocessing.

We next describe how to implement the \textsc{AddTerminal} operation, depicted in Figure~\ref{alg: TZAddTerminal}.
Let $T$ be the set of queried terminals. The main idea to construct a vertex distance sparsifier $H_T$ of $G$ with respect to $T$ is to exploit the bunches that we stored in the preprocessing step. More concretely, let $H_T$ be an initially empty graph. For each vertex $v \in T$, and every vertex in its bunch $u \in B(v)$, we add to $H_T$ the edge $(u,v)$ with weight $\dist_G(u,v)$. Finally, we return $H_T$ as a (vertex) distance sparsifier of $G$ with respect to $T$.


Next we formally analyzes the running time for the above operations and shows the correctness for the update operation.




\begin{proof}[Proof of Lemma~\ref{lem: efficientDistanceIVS}]
We first argue about the correctness.
To show that the resulting graph $H_T$ is indeed a vertex distance sparsifier with respect to $T$, we briefly review the update algorithm in the construction of Thorup and Zwick~\cite{ThorupZ05}, and show that this immediately applies to our graph setting.

\begin{algorithm2e}[t!]
\caption{\textsc{Preprocess}$(G, 2r-1)$}
\label{alg: PreprocessEmeregency}
Invoke \textsc{HierarchyConstruct$(G,r)$}, where instead of Steps 1-3 invoke \textsc{DetHierarchy$(G,r)$} \\
\For{each $v \in V$}
{ Store each $B(v)$, where $w \in B(v)$ holds $\dist_{G}(v,w)$.
}
Set $G_0 \gets G$. \\
\end{algorithm2e}

\begin{algorithm2e}[t!]
\caption{\textsc{AddTerminal}$(u)$}
\label{alg: TZAddTerminal}
Let $H_T$ be the vertex sparsifier maintained w.r.t. $T$.\\
Set $T \gets T \cup \{u\}$.\\
\For{every $v \in B(u)$}
{
    Add $(u,v)$ to $E(H_T)$ with weight $\dist_{G_0}(v,u)$.\\
}
\end{algorithm2e}


Let $u,v \in T$ by any two terminals. The algorithm uses the variables $w$ and $i$, and starts by setting $w \gets u$, and $i \gets 0$.
Then it repeatedly increments the value of $i$, swaps $u$ and $v$, and sets $w \gets p_i(u) \in B(u)$, until $w \in B(v)$.
Finally, it returns a distance estimate $\delta(u,v) = \dist_G(w,u) + \dist_G(w,v)$.
Observe that $w = p_i(u) \in B(u)$ for some $0 \leq i \leq r-1$ and $w \in B(v)$.
By construction of our vertex sparsifier $H_T$, note that the edges $(w,u)$ and $(w,v)$, and their corresponding weights, $\dist_G(w,u)$ and $\dist_G(w,v)$, are added to $H_T$.
Thus, there must exist a path between $u$ and $v$ in $H_T$ whose stretch is at most the stretch of the distance estimate $\delta(u,v)$.
Since in \cite{ThorupZ05} it was shown that for every $u, v \in T$,
\[ \dist_G(u,v) \leq \delta(u,v) \leq (2r-1) \dist_G(u,v), \]
we immediately get that
\[
        \dist_G(u,v) \leq \dist_{H_T}(u,v) \leq (2r-1) \dist_G(u,v).
\]

We finally analyze the running time for both operations.
First, note that by Lemma~\ref{lem: deterministicHierarhcy}, the deterministic algorithm for constructing the hierarchy of sets \textsc{DetHierarhcy$(G,r)$} runs in $O(rmn^{1/r}\log n)$ time.
Moreover, Thorup and Zwick~\cite{ThorupZ05} showed that given a hierarchy of sets, the bunches for all vertices in $G$ can be computed in $O(rmn^{1/r} \log n)$ time.
Combining these two bounds we get that the operation \textsc{Preprocess$(G,r)$} runs in $O(rmn^{1/r}\log n)$ time.
For the running time of \textsc{AddTerminal$(u)$}, recall that $H_T$ addes only edges between $u$ and its bunch, $B(u)$.
Since the size of a each individual vertex bunch is bounded by $O(rn^{1/r} \log n)$~(Lemma~\ref{lem: deterministicHierarhcy}), the time for adding edges adjacent to $u$ is bounded by $O(|B(u)|) = O(rn^{1/r} \log n)$.
This also bounds the recourse of adding $u$ to the terminal set.
\end{proof}

We finally prove the main result of this section, i.e., Theorem~\ref{thm: IncrementalApproximateAPSP}.
\begin{proof}[Proof of Theorem~\ref{thm: IncrementalApproximateAPSP}]
Let $G=(V,E)$ be a graph and consider an $(2r-1, O(rn^{1/r}\log{n}), O(rn^{1/r}\log^2{n}),O(rn^{1/r}\log{n}))$ -efficient (distance) IVS $H$ of $G$~(Lemma~\ref{lem: efficientDistanceIVS}).
Recall that given any pair of vertices $s,t$ in $G$, one can compute shortest path between $s$ and $t$ in $O(m \log n)$ time.
Thus, plugging the parameters $\alpha = 2r-1, f(n)=O(rn^{1/r}\log{n}),g(n)=O(rn^{1/r}\log^2{n}), r(n)=O(rn^{1/r}\log{n})$ and $h(n) = O(\log n)$ in Theorem~\ref{thm: metaTheorem} and  choosing the running time parameters as in Lemma~\ref{lem: optimalTradeoff}, we get an incremental algorithm such that for any pair of vertices $u$ and $v$ reports a query estimate $\delta(u,v)$ that approximates the shortest path distance between between $u$ and $v$ up to a $(2r-1)^{\ell}$ factor, and handles update operations in worst-case time of
\begin{align*}
    O\left(\ell c^{\ell - 1} r^{\ell} m^{1 / (\ell + 1)} n^{\ell/r} \log^{\ell-1} n \right)
\end{align*}
and query operations in worst-case time of
\begin{align*}
    O\left(\ell^2 c^{\ell - 1} r^{\ell} m^{1 / (\ell + 1)} n^{\ell/r} \log^{\ell-1} n \right).
\end{align*}
\end{proof}

\section{Offline Dynamic Algorithms via Vertex Sparsifiers}

\label{sec:meta_offline}
In this section, we show how to use efficient vertex sparsifier constructions to design \emph{offline} (approximate) dynamic algorithms for graph problems with certain properties while achieving fast amortized update and query time.
To achieve this we use a framework that has been exploited for solving offline $3$-connectivity~\cite{PengSS19}.
Our main contribution is to show that this generalizes to a much wider class of problems, leading to several interesting bounds which are not yet known in the \emph{online} dynamic graph literature.
Also, we show negative results on lower-bounds in dynamic problems.

We start by defining the model.
We are given an undirected graph $G=(V,E)$ and an \emph{offline} sequence of events or operations $x_1,\ldots,x_m$, where $x_i$ is ether an edge update~(insertion or deletion), or a query $q_i$ which asks about some graph property in $G$ at time $i$.
The goal is to process this sequence of updates in $G$ while spending total time proportional to $O(m f(m))$, where $f(m)$ is ideally some sub-linear function in $m$.

We next show that an analogue to Theorem~\ref{thm: metaTheorem} can also be obtained in the offline graph setting.
Our algorithm makes use of the notion of vertex sparsifiers as well as their useful properties including transitivity and decomposability. 
In our construction we want graph properties that admit (1) fast algorithms for computing vertex sparsifiers and (2) guarantee that the size of such sparsifers is reasonably small.
We formalize these requirements in the following definition. 

\begin{definition} \label{def: efficientVertexSparsifier}
Let $G=(V,E)$ be a graph, with a \emph{terminal} set $T \subseteq V$ and let $f(n), s(n) \geq 1$ be functions. We say that $(G',\alpha,f(n),s(n))$ is an $\alpha$-\emph{efficient vertex sparsifier} of $G$ with respect to $T$ iff  $G'$ is an $\alpha$-vertex sparsifier of $G$, the time to construct $G'$ is $O(m \cdot f(n))$, and the size of $G'$ is $O(|T| \cdot s(n))$.
\end{definition}

\begin{theorem} \label{thm: metaTheorem2}
Let $G=(V,E)$ be a graph, and for any $u,v \in V$, let $\mathcal{P}(u,v,G)$ be a solution to a minimization problem between $u$ and $v$ in $G$.
Let $f(n),s(n),h(n) \geq 1$ be functions, $\alpha, \ell \geq 1$ be parameters associated with the approximation factor, and let $\beta_0,\beta_1,\ldots, \beta_{\ell}$ with $\beta_0 = m$ be parameters associated with the running time.
Assume the following properties are satisfied
\begin{enumerate}
\itemsep0em 
\label{metaThM2: P1}\item $G$ admits an efficient vertex sparsifier $(G',\alpha,f(n),s(n))$,
\label{metaThM2: P2} \item $G'$ is transitive and decomposable,
\label{metaThM2: P3} \item The property $\mathcal{P}(u,v,G)$ can be computed in $O(m h(n))$ time in a graph with $m$ edges and $n$ vertices.
\end{enumerate}
Then there is an offline (approximate) dynamic algorithm that maintains for every pair of nodes $u$ and $v$, an estimate $\delta(u,v)$, such that
\begin{equation} \label{eq: approxMeta2}
        \mathcal{P}(u,v,G) \leq \delta(u,v) \leq \alpha^{\ell} \cdot \mathcal{P}(u,v,G).
\end{equation}
The total time for processing a sequence of $m$ operations is:
\begin{equation} \label{eq: runningTimeMeta2}   
        \tilde{O}\left (\beta_0\left( \sum_{j=1}^{\ell} \left(\frac{\beta_{j-1}}{\beta_j} \right)f(n) + \beta_{\ell} h(n)\right) s(n)  \right) \quad \text{where } \beta_0 = m.
\end{equation}
\end{theorem}

Before describing the underlying data-structure upon which the above theorem builds, we reduce the arbitrary sequence of operations into a more structured one, and also build a particular view for the problem.
These will allow us to greatly simplify the presentation. 

Concretely, first we may assume that each edge is inserted and deleted exactly once during the sequence of operations.
We achieve this by simply treating each edge instance as a new edge, i.e., we assume that each insertion of an edge $e=(u,v)$ inserts a new edge that is different from all previous instances of $(u,v)$. 

Second, since we are given the entire sequence of operations, for each edge $e$ we associate an interval $[i_e,d_e]$ which indicates the insertion and deletion time of $e$ in the operation sequence.
Furthermore, we denote by $q_t$ the time when query $q$ was asked in the operation sequence.
Let $[1,m]$ denote the interval covering the entire event sequence.
If we are interested in processing updates from a given interval $[r,s]$, we will define graphs that consists of two types of edges with respect to this interval: 
\begin{enumerate} \itemsep0em 
\item \emph{non-permanent edges}, which are edges affected by an event in this interval, i.e., $E^{p}_{[r,s]} = \{e \mid i_e \text{ or } d_e \in [r,s]\}, $
\item \emph{permanant edges}, which are edges present throughout the entire interval, i.e., $E^{np}_{[r,s]} = \{e \mid i_e < r \leq s < d_e\}.$
\end{enumerate} 
Additionally, it will be useful to define the queried vertex pairs within the interval $[r,s]$: $Q_{[r,s]} = \{q \mid q_t \in [r,s]\}$.


\paragraph*{Data Structure.}
We now describe a generic tree data-structure $\mathcal{T}$, which allows us to unify our framework and thus greatly simplify the presentation.
This tree structure is obtained by hierarchically partitioning the operation sequence into smaller disjoint intervals.
These intervals induce graphs that are suitable for applying vertex sparsifiers, which in turn allow us to process updates in a fast way, while paying some error in the accuracy of the query operations. 

Consider some integer parameter $\ell \geq 1$ and parameters $\beta_0, \beta_1, \ldots, \beta_{\ell}$ with $\beta_0 = m$.
The tree $\mathcal{T}$ has $\ell+1$ levels, where each level $i$ is associated with the parameter $\beta_i$, $i=0,\ldots,\ell$.
Each node of the tree stores some interval from the event sequence.
Formally, our decomposition tree $T$ satisfies the following properties:
\begin{enumerate}
\itemsep0em 
\item The root of the tree stores the interval $[1,m]$.
\item The intervals stored at nodes of same level are disjoint.
\item Each interval $[r,s]$ stored at a node in $\mathcal{T}$ is associated with 
    \begin{itemize}
    \itemsep0em 
    \item a graph $G_{[r,s]} = \left(V,E^{p}_{[r,s]}\right)$,
    \item a graph of \emph{new permanent edges} $H_{[r,s]} = G_{[r,s]} \setminus G_{[q,t]}$, where $G_{q,t}$ is the parent of $G_{[r,s]}$ in $\mathcal{T}$ (if any). 
    \item a set of \emph{boundary vertices} $\partial_{[r,s]} = V(E^{np}_{[r,s]}) \cup V(Q_{[r,s]})$.
    \end{itemize}
\item If $[r,s] \subseteq [q,t]$ then it holds that (a) $\partial_{[r,s]} \subseteq \partial_{[q,t]}$, and (b) $E^{p}_{[q,t]} \subseteq E^{p}_{[r,s]}$.
\item The length of the interval stored at a node at level $i$ is $\beta_i$.
\item A node at level $i$ has $\beta_{i}/\beta_{i+1}$ children.
\item The number of nodes at level $i$ is at most $O(\beta_0/\beta_i)$.
\end{enumerate}

The lemma below shows that a decomposition tree can be constructed in time proportional to the length of the operation sequence times the height of the tree.

\begin{lemma} \label{lemm: decomposTree}
Let $G=(V,E)$ be a dynamic graph where the sequence of operations is revealed upfront.
Then there is an algorithm that computes the decomposition tree $\mathcal{T}$ in $O(\ell m)$ time, where $m$ denotes the length of the operation sequence and $\ell$ is the height of the tree.
\end{lemma}
\begin{proof}
Let $\mathcal{T}$ be a tree with a single node (corresponding to its root) that stores the interval $[1,m]$.
We augment $\mathcal{T}$ in the following natural way:
(a) We partition the interval $[1,m]$ into $\beta_0/\beta_1 = m/\beta_1$ disjoint intervals, each of length $\beta_1$.
(b) For each of these intervals we create a node in the tree $\mathcal{T}$, and connect each node with the root of $\mathcal{T}$, i.e., those nodes form the children of the root, and thus the nodes at level $1$ of $\mathcal{T}$. 
(c) We recursively apply steps (a) and (b) to the newly generated nodes until we reach the $(\ell+1)$-st level of the tree. 

By the construction above, it easily follows that the generated tree $\mathcal{T}$ satisfies properties (1), (2), (4), (5), (6) and (7).
Thus, it remains to show how to compute the quantities in (3).
This can be achieved by (a) computing the intervals $[i_e,d_e]$, for every edge $e$ in the sequence (note that this is possible because we assumed that every edge is inserted and deleted exactly once within the interval $[1,m]$), and (2) for each node in the tree, computing the sets $E^{np}_{[r,s]}$ and $E^{p}_{[r,s]}$. 

For the running time, observe that computing the intervals $[i_e,d_e]$ takes $O(m)$ time.
Having computed these intervals, we can level-wise compute the permanent and non-permanent edges for each node in that particular level.
By disjointedness of the intervals, the amount of work we perform per level is $O(m)$.
Since there are most $O(\ell)$ levels, it follows that the running time for constructing the decomposition tree is $O(\ell m)$.
\end{proof}
 
\paragraph*{Computing vertex sparsifiers in the hierarhcy.}
We next show how to efficiently compute a vertex sparsifier $G'_{[r,s]}$ for each node $G_{[r,s]}$ from the decomposition tree $\mathcal{T}$.
The main idea behind this algorithm is to leverage the sparsifier computed at the parent nodes as well as apply the efficient vertex sparsifiers from Theorem~\ref{thm: metaTheorem2}~Part~1.
The procedure accomplishing this task for a single node of the tree $\mathcal{T}$ is formally given in Algorithm~\ref{alg: vertexNodeSparsiy}.
To compute the vertex sparsifier for every node, we simply apply it in a top-down fashion to the nodes of $\mathcal{T}$. 

\begin{algorithm2e}[h]
\caption{\textsc{VertexSparsify}$(G_{[r,s]})$}
\label{alg: vertexNodeSparsiy}
\If{$G_{[r,s]}$ is the root node}
{
        $G''_{[r,s]} = G'_{[r,s]} \gets (V,\emptyset)$, i,e, the empty graph.
}
\Else
{       Let $G_{[q,t]}$ be the parent of $G_{[r,s]}$ in $\mathcal{T}$ \\
        $G''_{[r,s]} \gets \left(G'_{[q,t]} \cup  H_{[r,s]} \right)$, where $G_{[q,t]}'$ is an efficient vertex sparsifier of $G_{[q,t]}$ with respect to $\partial_{[q,t]}$\\
        Let $G'_{[r,s]}$ be an $\alpha$-efficient vertex sparsifier of $G''_{[r,s]}$ with respect to $\partial_{[r,s]}$~(Theorem~\ref{thm: metaTheorem2}~Part~1)
}
\Return $G'_{[r,s]}$
\end{algorithm2e}

To argue about the usefulness of Algorithm~\ref{alg: vertexNodeSparsiy}, we need to bound the quality of sparsifiers produced at the nodes of $\mathcal{T}$.
The lemma below show that the quality grows multiplicatively with the number of levels in $\mathcal{T}$.

\begin{lemma} \label{lem: qualitySparsifier}
Let $G_{[r,s]}$ be a node of $\mathcal{T}$ at level $i \geq 0$ .
Then $G'=\textsc{VertexSparsify}(G_{[r,s]})$ outputs an $\alpha^{i}$-efficient vertex sparsifier of $G_{[r,s]}$ with respect to $\partial_{[r,s]}$.
\end{lemma}
\begin{proof}
We proceed by induction on $i$.
For the base case, i.e., $i = 0$, $G_{[1,m]}$ is the root node.
Since $E^{p}_{[1,m]} = \emptyset$ by definition of permanent edges, we get that $G'_{[1,m]} = G_{[1,m]}$, i.e., $G_{[1,m]}$ is a sparsifier of itself. 

Let $G_{[r,s]}$ be a node at level $i > 0$.
Let $G_{[q,t]}$ be the parent of $G_{[r,s]}$ in $\mathcal{T}$, and let $G'_{[q,t]}$ be its cut sparsifier at level $(i-1)$, as defined in Algorithm~\ref{alg: vertexNodeSparsiy}.
By Property (4) of $\mathcal{T}$ note that $E^{p}_{[q,t]} \subseteq E^{p}_{[r,s]}$ since $[r,s] \subseteq [q,t]$.
Also recall that $H_{[r,s]} = G_{[r,s]} \setminus G_{[q,t]}$.
By induction hypothesis, we know that $G'_{[q,t]}$ is an $\alpha^{i-1}$-efficient vertex sparsifier of $G_{[q,t]}$ with respect to $\partial_{[q,t]}$.
This together with the decomposability property in Theorem~\ref{thm: metaTheorem2}~Part~2 imply that that $G''_{r,s} = G'_{[q,t]} \cup (G_{[r,s]} \setminus G_{[q,t]})$ is an $\alpha^{i-1}$-efficient vertex sparsifier of $G_{[q,t]} \cup (G_{[r,s]} \setminus G_{[q,t]}) = G_{[r,s]}$ with respect to $\partial_{[q,t]}$.
Now, by Theorem~\ref{thm: metaTheorem2}~Part~1 we get that $G'_{[r,s]}$ is an $\alpha$-efficient vertex sparsifier of $G''_{[r,s]}$ with respect to $\partial_{[r,s]}$.
Since $\partial_{[r,s]} \subseteq \partial_{[q,t]}$, and applying the transitivity property~(Theorem~\ref{thm: metaTheorem2}~Part~2) on $G'_{[r,s]}$ and $G''_{[r,s]}$, we get that $G'_{[r,s]}$ is an $\alpha^{i-1+1} = \alpha^i$-efficient vertex sparsifier of $G_{[r,s]}$.
\end{proof}

We now state a crucial property of the nodes in the decomposition tree $\mathcal{T}$, which allows us to get a reasonable bound on the running time for computing vertex sprasifiers for the nodes in $\mathcal{T}$.

\begin{lemma} \label{lem: boundNewPermanentedges}
Let $G_{[r,s]}$ be a node in the decomposition tree $\mathcal{T}$, and let $G_{[q,t]}$ be its parent.
Then we have that the number of new permanent edges of $G_{[r,s]}$ is bounded by the number of non-permanent edges of its parent, i.e., $|E\left(H_{[r,s]}\right)| \leq |E^{np}_{[q,t]}|$.
\end{lemma}
\begin{proof}
If an edges in in $H_{[r,s]}$, then it is not in $G^{p}_{[r,s]}$, thus it is a non-permanent edge in $G_{[q,t]}$.
\end{proof}

The lemma below gives a bound on the running time for computing vertex sparsifers in $\mathcal{T}$.

\begin{lemma} \label{lem: runningTimeForSparsifiers}

The total running time for computing the vertex sparsifiers for each node in the decomposition tree $T$ of height $\ell$ is bounded by 
\[
        \tilde{O} \left(\beta_0 \cdot \left(\sum_{j=1}^{\ell} \frac{\beta_{j-1}}{\beta_j}f(n)s(n)\right)\right), \quad \text{where }\beta_0 = m.
\]
\end{lemma}
\begin{proof}
For $i \geq 1$, let $Y(i)$ be the total time for computing the vertex sparsifiers for all the nodes in $\mathcal{T}$ up to (and including) level $i$.
Furthermore, let $Z(i)$ be the total time for computing the vertex sparsifier of the nodes at level $i$ in $Y$~(and excluding other levels).
We will show by induction on the number of levels $i$ that $Y(i) = O\left(\beta_0 \cdot \left(\sum_{j=1}^{i} \frac{\beta_{j-1}}{\beta{j}}\right)f(n)s(n)\right)$, which with $i=k$ implies the claim we want to prove.

For the base case, i.e., $i=1$, consider any node $G_{[r,s]}$ at level $1$ of $\mathcal{T}$.
By construction of $\mathcal{T}$, $G_{[r,s]}$ contains at most $O(\beta_0)$ permanent edges.
Furthermore, note that the parent of $G_{[r,s]}$ is the root node $G_{[1,m]}$, for which $G'_{[1,m]} = (V,\emptyset)$.
Thus, by Theorem~\ref{thm: metaTheorem2}~Part~1 we get that the time to compute an efficient vertex sparsifier per node is $O(\beta_0 \cdot f(n))$.
By Property (7) of $\mathcal{T}$, the number of nodes at level $1$ is $O(\beta_0/\beta_1)$, implying that the total running time is $Y(1) = Z(1) = O\left(\beta_0 \left( \frac{\beta_0}{\beta_1} \right) f(n)\right) = O\left(\beta_0 \left( \frac{\beta_0}{\beta_1} \right) f(n)g(n)\right)$, as desired. 

We next show the inductive step.
Let $G_{[r,s]}$ be a node at level $i > 1$, and let $G_{[q,t]}$ be its parent.
We want to bound the size of the intermediate graph $G''_{[r,s]} = (G'_{[q,t]} \cup H_{[r,s]})$, as defined in Algorithm~\ref{alg: vertexNodeSparsiy}, which in turn determines the running time for computing an efficient vertex sparsifier of $G_{[r,s]}$.
To this end, first observe that Theorem~\ref{thm: metaTheorem2}~Part~1 implies that the size of sparsifier $G'_{[q,t]}$ of $G_{[q,t]}$ is bounded by 
\[ O(|\partial_{[q,t]}| \cdot s(n)) \leq |V(E^{np}_{[r,s]}) \cup V(Q_{[r,s]})| \cdot s(n) \leq O(\beta_{i-1} \cdot s(n)),\]
since the number of non-permanent edges and queries is proportional to the length of the interval being considered.
Second, by Lemma~\ref{lem: boundNewPermanentedges}, we also have that $|E(H_{[r,s]})| \leq |E^{np}_{q,t}| \leq O(\beta_{i-1})$, thus giving that $|G''_{[r,s]}| \leq O(\beta_{i-1} \cdot s(n))$.
As Algorithm~\ref{alg: vertexNodeSparsiy} runs \textsc{CutSparsify} on the graph $G''_{[r,s]}$, Theorem~\ref{thm: metaTheorem2}~Part~1 gives that the running time to compute an efficient vertex sparsifier for the node $G_{[r,s]}$ is $O(\beta_{i-1} \cdot f(n) s(n))$, and that its size is $O(\beta_{i-1} \cdot s(n))$.
Combining this together with the fact that the number of nodes at level $i$ is at most $O(\beta_0/\beta_i)$ (Property (7) of $\mathcal{T}$) imply that \[Z(i) = O\left(\beta_0 \cdot \frac{\beta_{i-1}}{\beta_i} f(n) s(n)\right). \] 

To complete the inductive step, note that by induction hypothesis, \[Y(i-1) = O\left(\beta_0 \cdot \left(\sum_{j=1}^{i-1} \frac{\beta_j-1}{\beta_j}\right)f(n)s(n)\right).\]
Summing over this and the bound on $Z(i)$ we get
\begin{align*}
        Y(i) & = Y(i-1) + Z(i) \\
        &  = O\left(\beta_0 \cdot \left(\sum_{j=1}^{i-1} \frac{\beta_{j-1}}{\beta_j}\right) f(n)s(n) \right) + O\left(\beta_0 \cdot \left(\frac{\beta_{i-1}}{\beta_i}\right)f(n)s(n)\right) \\
         & = O\left(\beta_0 \cdot \left(\sum_{j=1}^{i} \frac{\beta_{j-1}}{\beta{j}}f(n)s(n)\right)\right).
\end{align*}
\end{proof}

\paragraph*{Processing operations in the hierarchy.}
So far we have shown how to reduce the sequence of operations into smaller intervals in a hierarchical manner, while (approximately) preserving the properties of the edges and queries involved in the offline sequence.
In what follows, we observe that for processing these events, it is sufficient to process the nodes (and their corresponding intervals) stored at the last level $\ell$ of the tree decomposition $\mathcal{T}$ (note that this is possible because intervals at level $\ell$ form a partitioning of the event sequence $[1,m]$, and all vertex pairs within intervals that will be involved in edge updates or queries are preserved using vertex sparsifiers).



The algorithm for processing the updates is quite simple: for every node $G_{[r,s]}$ at level $\ell$ of $\mathcal{T}$, we process all operations in the interval consecutively: for each edge insertion or deletion we add or remove that suitable edges to $G'_{[r,s]}$, and for each query $(x,y)$ we run on the vertex sparsifier $G'_{[r,s]}$ the static algorithm from Theorem~\ref{thm: metaTheorem2}~Part~3 to calculate the property $\mathcal{P}(x,y,G'_{[r,s]})$ between $x$ and $y$ in $G'_{[r,s]}$.
(note that this is possible since $\partial_{[r,s]} \supseteq \{x,y\}$ by construction of $\mathcal{T}$). 

We next analyze the total time for processing the sequence of events in the last level of $\mathcal{T}$.

\begin{lemma} \label{lem: procesingUpdates}
The total time for processing the whole sequence of operations at level $\ell$ of the decomposition tree $\mathcal{T}$ is $\tilde{O}(\beta_0 \beta_\ell \cdot s(n) h(n))$, where $\beta_0 = m$.
\end{lemma}
\begin{proof}
As in the worst-case there can be at most $O(\beta_\ell)$ queries within the interval, and since the size of $G'_{[r,s]}$ is also bounded by $O(\beta_\ell s(n))$, by Theorem~\ref{thm: metaTheorem2}~Part~3 it follows that answering all the queries and processing the non-permanent edges within a single interval at level $\ell$ is bounded by $\tilde{O}(\beta_\ell^{2} s(n) h(n))$.
Combining this with the fact that the number of nodes at level $\ell$ is $O(\beta_0/\beta_\ell)$ (Property (7) of $\mathcal{T}$), we get that the total cost for processing the queries is $\tilde{O}(\beta_0 \beta_\ell \cdot s(n) h(n))$.
\end{proof}

Combining Lemma~\ref{lem: runningTimeForSparsifiers} and Lemma~\ref{lem: procesingUpdates} leads to an overall performance of
\[
        \tilde{O}\left (\beta_0\left( \sum_{j=1}^{\ell} \left(\frac{\beta_{j-1}}{\beta_j} \right)f(n) + \beta_{\ell} h(n)\right) s(n)  \right) \quad \text{where } \beta_0 = m,
\]
which proves the claimed total update time in Theorem~\ref{thm: metaTheorem2}.

We finally prove the correctness of our algorithm.
Concretely, we show that the estimate we return when processing any query $(x,y)$ in the last level of the hierarchy approximates the property $\mathcal{P}$ of the graph $G$ up to an $\alpha^{\ell}$ factor, thus proving the claimed estimate in Theorem~\ref{thm: metaTheorem2}.

To this end, let $q_i$ be a query in the sequence of operations $[1,m]$.
Since the intervals at level $\ell$ of $\mathcal{T}$ form a partitioning of $[1,m]$, there must exist an interval $[r,s]$ that contains the query $q_i$.
Let $(x,y)$ be the queried vertex pair of $q_i$.
By Lemma~\ref{lem: qualitySparsifier}, we get that the graph $G'_{[r,s]}$ at level $\ell$ is an $\alpha^{\ell}$-vertex sparsifier of $G_{[r,s]}$ with respect to $\partial_{[r,s]}$.
Since by construction $\partial_{[r,s]} \supseteq \{x,y\}$, we get that the $G'_{[r,s]}$ approximates the property $\mathcal{P}(x,y,G)$ of $G_{[r,s]}$ up to an $\alpha^{\ell}$ factor.
Finally, recall that we run the algorithm from Theorem~\ref{thm: metaTheorem2}~Part~3 on $G'_{[r,s]}$, thus worsening the approximation in the worst-case by at most a constant factor, which yields the claimed bound.



\subsection{Applications to Offline Shortest Paths and Max-Flow}
In this sub-section, we show how to use our general Theorem~\ref{thm: metaTheorem2} to design offline dynamic algorithms for the approximate All Pair Shortest Paths and All Pair Max-Flow with reasonably small total update time. 

We first consider shortest paths.
Recall that our goal is to show that assumptions (1), (2) and (3) from Theorem~\ref{thm: metaTheorem2} are satisfied with certain parameters for the shortest path measure.
For (1) we make the following observation: given a graph $G$, a subset of terminals $T$, and a parameter $r \geq 1$, we can construct an \emph{efficient} (vertex) distance sparsifier $(H_T,(2r-1),\tilde{O}(n^{1/r}),\tilde{O}(n^{1/r}))$ by simply constructing an efficient incremental vertex sparsifier for $G$ using Lemma~\ref{lem: efficientDistanceIVS} and add $T$ to its terminal set.
Also note that assumption (2) is satisfied by the transitivity and decomposability of $H$, and finally recall that (3) follows by any $\tilde{O}(m)$ time single pair shortest path algorithm.
These together imply the following result.

\begin{theorem}
\label{thm: offlineShortestPaths}
Let $G=(V,E)$ be an undirected, weighted graph. For every $r,\ell \geq 1$, there is an \emph{offline} fully dynamic approximate \emph{All Pair Shortest Path} algorithm that maintains for every pair of nodes $u$ and $v$, a distance estimate $\delta(u,v)$ such that
\[
        \dist_G(u,v) \leq \delta(u,v) \leq (2r-1)^{\ell} \dist_G(u,v).
\]
The total time for processing a sequence of $m$ operations is
\[
        \tilde{O}(m \cdot m^{1/(\ell+1)}n^{2/r}).
\]
\end{theorem}

We now proceed with max flow.
Following essentially the same idea as with shortest paths, we need to show that assumptions (1), (2) and (3) from Theorem~\ref{thm: metaTheorem2} are satisfied with certain parameters for the max flow measure.
For (1) we show the existence of \emph{efficient} (vertex) flow sparsifier $(H_T,O(\log^{4} n),\tilde{O}(1),\tilde{O}(1))$ via the following lemma.

\begin{theorem}[\cite{RackeST14,Peng16}]
\label{thm: fastFlowSparsifier}
Given an undirected, weighted graph $G=(V,E)$, there is an $\tilde{O}(m)$ time randomized algorithm $\textsc{FlowSparsify}(G)$ that with high probability computes a tree-based flow sparsifier $H=(V',E')$ with $V \subseteq V'$ satisfying the following properties
\begin{enumerate}
\itemsep0em 
\item $H$ is a bounded degree rooted tree,
\item $H$ has quality $O(\log^{4} n)$,
\item There is a one-to-one correspondence between the leaf nodes of $H$ and the nodes in $G$,
\item The height of $H$ is at most $O(\log^2 n)$.
\end{enumerate} 
\end{theorem}
\begin{proof}
The original construction of R\"acke et al.~\cite{RackeST14} produces a rooted tree $H'$ which satisifes the above properties, except that $H'$ has unbounded degree and the height of the tree is $O(\log n)$. Since we will exploit the bounded degree assumption in the subsequent applications of our data-structure, here we present a standard reduction from $H'$ to a bounded degree tree $H$ at the cost of increasing the height of the tree by a logarithmic factor. 

Let $H'$ be the rooted tree we described above. Let $u \in H'$ be an internal node of degree larger than $2$ and let $C(u)$ be its children. We start by removing all edges between $u$ and its children $C(u)$ from $H'$, and record all their corresponding edge weights. Next, we create a binary rooted tree $\tilde{H}$ where the children $C(u)$ are the leaf nodes, i.e., $L(\tilde{H}) = C(u)$, and $u$ is the root of $\tilde{H}$. To complete the construction of $\tilde{H}$ we need to define its edge weights. To this end, for any subtree $R \subseteq \tilde{H}$, let $E(L(R))$ denote the set of edges incident to leaf nodes in $R$. We distinguish the following two cases. (1) If $e=(x,y) \in E(L(\tilde{H}))$, i.e., $e$ is an edge incident to a leaf of $\tilde{H}$, and $x \in L(\tilde{H}) = C(u)$, we set $\ww_{\tilde{H}}(x,y) = \ww_{H'}(x,u)$. (2) If $e=(x,y) \not \in E(L(\tilde{H}))$, then let $\tilde{H}_x$ and $\tilde{H}_y$ be the trees obtained after deleting the edge $e$ from $\tilde{H}$. Furthermore, for any subtree $R \subseteq \tilde{H}$ define
\[ \ww(R) := \sum_{e\in E(L(R))} \ww_{\tilde{H}}(e).\]

Finally, for $e=(x,y) \not \in E(L(\tilde{H}))$ and $e \in \tilde{H}$ we set
\[
        \ww_{\tilde{H}}(x,y) = \min\{\ww(\tilde{H}_x), \ww(\tilde{H}_y)\}.
\]
Note that the weight-sums $\ww(\tilde{H}_x)$ and $\ww(\tilde{H}_y)$ can be calculated since we first defined the weights for edges in $E(L(\tilde{H}))$. Also observe that $H'$ remains a tree because we simply removed children of $u$~(which could be viewed as a star) and replaced this by another bounded degree tree $\tilde{H}$. We repeat the above process for every internal node of $H'$ until $H'$ becomes a bounded degree rooted tree, and denote by $H$ the final resulting tree.

We claim that $H$ has depth at most $O(\log^2 {n})$. To see this, recall that the initial height of $H'$ was $O(\log n)$, and every replacement of the star centered at a non-terminal with a bounded degree tree increases the height by an additive of $O(\log n)$. 
Summing over $O(\log n)$ levels gives the claimed bound.

Finally, it is easy to see that $H$ is flow sparsifier of quality $1$ for $H'$ with respect to all leaf nodes of $H'$, which in turn correspond to the nodes of graph $G$. Thus, $H$ is also a flow sparsifier for $G$ with quality $O(\log^{4} n)$.
\end{proof}

We next show how to construct a vertex sparsifier w.r.t. a given terminal set $T$.

The construction involves 2 phases: (1) preprocessing, and (2) constructing vertex sparsifier.
In the preprocessing phase, given a graph $G$, we simply invoke \textsc{FlowSparsify$(G)$} from Theorem~\ref{thm: fastFlowSparsifier} and let $H$ be the resulting tree-based sparsifier.
The main idea for constructing a (vertex) flow sparsifier $H_T$ of $G$ with respect to $T$ is to exploit the fact that $H$ is a tree.
Concretely, let $H_T$ be an initially empty graph.
For $v \in T$, let $H[v, r]$ be the path in $H$ from $v$ to the root $r$ of $H$~(since $v \in T \subseteq V$, recall that $v$ is a leaf node of $H$ by Property~(3) in Theorem~\ref{thm: fastFlowSparsifier}).
For each $v \in T$, and every edge $e \in H[v, r]$, we add $e$ with weight $w_H(e)$ to $H_T$.
Finally, we return $H_T$ as a (vertex) flow sparsifier of $G$ with respect to $T$. 
The following lemma analyzes the running time for the above procedure and shows the correctness.

\begin{lemma} \label{lem: staticFlowVertexSparsifier}
Given an undirected, weighted graph $G=(V,E)$, and a subset of vertices $T$,
there is an algorithm that produce a (vertex) flow sparsifier $H_T$ w.r.t. $T$ in time $\Otil(m)$.
$H_T$ has size $O(|T| \log^{2} n)$ and quality $O(\log^{4} n)$.

In other words, there is an efficient (vertex) flow sparsifier $(H_T, O(\log^{4} n), \Otil(1), \Otil(1))$.
\end{lemma}

\begin{proof}[Proof of Lemma~\ref{lem: staticFlowVertexSparsifier}]
We first argue about the correctness of the output $H_T$.
First, we show that $H_T$ is $1-$(vertex) flow sparsifier of $H$ with respect to $T$.
To see this, note that since $H$ is a tree, every (multi-commodity) flow among any two leaf vertices $(u,v)$ is routed according to the unique shortest path between between $u$ and $v$ in $H$, denoted by $H[u,v]$.
Since $H_T$ is formed taking the union of the paths $H[v,r]$, for each $v \in T$, and $H[u,v] \subseteq \left(H[v,r] \cup H[u, r]\right)$, it follows that $H[u, v]$ is also contained in $H_T$.
Thus every flow  that we can route in $H$ among any two pairs in $T$, we can feasible route in $H_T$.
For the other direction, observe that by construction $H_T \subseteq H$.
Therefore, any flow among any two pairs in $T$ that can be feasibly routed in $H_T$, can also be routed in $H$~(this follows since $H$ has more edges than $H_T$, and thus the routing in $H$ has more flexibility).
Combining the above we get that $H_T$ is a quality $1$-(vertex) flow sparsifier of $H$.
Since $H$ is flow sparsifier of $G$ with quality $O(\log^{4} n)$~(Property~(2) in Theorem~\ref{thm: fastFlowSparsifier}) and $T \subseteq V$, applying transitivity on $H_T$ and $H$
we get that $H_T$ is a quality $O(\log^{4} n)$ (vertex) flow sparsifier of $G$ with respect to $T$.

We finally analyze the running time for both operations.
Recall that the preprocessing phase is implemented by simply invoking $\textsc{FlowSparsifiy(G)}$.
By Theorem~\ref{thm: fastFlowSparsifier}, we know that the latter can be implemented in $\tilde{O}(m)$, which in turn bounds the running time of our preprocessing phase.
For the running time of constructing vertex sparsifier, recall that $H_T$ consists of the union over the paths $P(v,r,H)$, for each $v \in T$.
Since the length of each such path is bounded by $O(\log^{2} n)$~(Property~$(4)$ in Theorem~\ref{thm: fastFlowSparsifier}), we get that the size of $H_T$ is bounded by $O(|T| \log^2 n)$.
Note that after having access to any leaf vertex $v$, the path $H[v,r]$ can be retrieved from $H$ in time proportional to its length.
This implies that the time to output $H_T$ is also bounded by $O(|T| \log^2 n)$.
\end{proof}


Also note that assumption (2) is satisfied by the transitivity and decomposability of $H$, and finally recall that (3) follows by employing the $\tilde{O}(m)$ time (approximate) $(s,t)$-maximum flow algorithm due to Peng~\cite{Peng16}. These together imply the following theorem.

\begin{theorem}
\label{thm: offlineMaxFlow}
Let $G=(V,E)$ be an undirected, weighted graph. For every $\ell \geq 1$, there is an \emph{offline} fully dynamic approximate \emph{All Pairs Max Flow} algorithm that maintains for every pair of nodes $u$ and $v$, a flow estimate $\delta(u,v)$ such that
\[
        \frac{1}{\tilde{O}(\log^{4\ell} n)} \maxflow_G(u,v) \leq \delta(u,v) \leq \maxflow_G(u,v).
\]
The total time for processing a sequence of $m$ operations is
\[
        \tilde{O}(m \cdot m^{1/(\ell+1)}).
\]
\end{theorem}

\subsection{Implications on Hardness of Approximate Dynamic Problems}
\label{sec:hardness}

\subsubsection{Approximate max flow and cut sparsifiers}

Assuming the OMv conjecture, Dahlgaard \cite{Dahlgaard16} show that any incremental exact max flow algorithm on undirected graphs must have amortized update time at least $\Omega(n^{1-o(1)})$.
However, the hardness of approximation is not known
\footnote{However, on directed graphs, the hardness of approximation is known.
	This is because even dynamic reachability is hard under several conjectures\cite{AbboudW14,HenzingerKNS15}.}: 
\begin{proposition}
	There is no polynomial lower bound for dynamic $\omega(\mbox{polylog}(n))$-approximate max flow in the offline setting.
\end{proposition}
This follows directly from \Cref{thm: offlineMaxFlow}.
Thus the important open problem is whether we can prove a hardness for dynamic $(1+\epsilon)$-approximate max flow algorithms on undirected graphs for a constant $\epsilon>0$. 

On the other hand, it is not known whether, given a set of $k$ terminals, there is a $(1+\epsilon)$-approximate cut (vertex) sparsifier of size $\mbox{poly}(k,1/\epsilon)$ or even $\mbox{poly}(k,1/\epsilon,\log n)$. 
If a cut sparsifier can only contain terminals as nodes, then the approximation ratio must be at least $\Omega(\sqrt{\log k}/\log\log k)$~\cite{mm10}. If we need an exact cut sparsifier, then the size must be at least $2^{\Omega(k)}$~\cite{krauthgamer2017refined}.

In what follows we draw a connection between these two open problems; if there is a very efficient algorithm for the above cut sparsifier, then there cannot be a $\Omega(n^{1-o(1)})$ lower bound in the offline setting for the  dynamic approximate max flow.
Moreover, if the cut-sparsifier has size almost best possible, then there cannot be even a super-polylogarithmic lower bound.
Concretely, we show the following.
\begin{theorem}
\label{thm:lb max flow}
    If there is an algorithm that, given a undirected graph $G=(V,E)$ with $m$ edges and a set $T\subset V$ of $k$ terminals, constructs an $(1+\epsilon)$-approximate cut vertex sparsifier of size $s=\poly(k,1/\epsilon,\log n)$ in time $O(m\poly(\log n,1/\epsilon))$, there is an offline dynamic algorithm for maintaining $(1+\epsilon')$-approximate value of max flow with update time $u=O(n^{1-\gamma}\poly(1/\ensuremath{\epsilon}'))$ for some constant $\gamma>0$.
    Moreover, if the size of the sparsifier $s=k\cdot\poly(1/\epsilon,\log n)$, then we obtain the update time of $u=O(\poly(\log n,1/\epsilon'))$. The dynamic algorithm is Monte Carlo randomized and it is correct with high probability.
\end{theorem}
\begin{proof}
	Let us assume $\epsilon'$ is a constant for simplicity.
	The proof generalizes easily when $\epsilon'$ is not a constant. 
	
	First, we only need to consider offline dynamic algorithms where the underling graph has $m=\tilde{O}(n)$ edges at every time step and the length of the update sequences is $n$.
	This is because there is a dynamic algorithm by \cite{AbrahamDKKP16} that can maintain a cut sparsifier $H=(V,E')$ of a graph $G=(V,E)$ when the terminal set is $V$ with $\tilde{O}(1)$ worst-case update.
	So we can work on $H$ instead, and divide the update sequences into segements of length $n$.
	If we have an offline dynamic algorithm with update time $u$ on average on each period, then the average update time is $\tilde{O}(u)$ over the whole sequence.
	
	We set $\epsilon=\epsilon'/10\log n$.
	Suppose that the sparsifier from the assumption has size only $s=k\cdot\mbox{poly}(1/\epsilon,\log n)=\tilde{O}(k)$.
	Then, we apply the same proof as in \Cref{thm: offlineMaxFlow},  except that the number of levels of the decomposition tree will be $\log n$ instead of $O(\sqrt{\log n})$.
	The quality of the cut-sparsifier at any level is at most $(1+\epsilon)^{\log n}=(1+\epsilon'/10\log n)^{\log n}\le(1+\epsilon')$.
	The total running time will be $\tilde{O}(m^{1+\frac{1}{\log n+1}})=\tilde{O}(n)$.
	The latter implies that update time on average is $O(\mbox{polylog}(n))$.
	
	Assume that $s=k^{c}\cdot\mbox{poly}(1/\epsilon,\log n)=\tilde{O}(k^{c})$ for some constant $c>1$. Then, we can apply again the same proof from \Cref{thm: offlineMaxFlow}.
	By using only two levels of the decomposition tree, we can obtain an update time of $\tilde{O}(n^{1-\frac{1}{c+1}})$.
	Concretely, if we set $\beta_{0}=m$ and $\beta_{1}=m^{1/(c+1)}$ then the  time for computing the decomposition tree is $\frac{\beta_{0}}{\beta_{1}}\cdot\tilde{O}(\beta_{0})=\tilde{O}(n^{2-\frac{1}{c+1}})$.
	The total time for running approximate max flow on the cut-sparsifier in the second level at each step is $\beta_{0}\cdot\tilde{O}(\beta_{1}^{c})=\tilde{O}(n^{2-\frac{1}{c+1}})$.
	Thus it follows that the update time is $\tilde{O}(n^{1-\frac{1}{c+1}})$ on average.
\end{proof}

\subsubsection{Approximate distance oracles on general graphs}

There are previous hardness results for approximation algorithms for dynamic shortest path problems (including single-pair, single-source and all-pairs problems) \cite{HenzingerKNS15}.
All such results show a very high lower bound, e.g. $\Omega(n^{1-\epsilon})$ or $\Omega(n^{1/2-\epsilon})$ time on an $n$-node graph.
However, they hold only when the approximation factor is a small constant.
It is open whether one can obtain weaker polynomial lower bounds for larger approximation factors.
We show that it is impossible to show super-constant factor lower-bounds in several settings.
\begin{proposition}
	There is no polynomial lower bound for dynamic $\omega(1)$-approximate distance oracles in the offline setting (and also in the online incremental setting). 
	
	More formally, for any lower bound stating that $\omega(1)$-approximate offline dynamic distance oracle algorithm on $n$-node graphs requires at least $u(n)$ update time or $q(n)$ query time, then we have $u(n)=n^{o(1)}$ and $q(n)=n^{o(1)}$.
	The same holds for online incremental algorithm with worst-case update time.
\end{proposition}
This follows directly from \Cref{thm: IncrementalApproximateAPSP,thm: offlineShortestPaths}.

\subsubsection{Approximate distance oracles on planar graphs}

Similar to the situations above, assuming the APSP conjecture, Abboud and Dahlgaard \cite{AbboudD16} show that any offline fully dynamic algorithm for exact distance oracles on planar graph requires either update time or query time of $\Omega(n^{1/2-o(1)})$. 
We can still hope for a hardness result for $(1+\epsilon)$-approximate distance oracles, but this remains an important open problem in the field of dynamic algorithms.

Recall the definition of \emph{distance approximating minors}, which are vertex distance sparsifiers that are required to be minors of the input graph.
In the exact setting, Krauthgamer et al. \cite{KrauthgamerNZ14} showed that any distance preserving minor with respect to $k$ terminals, even when restricted to planar graphs, must have size $\Omega(k^{2})$ size.
Cheung et al.~\cite{cheung2016} showed that for planar graphs there is a $(1+ \epsilon)$-distance approximating minor of size $\tilde{O}(k^{2} \epsilon^{-2})$. 
The natural question is whether there is a $(1+\epsilon)$-approximate
{\em minor} distance sparsifier for $k$ terminals that has size $k^{1.99}\cdot\mbox{poly}(1/\epsilon,\log n)$.

We again draw a connection between dynamic graph algorithms and vertex sparsifiers; if there is a very efficient algorithm for such distance sparsifiers, then we cannot extend the $\Omega(n^{1/2-o(1)})$ lower bound to the approximate setting.
Moreover, if the sparsifier has the (almost) best possible size, then there cannot be even a super-polylogarithmic lower bound.
More precisely, we show the following.
\begin{theorem}
\label{thm:lb distance}
    Let $G$ be an undirected graph $G=(V,E)$ with $m$ edges and a set $T \subset V$ of $k$ terminals.
    If there is an algorithm that constructs a $(1+\epsilon)$-distance approximating minor of size $s=k^{2/(1+3\gamma)}\cdot\poly(1/\epsilon,\log n)$, for some constant $0<\gamma\le1/3$, in time $O(m\poly(\log n,1/\epsilon))$, then there is an offline dynamic $(1+\epsilon')$-approximate distance oracle algorithm for with update and query time $u=O(n^{1/2-\gamma/2})$.
	In fact, if the size of the sparsifier is $s=k\cdot\poly(1/\epsilon',\log n)$, then we obtain an update and query time of $u=O(\poly(\log n))$.
\end{theorem}
The proof will be very similar to the one in \Cref{thm:lb max flow} except that we need to be more careful about planarity.
Thus we first prove the following useful lemma.
\begin{lemma}
	Each vertex sparsifier $G'_{[r_{p},s_{p}]}$ corresponding to a node in our decomposition tree is planar.
\end{lemma}
\begin{proof}
	First, consider a sequence of $H_{[r_{1},s_{1}]},H_{[r_{2},s_{2}]},\dots,H_{[r_{p},s_{p}]}$ corresponding to a path in the decomposition tree, where $H_{[r_{1},s_{1}]}$ is a child of the root
	\footnote{Note that the graph $H_{[r,s]}$ is not defined at the root.}
	, and $H_{[r_{i},s_{i}]}$ is a parent of $H_{[r_{i+1},s_{i+1}]}$.
	Observe that $\cup_{1\le i\le p}H_{[r_{i},s_{i}]}=G_{[r_{p},s_{p}]}$ which is planar. 
	
	From \Cref{alg: vertexNodeSparsiy}, we unfold the recursion and obtain that 
	\[
	    G'_{[r_{p},s_{p}]}=
	    \cutsp(\cutsp(\ldots)\cup H_{[r_{p-1},s_{p-1}]})\cup H_{[r_{p},s_{p}]}).
	\]
	Note that we omit the second parameter of $\cutsp$ only for readability.
	We assume by induction $G'_{[r_{p-1},s_{p-1}]}=\cutsp(\cutsp(...)\cup H_{[r_{p-1},s_{p-1}]})$ is planar.
	We will prove that $G'_{[r_{p},s_{p}]}$ planar.
	To this end, observe that $G'_{[r_{p-1},s_{p-1}]}$ is a minor of $\cup_{1\le i\le p-1}H_{[r_{i},s_{i}]}$.
	Next, we need the following observation.
	\begin{claim}
	\label{claim:preserve planar}
		Let $G_{1}$ be a minor of $G_{2}$.
		Let $(u,v)$ be an edge such that $u,v\in V(G_{1})\cap V(G_{2})$, i.e., the endpoints are nodes of both $G_{1}$ and $G_{2}$.
		Then, $G_{1} \cup \{(u,v)\}$ is a minor of $G_{2} \cup \{(u,v)\}$.
		In particular, if $G_{2} \cup \{(u,v)\} (u,v)$ is planar, then so is $G_{1}\cup \{(u,v)\}$.
	\end{claim}
	
	We apply \Cref{claim:preserve planar} where $G_{2}=\cup_{1\le i\le p-1}H_{[r_{i},s_{i}]}$ and $G_{1}=G'_{[r_{p-1},s_{p-1}]}$.
	As the endpoints of $H_{[r_{i},s_{i}]}$ are in both $G_{1}$ and $G_{2}$ by construction and $G_{2}\cup H_{[r_{p},s_{p}]}=\cup_{1\le i\le p}H_{[r_{i},s_{i}]}$ is planar, then $G_{1}\cup H_{[r_{p},s_{p}]}$ is planar.
	Finally, $G'_{[r_{p},s_{p}]}=\cutsp(G_{1}\cup H_{[r_{p},s_{p}]})$ is a minor of $G_{1}\cup H_{[r_{p},s_{p}]}$, so $G'_{[r_{p},s_{p}]}$ is planar.
\end{proof}
Now, we prove \Cref{thm:lb distance}.
\begin{proof}[Proof of \Cref{thm:lb distance}]
    We first prove the case when $s=k\cdot\mbox{poly}(1/\epsilon,\log n)$.
	We again prove the theorem when $\epsilon'$ is a constant for simplicity.
	Set $\epsilon=\epsilon'/10\log n$.
	We build the corresponding decomposition tree with $\log n$ levels.
	The quality of the sparsifier at any level is at most $(1+\epsilon)^{\log n}=(1+\epsilon'/10\log n)^{\log n}\le(1+\epsilon')$.
	The total running time will be $\tilde{O}(m^{1+\frac{1}{\log n+1}})=\tilde{O}(n)$ using the same argument as in Lemma~\ref{lem: insertRunningTime}.
	That is the update time on average is $O(\mbox{poly}(\log n))$.
	
	For the case when $s=k^{2/(1+3\gamma)}\cdot\mbox{poly}(1/\epsilon,\log n)$, the proof is the same except the parameters need to be carefully chosen.
	Set $\epsilon=\epsilon'\gamma/2$.
	We choose $\beta_{0}=m=O(n)$, $\beta_{1}=n^{(1+\gamma)/2}$, and $\beta_{i+1}=n^{(1+\gamma-2\gamma i)/2}$ for $i\ge0$.
	We get that there will be at most $1/\gamma$ levels in the decomposition tree and thus the quality at each level is at most 
	\[
	    (1+\epsilon)^{1/\gamma} \le
	    e^{\epsilon/\gamma} =
	    e^{\epsilon'/2}\le(1+\epsilon')
	\]
	because $(1+x)\le e^{x}$ for any $x$ and $e^{x/2}\le(1+x)$ for $0\le x\le1$. 
	
	For each $i$, the total time to build the sparsifiers in level $i+1$ by running the algorithm sparsifier at level $i$ is $n/\beta_{i+1}\cdot\tilde{O}(\beta_{i}^{2/(1+3\gamma)})$.
	This is because there are $n/\beta_{i+1}$ many sparsifiers, and the algorithm is applied on a graph of size $\tilde{O}(\beta_{i}^{2/(1+3\gamma)})$.
	By direct calculation we have that
	\[
	    n/\beta_{i+1}\cdot\beta_{i}^{2/(1+3\gamma)}=
	    n^{1-(1+\gamma-2i\gamma)/2+\frac{(1+\gamma-2(i-1)\gamma)}{(1+3\gamma)}}\le n^{1.5-\gamma/2}.
	\]
	To see this, note that $2/(1+3 \gamma) \geq 1$ and consider the following chain of inequalities:
	\begin{align*}
	    \frac{(1+\gamma-2(i-1)\gamma)}{(1+3\gamma)}&-(1+\gamma-2i\gamma)/2\\
	    &\le \frac{1+\gamma}{1+3\gamma}-(i-1)\gamma-\frac{1+\gamma}{2}+i\gamma\\
	    &\le (1-\gamma)+\gamma-1/2-\gamma/2\\
	    &=   1/2-\gamma/2.
	\end{align*}
	It follows that the total time over all levels is $\frac{1}{\gamma}\cdot O(n^{1.5-\gamma/2})$, which is turn implies an average update time of $O(n^{0.5-\gamma/2})$. This completes the proof. 
\end{proof}

\section{Fully-Dynamic Algorithms via Fully-Dynamic Vertex Sparsifier}

\label{sec: fullyDynamicMetaTHm}

In this section, we present a meta data-structure for dynamically maintain $\P$, some properties of the graph $G$.
Like before~\label{def: IVS}, we also utilize the idea of $\alpha-$\emph{vertex sparsifier}.
Unlike the incremental scheme, we introduce another parameter $s$ specifying the edge-sparsity of the vertex sparsifier.

Let $G=(V,E)$ be a graph.
An $(\alpha,s)-$\emph{vertex sparsifier} of $G$ for $\P$ w.r.t. the terminal set $T$ is a graph $H$ that is an $\alpha-$\emph{vertex sparsifier} of $G$ for $\P$ w.r.t. $T$ with number of edges, $|E(H)|$, bounded by $O(|T|s)$.

\begin{definition}[Fully-Dynamic Vertex Sparsifer]
\label{def: DVS}
Let $G=(V,E)$ be a graph, $T \subseteq V$ be the set of terminals, and $\alpha$ be an non-negative parameter.
A data-structure $\D$ is an \emph{$(\alpha, s)-$Fully-Dynamic Vertex Sparsifier} (abbr. \emph{$(\alpha, s)-$DVS}) of $G$ if $\D$ explicitly maintains an $(\alpha, s)-$vertex sparsifier $H_T$ of $G$ with respect to $T$ and supports following operations: 
\begin{enumerate}
\setlength{\itemsep}{0em}
\item \textsc{Preprocess}$(G, T)$: preprocess $G$ in time $t_p$
\item \textsc{AddTerminal}$(u)$: let $T'$ be $T \cup \{u\}$ and update $H_T$ to $H_{T'}$ in time $t_a$ such that i) $H_{T'}$ is an $(\alpha, s)-$vertex sparsifier of $G$ with respect to $T'$, and ii) the \emph{recourse}, number of edge changes from $H_T$ to $H_{T'}$, is at most $r_a$.
This operation should return the set of edges inserted and deleted from $H_T$.
\item \textsc{Delete}$(e)$: delete $e$ from $G$ while maintaining $H_T$ being an $(\alpha, s)-$vertex sparsifier of $G$ with respect to $T$ in $t_d$ time and $r_d$ recourse in $H_T$.
\end{enumerate}
Also, such $H_T$ have size $O(|T|s)$.
\end{definition} 

Given a $(\alpha,s)-$DVS $\D$ of $G$, we can support edge insertion by first adding 2 endpoints to the terminal set and then add the edge directly to the vertex sparsifier $H_T$.
The correctness comes from the decomposability.
To specify the cost, the following corollary is presented to formalized the approach.
\begin{lemma}
\label{lem:DVS1levelEdgeInsertion}
Let $G=(V,E)$ be a graph and $\D$ be an $(\alpha,s)-$DVS of $G$.
Let $t_a, r_a$ be the time and recourse for $\D$ to handle \textsc{AddTerminal} operation.
$\D$ also maintains an $(\alpha,s)-$vertex sparsifier $H_T$ of $G$ with respect to $T$ subject to the following operation:
\begin{itemize}
\setlength{\itemsep}{0em}
\item \textsc{Insert}$(e)$: insert $e$ to $G$ while maintaining $H_T$ being an $(\alpha, s)-$vertex sparsifier of $G$ with respect to $T$ in $O(t_a)$ time and $O(r_a)$ recourse in $H_T$.
\end{itemize}
\end{lemma}
\begin{proof}
The proof works similar to the one of Lemma~\ref{lem:IVS1levelEdgeInsertion}.
\end{proof}






To design a fully-dynamic data structure that support update and queries in sub-linear time, we will focus on building a fully-dynamic vertex sparsifier whose update time and recourse are sub-linear in $n$. This requirement is made precise in the following definition.

\begin{definition}[Efficient DVS] \label{def: efficientDVS}
Let $G=(V,E)$ be a graph, $\alpha$ be an non-negative parameter, and $s(n)$, $f(n,k)$, $g(n, k)$, $r(n, k) \geq 1$ be functions where $n$ and $k$ correspond to the maximum size of the vertex set and the terminal set respectively.
We say that $\D$ is an $(\alpha,s(k),f(n,k),g(n, k),r(n, k))-$\emph{efficient DVS} of $G$ if $\D$ is an $(\alpha,s(k))-$\emph{DVS} of $G$ with \emph{Preprocessing} time $t_p=O(m \cdot f(n,k))$, \emph{AddTerminal} time $t_a = O(g(n, k))$, and recourse $r_a = O(r(n, k))$, and \emph{Delete} time $t_d = O(g(n, k))$, and recourse $r_d = O(r(n, k))$.
\end{definition}


\label{sec: FullyDynamicVertexSparsifier}
Next we show how to use an efficient fully-dynamic vertex sparsifier to design an (approximate) fully-dynamic algorithms for problems with certain properties while achieving fast amortized update and query time.

\begin{theorem}
\label{thm: FullyDynamicMetaTheorem}
Let $G=(V,E)$ be a graph, and for any $u,v \in V$, let $\mathcal{P}(u,v,G)$ be a solution to a minimization problem between $u$ and $v$ in $G$.
Let $s(k),f(n,k),g(n, k),r(n, k),h(n) \geq 1$ be functions with $s(k) \le r(n, k)$, $\alpha\geq 1$ be the approximation factor, $\ell\geq 1$ be the depth of the data structure, and let $\mu_0,\mu_1,\ldots, \mu_{\ell}$ with $\mu_0 = m$ be parameters associated with the running time.
Assume the following properties are satisfied
\begin{enumerate}
\setlength{\itemsep}{0em}
\label{metaThM: P1}\item $G$ admits an $(\alpha, s(k),f(n,k),g(n, k),r(n,k))-$\emph{efficient DVS}
 \item The property $\mathcal{P}(u,v,G)$ can be computed in $O(m h(n))$ time in a static graph with $m$ edges and $n$ vertices.
\end{enumerate}
Then there is a (approximate) fully-dynamic algorithm that maintains for every pair of nodes $u$ and $v$, an estimate $\delta(u,v)$, such that
\begin{equation} \label{eq: approxMeta}
        \mathcal{P}(u,v,G) \leq \delta(u,v) \leq \alpha^{\ell} \cdot \mathcal{P}(u,v,G),
\end{equation}
with amortized update time of
\[
        T_u = O\left(\sum_{i=1}^{\ell}c^{i-1}\prod_{j=1}^{i-1}r(\mu_j,\mu_{j-1})\left(\frac{\mu_{i-1}s(\mu_{i-1})f(\mu_{i-1},\mu_i)}{\mu_{i}}+g(\mu_{i-1}, \mu_i)\right)\right),
\]
and amortized query time of
\[
        T_q = O\left(\ell T_u + \mu_\ell s(\mu_\ell) h(\mu_\ell) \right),
\]
where $c < 3$ is an universal constant.
\end{theorem}

We prove the theorem in the rest of the section.

\paragraph*{Data Structure.} 

Given some integer parameter $\ell \ge 1$, and parameters $m = \mu_0 \ge \ldots, \mu_\ell$.
Our data structure maintains
\begin{enumerate}[noitemsep]
\item a hierarchy of graphs $\{G_i\}_{0 \le i \le \ell}$,
\item a hierarchy of terminal sets $\{T_i\}_{1 \le i \le \ell}$, each associated with the parameters $\{\mu_{i}\}_{1 \leq i \leq \ell}$, and
\item a hierarchy of $(\alpha, s(k), f(n,k), g(n, k), r(n, k))$-efficient DVSs $\{\D_i\}_{1 \le i \le \ell}$, each associated with a graph $G_{i-1}$ and the terminal set $T_i$.
\end{enumerate}

The data structure is initialized recursively.
First, initialize $T_i \gets \phi$ for $1 \le i \le \ell$ and $G_0 \gets G$.
For $1 \le i \le \ell$, construct an $(\alpha, s(n), f(n,k), g(n,k), r(n,k))$-efficient DVS $\D_i$ of graph $G_{i-1}$ w.r.t. the terminal set $T_i$ and set $G_i$ be the sparsifier maintained by $\D_i$.
Hence, $G_i$ is an $\alpha-$vertex sparsifier of $G_{i-1}$ w.r.t. $T_i$.
By transitivity, we know $G_{\ell}$ is an $\alpha^\ell-$vertex sparsifier of $G_0$, which is the input graph $G$.

When answering a query, we add both $s$ and $t$ to the terminal sets of every $\D_i$.
Then compute $\mathcal{P}(s, t, G_\ell)$ and output the value as an estimate of $\mathcal{P}(s,t,G)$.

When dealing with a edge update in $G$, we first add both endpoints to the terminal set of $\D_1$.
Then we update the edge in $G_1$, the vertex sparsifier maintained by $\D_1$.
Edge changes propagate down the hierarchy.
To bound the size of each $G_i, 1 \le i \le \ell$, rebuild $\D_i$ for every $2\mu_i$ updates in $G_{i-1}$ w.r.t. the most recently added $\mu_i$ terminal vertices.

The rebuild scheme is easier than the incremental case.
When rebuilding $\D_i$, we treat it as a series of edge updates in $G_i$.
Since $\D_i$'s are DVS, they can handle both edge insertions and deletion.

\paragraph*{Running Time.}

We first study the update time of our data structure.

$\D_i$ maintains the graph $G_{i-1}$, which has at most $\mu_{i-1}$ vertices, and it has at most $\mu_i$ terminal vertices due to rebuild.
Hence $\D_i$ spends $O(g(\mu_{i-1}, \mu_i))$-time per operation of \textsc{AddTerminal} or edge updates.

Since rebuild is incurred every $2\mu_i$ operations for the data structure $\D_i$,
we can charge the rebuild cost among $\mu_i$ operations.
Note that $\D_i$ is an $(\alpha,s(n))$-DVS of $G_{i-1}$, which is a graph with $O(\mu_{i-1})$ vertices and $O(\mu_{i-1}s(\mu_{i-1}))$ edges.
By amortizing the rebuild cost, we know the time $\D_i$ spent on either \textsc{AddTerminal} or edge updates is:
\begin{align*}
    O\left(\frac{\mu_{i-1}s(\mu_{i-1})f(\mu_{i-1},\mu_i)}{\mu_i} + g(\mu_{i-1}, \mu_i)\right).
\end{align*}

Since an update in $\D_i$ creates $O(r(\mu_{i-1}, \mu_i))$ updates in $G_i$, which is handled by the data structure in the next level, $\D_{i+1}$, we have to incoporate such quantity in to the analysis.
Also, we have to take the recourse from rebuild into account.
Every rebuild creates $O(|E(G_i)|) = O(\mu_is(\mu_i))$ updates to $G_i$.
By amortizing it over $\mu_i$ updates in $\D_i$, each update in $G_i$ has recourse $O(s(\mu_i) + r(\mu_{i-1}, \mu_i)) = O(r(\mu_{i-1}, \mu_i))$ since $s(n) \le r(n, k)$.
Thus one update in $\D_i$ creates $\le c r(\mu_{i-1}, \mu_i)$ updates in $G_i$ for $c$ being some universal positive constant.
We can now analyze the amount of updates handled by $\D_i$ when 1 edge update happens in $G$.
By simple induction, we know there will be $\le \prod_{j=1}^{i-1}(c r(\mu_j,\mu_{j+1}))$ updates in $G_{i-1}$.
Thus for the data structure in $i$-th level, $\D_i$, there will be $\le c^{i-1}\prod_{j=1}^{i-1}r(\mu_j,\mu_{j+1})$ updates.
Combining these 2 quantities, we can bound the amortized update time of our data structure:
\begin{align*}
    T_u
    &=
    O\left(\sum_{i=1}^{\ell}c^{i-1}\prod_{j=1}^{i-1}r(\mu_j,\mu_{j+1})\left(\frac{\mu_{i-1}s(\mu_{i-1})f(\mu_{i-1},\mu_i)}{\mu_i} + g(\mu_{i-1}, \mu_i)\right)\right).
\end{align*}

We next study the query time of our data-structure.
When answering a query, we add $s$ and $t$ to each layer of the terminal set.
When adding terminals to each layer of data structure $\D_i$, edge changes also propagate to lower levels.
As for the analysis of update time, we have to take the recourse into account.
The time can be bounded by $O(\ell T_u)$ where $T_u$ is the update time of our data structure.

Then we compute $\mathcal{P}(s, t, G_{\ell})$ in $G_{\ell}$, which has $\mu_\ell s(\mu_\ell)$ edges as guaranteed by the definition of DVS.
Since we have an $O(mh(n))$ algorithm for computing $\mathcal{P}(s, t, G)$ in an $m$-edge $n$-vertex graph $G$, $\mathcal{P}(s, t, G_{\ell})$ can be computed in $\mu_\ell s(\mu_\ell)h(\mu_\ell)$ time.
Combining these 2 bounds, we can bound the query time by:
\begin{align*}
    T_q = O\left(\ell T_u + \mu_\ell s(\mu_\ell) h(\mu_\ell) \right).
\end{align*}

\begin{lemma}
\label{lem: DVSOptimalTradeoff}
Let $\{\mu_i\}_{0 \leq i \leq \ell}$ be a family of parameters with $\mu_0 = m$.
Suppose $s(n) = n^{o(1)}, f(n,k)=O((n/k)^d)$, both $g(n,k)$ and $r(n,k)$ are of order $O((n/k)^e)$ for some positive constants $d + 1 \le e$ and $h(n)=n^{o(1)}$.
Let $t$ be any positive constant, if we set
\[
    \mu_i = m^{1-i/(\ell+t)},
    ~\text{ where }1 \leq i \leq \ell,
\] then the update time is
\begin{equation}
    T_u
    =
    O\left(\sum_{i=1}^{\ell}c^{i-1}\prod_{j=1}^{i-1}r(\mu_j,\mu_{j+1})\left(\frac{\mu_{i-1}s(\mu_{i-1})f(\mu_{i-1},\mu_i)}{\mu_{i}}+g(\mu_{i-1},\mu_i)\right)\right)
    =
    O\left(\ell c^{\ell}m^{e\ell / (\ell+t)}m^{o(1)/(\ell+t)}\right)
    ,
\end{equation}
and the query time is
\begin{equation}
    T_q = O\left(\ell T_u + \mu_\ell s(\mu_\ell) h(\mu_\ell) \right) =
    O\left(\max\{\ell^2 c^{\ell}m^{e\ell / (\ell+t)}m^{o(1)/(\ell+t)},m^{(t+o(1)) / (\ell+t)}\}\right).
\end{equation}
\end{lemma}
\begin{proof}
Plug in the choice of $\mu_i$, we have
\begin{align*}
    \frac{\mu_{i-1}s(\mu_{i-1})f(\mu_{i-1},\mu_{i})}{\mu_{i}}+g(\mu_{i-1},\mu_i) = 
    m^{(1+o(1)+d)/(\ell+t)}+m^{e/(\ell+t)} =
    O(m^{(o(1)+e)/(\ell+t)}).
\end{align*}

Also, since $i \le \ell$,
\begin{align*}
    c^{i-1}\prod_{j=1}^{i-1}r(\mu_j,\mu_{j+1})
    &\le
    c^{\ell-1}\prod_{j=1}^{\ell-1}r(\mu_j,\mu_{j+1})\\
    &=
    O(c^{\ell-1}\prod_{j=1}^{\ell-1}m^{e/(\ell+t)})\\
    &=
    O(c^{\ell-1}m^{e(\ell-1)/(\ell+t)})
\end{align*}

Combining these 2 inequalities, we can bound the update time by
\begin{align*}
    T_u &=
    O\left(
        \sum_{i=1}^{\ell}
            c^{i-1}\prod_{j=1}^{i-1}r(\mu_j,\mu_{j+1})
            \left(
                \frac{\mu_{i-1}s(\mu_{i-1})f(\mu_{i-1},\mu_i)}{\mu_{i}}+g(\mu_{i-1},\mu_i)
            \right)
    \right)\\
    &=
    O\left(\sum_{i=1}^{\ell}c^{\ell-1}m^{e(\ell-1)/(\ell+t)}m^{(o(1)+e)/(\ell+t)}\right)\\
    &=
    O\left(\ell c^{\ell}m^{e\ell / (\ell+t)}m^{o(1)/(\ell+t)}\right)
\end{align*}

The bound for query time is straightforward from the definition of $\mu_\ell$.
\end{proof}

\begin{corollary}
\label{coro: DVSOptimalTradeoffFastInit}
To asymptotically minimize the query time given in Lemma~\ref{lem: DVSOptimalTradeoff}, we set
\[
    t \gets e\ell, \ell \gets O(1).
\]
We have update time
\begin{equation}
\label{eq: DVSUpdateMinimized}
    T_u = 
    O\left(m^{(e + o(1))/(e+1)}\right)
\end{equation}
and query time
\begin{equation}
\label{eq: DVSQueryMinimized}
    T_q =
    O\left(m^{(e + o(1))/(e+1)}\right).
\end{equation}
\end{corollary}

\begin{corollary}
\label{coro: DVSOptimalTradeoffSlowInit}
Under the same setting as Lemma~\ref{lem: DVSOptimalTradeoff} except $d+1 > e$,
we set
\[
    t \gets e\ell, \ell \gets \sqrt{\log{m}}.
\]
We can asymptotically minimize update time
\begin{equation}
\label{eq: DVSUpdateMinimized}
    T_u = 
    O\left(m^{e/(e+1) + o(1)}\sqrt{\log{m}}\right)
\end{equation}
and query time
\begin{equation}
\label{eq: DVSQueryMinimized}
    T_q =
    O\left(m^{e/(e+1) + o(1)}\log{m}\right).
\end{equation}
\end{corollary}
\begin{proof}
    Under this setting, the only difference in the analysis of Lemma~\ref{lem: DVSOptimalTradeoff} is
    \begin{align*}
        \frac{\mu_{i-1}s(\mu_{i-1})f(\mu_{i-1},\mu_{i})}{\mu_{i}}+g(\mu_{i-1},\mu_i) = 
    m^{(1+o(1)+d)/(\ell+t)}+m^{e/(\ell+t)} =
    O(m^{(1+o(1)+d)/(\ell+t)}) = O(m^{o(1)}).
    \end{align*}
    Thus,
    \begin{align*}
        T_u &=
        O\left(
            \sum_{i=1}^{\ell}
                c^{i-1}\prod_{j=1}^{i-1}r(\mu_j,\mu_{j+1})
                \left(
                    \frac{\mu_{i-1}s(\mu_{i-1})f(\mu_{i-1},\mu_i)}{\mu_{i}}+g(\mu_{i-1},\mu_i)
                \right)
        \right)\\
        &=
        O\left(\sum_{i=1}^{\ell}c^{\ell-1}m^{e(\ell-1)/(\ell+t)}m^{o(1)}\right)\\
        &=
        O\left(\ell c^{\ell}m^{e(\ell-1) / (\ell+e\ell)}m^{o(1)}\right)\\
        &=
        O\left(\ell c^{\ell}m^{e / (e+1)}m^{o(1)}\right)\\
        &=
        O\left(\sqrt{\log{m}} \cdot m^{e / (e+1)} m^{o(1)}\right)\\
    \end{align*}
\end{proof}

\section{Fully-Dynamic All Pair Max-Flow/Min-Cut}
In this section, we incorporate \emph{Local Sparsifier} with tools developed previously to give a data structure for min cut query in a dynamic graph.
The main theorem we prove is as follows.

\begin{theorem}
\label{theorem:DynamicMinCut}
Given a graph $G=(V, E, c)$ with weight ratio $U=\poly(n)$.
There is a dynamic data structure maintaining $G$ subject to the following operations:
\begin{enumerate}
    \item $\textsc{Insert}(u, v, c)$:
          Insert the edge $(u, v)$ to $G$ in amortized $O(m^{2/3}\log^{7}{n})$-time.
    \item $\textsc{Delete}(e)$:
          Delete the edge $e$ from $G$ in amortized $O(m^{2/3}\log^{7}{n})$-time.
    \item $\textsc{MinCut}(s, t)$:
          Output a $\Tilde{O}(\log{n})$-approximation to the min-$st$-cut value of $G$ in $\Tilde{O}(m^{2/3})$-time w.h.p..
          The cut set $S$ can be obtained on demand with linear overhead in $|S|$.
\end{enumerate}
\end{theorem}
Using the so-called \emph{dynamic sparsifier}, we can speed-up Theorem~\ref{theorem:DynamicMinCut} by a factor of $O(m/n)$.
Which is significant for $G$ being a dense graph originally.

\begin{corollary}
\label{corollary:DynamicMinCut}
Given a graph $G=(V, E, c)$ with weight ratio $U=\poly(n)$.
There is a dynamic data structure maintaining $G$ subject to the following operations:
\begin{enumerate}
    \item $\textsc{Insert}(u, v, c)$:
          Insert the edge $(u, v)$ to $G$ in amortized $O(n^{2/3}\log^{41/3}{n})$-time.
    \item $\textsc{Delete}(e)$:
          Delete the edge $e$ from $G$ in amortized $O(n^{2/3}\log^{41/3}{n})$-time.
    \item $\textsc{MinCut}(s, t)$:
          Output a $\Tilde{O}(\log{n})$-approximation to the min-$st$-cut value of $G$ in $\Tilde{O}(n^{2/3})$-time w.h.p..
          The cut set $S$ can be obtained on demand with linear overhead in $|S|$.
\end{enumerate}
\end{corollary}
\subsection{Cut, Flow and $L_\infty$-Embeddability}
In this section, we present concepts critical in building local sparsifier preserving cut/flow value.

\begin{definition}
\label{definition:graph_cut}
    Given a graph $G = (V, E, c)$, a cut $C$ is any proper subset of $V$, i.e., $\emptyset \neq C \subset V$.
    Define $E(C) \coloneqq E(C, V \setminus C)$ and $c(C) \coloneqq \sum_{e \in E(C)}{c(e)}$.
\end{definition}

\begin{definition}~\cite{Madry10}
\label{definition:multicommodity_flow}
    Given a graph $G = (V, E, c)$, $\mathbf{f} = (f^1, \ldots, f^k) \in \Real^{E \times k}$ is \emph{multicommodity flow} if $f_i$ is a $s_it_i$-flow.
    Define $|\mathbf{f}(e)| \coloneqq \sum_{i=1}^{k}{|f^i(e)|}$ as the total flow crossing the edge $e \in E$.
    A multicommodity flow $\mathbf{f}$ is feasible if for every edge $e$, we have $|\mathbf{f}(e)| \le c(e)$.
    \emph{Single-commodity flow} $\mathbf{f}$ is a multicommodity flow with $k=1$.
\end{definition}
We use the notion graph embedding for describing relation between flows in 2 graphs.

\begin{definition}~\cite{Madry10}
\label{definition:embedding}
    Given 2 graphs with same vertex set $G=(V,E,c), H=(V,E_H,c_H)$.
    For every edge $e=uv \in E$, $f^e$ is a flow that routes $c(e)$ amount of flow from $u$ to $v$ in $H$.
    Then $\mathbf{f} \coloneqq (f^e\mid e \in E) \in \Real^{E_H \times E}$, the collection of $f^e$ for every edge $e \in E$, is an \emph{embedding} of $G$ into $H$.
\end{definition}
Use the notion of embedding, we can define embeddability between graphs over same vertex set.

\begin{definition}~\cite{Madry10}
\label{definition:embeddability}
    Given $t \ge 1$, and 2 graphs with same vertex set $G=(V,E,c), H=(V,E_H,c_H)$.
    We say $G$ is \emph{$t$-embeddable} into $H$, denoted by $G \preceq_t H$ if there is an embedding $\mathbf{f}$ of $G$ into $H$ such that for every $e_H \in E_H$, $|\mathbf{f}(e_H)| \le t \cdot c_H(e_H)$.
    For $t=1$, we ignore the subscript and say $G$ is embeddable into $H$ and $G \preceq H$.
\end{definition}
Here we also define the notion of cut approximation between graphs.

\begin{definition}
\label{definition:CutApproximation}
    Given a graph $G=(V, E, c)$ and $\epsilon \in (0, 1)$, we say a graph $H=(V, E_H \subseteq E, c_H)$ is a \emph{$(1+\epsilon)$-cut-sparsifier} of $G$ if $S \subseteq V$,
\[
    (1-\epsilon)c(S) \le c_H(S) \le (1+\epsilon)c(S).
\]
    Denoted by $H \sim_\epsilon G$.
\end{definition}
Here we present some structural results about the value of max flow/min cut between mutually embeddable graphs.

\begin{lemma}
\label{lemma:embed_cut_upper_bound}
    Given $t \ge 1$, and 2 graphs $G=(V,E,c)$, $H=(V,E_H,c_H)$ on the same vertex set such that $G \preceq_t H$.
    For any cut $S$, we have $c(S) \le t\cdot c_H(S)$.
\end{lemma}
\begin{proof}
    By the property of flow, $|\sum_{d \in E_H(S)}{f^e(d)}| = c(e)$ if $e \in E(S)$ and $0$ otherwise.
    Therefore,
    \begin{align*}
        c(S) &= \sum_{e \in E(S)}{c(e)} = \sum_{e \in E}{\left|\sum_{d \in E_H(S)}{f^e(d)}\right|} \\
        &\le \sum_{e \in E}{\sum_{d \in E_H(S)}{\left|f^e(d)\right|}} \\
        &= \sum_{d \in E_H(S)}{\sum_{e \in E}{\left|f^e(d)\right|}} \\
        &= \sum_{d \in E_H(S)}{\left|f(d)\right|} \\
        &\le \sum_{d \in E_H(S)}{t \cdot c_H(d)} = t \cdot c_H(S).
    \end{align*}
\end{proof}

\begin{corollary}
\label{corollary:mutual_embed_preserve_cut}
    Given $t \ge 1$, and 2 graphs $G=(V,E,c)$, $H=(V,E_H,c_H)$ on the same vertex set such that $G \preceq H$ and $H \preceq_t G$.
    For any cut $C \subseteq V$, we have
    \begin{itemize}
        \item $c(C) \le c_H(C)$ and
        \item $c_H(C) \le t\cdot c(C)$
    \end{itemize}
\end{corollary}
\begin{proof}
It directly comes from Lemma~\ref{lemma:embed_cut_upper_bound}.
\end{proof}

\subsection{$J$-trees as Vertex Sparsifier}
Here we introduce the notion of $j$-tree ~\cite{Madry10}.
Intuitively, $j$-tree is obtained from the original graph by contracting vertices.
The resulting graph has at most $j$ vertices but preserves the value of cuts between them.
We deploy such a routine in reducing the number of vertices in the original graph.
This helps the design of the data structure we need.
Since we can answer queries by computing min cut values in a smaller graph.

\begin{definition}~\cite{Madry10}
\label{definition:JT}
    Given $j \le 1$, $H = (V_H, E_H, c_H)$ is a \textit{$j$-tree} if it is a connected graph being a union of a \textit{core} $C(H)$ which is a subgraph of $H$ induced by some vertex set $C \subseteq V_H$ with $|C| \le j$;
    and of an \textit{envelope} $F(H)$ which is a forest on $H$ with each component having exactly one vertex in the core $C(H)$.
    For any core vertex $u \in C$, define $F(u)$ to be the vertex set of the component containing $u$ in the envelope.
    Also for $S \subseteq C$, $F(S)$ is the union of $F(u), \forall u \in S$.
\end{definition}
We are interested in such $j$-tree structure because of the following lemma.
It suggests that we can approximate max flow/min cut within several simpler graphs and introduce a sampling scheme that can reduce the number of them to consider.

But first, we are going to define the notion of the congestion $\rho$-decomposition of a graph $G$.

\begin{definition}~\cite{Madry10}
\label{definition:Decomp}
    A family of graphs $G_1 \ldots G_{k}$ is a $(k, \rho, j)$-decomposition of a graph $G = (V, E, c)$ if
    \begin{enumerate}
        \item Each $G_i$ is a $j$-tree, and there are at most $k$ of them.
        \item $G \preceq G_i, \forall i$.
        \item $\sum{G_i} \preceq_{k \cdot \rho} G$.
    \end{enumerate}
\end{definition}
Using this definition, we are able to state more clearer of the benefit of using $j$-trees.

\begin{lemma}~\cite{Madry10}
\label{lemma:static_jtree_decomposition}
    Given any graph $G=(V,E,c)$ with weight ratio $U$ and $k \ge 1$, we can find in $O(km\log^{4}n)$-time a
    \[
        \left(
            k,
            \rho:=O\left(\log{n}\cdot\log\log{n}\cdot(\log\log\log{n})^3\right),
            O\left(\log^2{n}\frac{m\log{U}}{k}\right)
        \right)
    \]
    -decomposition of $G$, $G_1 \ldots G_k$.
    The weight ratio of each $G_i$ is $O(mU)$.
    Moreover, if we sample $G_i$ with probability $1/k$, for any fixed cut $S$, the size of this cut in $G_i$ is at most $2 \rho$ times the size of the cut in $G$ with probability at least $\frac{1}{2}$.
\end{lemma}
To speed-up, we can not afford to maintain this many $j$-trees.
The following lemma shows that we can maintain only $O(\log{n})$ of them while preserving the approximation quality.

\begin{lemma}
\label{lemma:jtree_preserve_all_min_st_cut}
    Given any graph $G=(V, E, c)$ with weight ratio $U$ and $k = \Omega(\log{n})$ and a $\left( k, \rho, j \right)$-decomposition of $G$, $G_1, \ldots, G_k$.
    By sampling $O(\log{n})$ graphs from $G_1, \ldots ,G_k$, every min $st$ cut is preserved up to a $2\rho$-factor with high probability.
\end{lemma}
\begin{proof}
    Let $C(s, t)$ be some minimum $st$-cut in $G$.
    Let $\mathcal{C} = \{C(s, t) \mid s, t \in {V \choose 2}\}$.
    Clearly $|\mathcal{C}| \le n^2$.
    By lemma~\ref{lemma:static_jtree_decomposition} we know $c(C(s, t)) \le c_{G_i}(C(s, t))$ and with probability at least 0.5, $c_{G_i}(C(s, t)) \le 2\rho c(C(s, t))$ by sampling $G_i$ uniformly.
    By sampling $t=d\log{n}$ of them, say $G_1, \ldots, G_t$, we have with probability at least $1 - n^{-d}$ that:
    \begin{align*}
        \min_{i \in [t]}\left\{c_{G_i}\left(C(s, t)\right)\right\} \le 2\rho c\left(C(s, t)\right).
    \end{align*}
    Taking the union bound over all cuts in $\mathcal{C}$ which has at most $n^2$ of them, we have with probability $1 - n^{-d+2}$ that \begin{align*} 
        \forall s, t \in {V \choose 2}, c\left(C(s, t)\right) \le \min_{i \in [t]}\left\{c_{G_i}\left(C(s, t)\right)\right\} \le 2\rho c\left(C(s, t)\right).
    \end{align*}
    \end{proof}
Lemma~\ref{lemma:jtree_preserve_all_min_st_cut} says that optimal $st$-cut can be preserved using $O(\log{n})$-many $j$-trees up to a $O(\rho)$-factor with high probability.
In our usage, we compute cuts only in the core instead of the entire $j$-tree.
Cuts in the core are then projected back to the $j$-tree.
To justify the correctness in terms of minimum $st$-cut, we define the concept of \emph{core cut}.

\begin{definition}
\label{definition:core_cut}
    Let $j \ge 1$ and $H=(V, E_H, c_H)$ be some $j$-tree.
    Let $C_H=(C \subseteq V, E_C, c_C)$ be the core of $H$.
    Given any cut $S \subseteq C$ of $C_H$, the \emph{core cut} of $S$ with respect to $C$ in $H$ is
    \begin{align*}
    \Pi(S) \coloneqq \bigcup_{u \in S}{F(u)}.
    \end{align*}
    That is, extend the cut $S$ by including trees in the envelope rooted at vertex in $S$.
\end{definition}
Now we present a lemma that justify computing min cuts in the core.

\begin{lemma}
\label{lemma:only_core_cuts_matter}
    Let $j \ge 1$ and $H=(V, E_H, c_H)$ be some $j$-tree.
    Let $C_H=(C \subseteq V, E_C, c_C)$ be the core of $H$.
    For any cut $S \subseteq V$ of $H$, we have
    \begin{align*}
        c_H(\Pi(S \cap C)) \le c_H(S).
    \end{align*}
    That is, to find minimum cut separating core vertices in $H$, it suffice to check only cuts in the core and then construct the core cut.
\end{lemma}
\begin{proof}
    Let $E_F \coloneqq E_H \setminus E_C$ be the edge set of the envelope.
    Note that
    \begin{enumerate}
        \item $E(S) \cap E_C = E(\Pi(S \cap C)) \cap E_C$.
              Crossing edges in the core are identical for both cuts.
        \item $E(\Pi(S \cap C)) \cap E_F = \phi$.
              There is no crossing edges in the envelope for cut $\Pi(S \cap C)$.
    \end{enumerate}
    Conclusively, we have
    \begin{align*}
        c_H(\Pi(S \cap C)) &= c_H(E(\Pi(S \cap C)) \cap E_C) + c_H(E(\Pi(S \cap C)) \cap E_F) \\
        &= c_H(E(S) \cap E_C) + 0 \\
        &\le c_H(E(S) \cap E_C) + c_H(E(S) \cap E_F) = c_H(S).
    \end{align*}
\end{proof}
Then to make this dynamic, we just need to support the `add terminal' operation as defined in the local sparsifiers paper~\cite{GoranciHS18}, and in Chapter 6 of~\cite{Goranci19:thesis}.

\subsection{Dynamic Cut Sparsifier}

\begin{lemma}~\cite{AbrahamDKKP16}
\label{lemma:DynamicCutSparsifier}
    Given a graph $G=(V, E, c)$ with weight ratio $U$.
    There is an $(1 + \epsilon)$-cut sparsifier $H$ of $G$ w.h.p., Such $H$ supports the following operations:
    \begin{itemize}
        \item $\texttt{Insert}(u, v, c)$:
              Insert the edge $(u, v)$ to $G$ in amortized $O(\log^5{n}\epsilon^{-2}\log{U})$-time.
        \item $\texttt{Delete}(e)$:
              Delete the edge $e$ from $G$ in amortized $O(\log^5{n}\epsilon^{-2}\log{U})$-time.
    \end{itemize}
    Such $H$ is a subgraph of $G$ with different weight and the weight ratio of $H$ is $O(nU)$.
    Moreover, we maintain a partition of $H$ into $k=O(\log^3{n}\epsilon^{-2}\log{U})$ disjoint forests $T_1, \ldots, T_k$ with each vertex keeps the set of its neighbors u in each forest $T_i$.
    After each edge insertion/deletion in $G$, at most 1 edge change occurs in each forest $T_i$.
\end{lemma}
By first deploying a dynamic cut sparsifier from Lemma~\ref{lemma:DynamicCutSparsifier}, we can speed-up Theorem~\ref{theorem:DynamicMinCut} by ripping of the dependency on $m$.

\begin{proof}[Proof of Corollary~\ref{corollary:DynamicMinCut}]
    First apply Lemma~\ref{lemma:DynamicCutSparsifier} to acquire a 2-approxmation dynamic cut sparsifier $H$ of $G$.
    Such $H$ has $O(n\log^4{n})$ edges.
    Then we incur Theorem~\ref{theorem:DynamicMinCut} to maintain the sparsifier $H$.
    For every edge update in $G$, by Lemma~\ref{lemma:DynamicCutSparsifier}, it becomes $O(\log^3{n}\log{U})=O(\log^4{n})$ edge changes in $H$.
    Therefore, one single edge update can be handled in amortized time 
    \begin{align*}
        O\left(\log^4{n} \cdot (n\log^4{n})^{2/3}\log^{7}{n}\right) = O\left(n^{2/3}\log^{41/3}{n}\right).
    \end{align*}
    For every $\mathrm{MinCut}(s, t)$ query, we incur the same query on the sparsifier $H$, which can be computed in $\Tilde{O}((n\log^4{n})^{2/3})=\Tilde{O}(n^{2/3})$ -time.
    Since such $H$ preserves cut/flow value up to a factor of $2$, so the result is still within $\Otil(\log{n})$ approximation with high probability.
\end{proof}

\subsection{An $\Tilde{O}(m)$-time 2-Approximation Max Flow/Min Cut Solver}
Here we present a near-linear time max flow algorithm from ~\cite{Peng16}.

\begin{lemma}(~\cite{Peng16}, rephrased)
\label{lemma:flow_solver}
    Given a graph $G=(V, E, c)$ with weight ratio $U=\poly(n)$ and source/sink pair $s$ and $t$.
    We can compute a 2-approximation for value of the minimum cut between $s$ and $t$ in $O(m\log^{32}{n}\log^2{\log{n}})=\Tilde{O}(m)$-time . The cut set $S$ can be obtained on demand with linear overhead in $|S|$.
\end{lemma}

\subsection{The main theorem}
In this section, we give details in building a dynamic data structure for approximately computing minimum $st$-cuts in a dynamic graph.
The high-level idea is to build $(k, \rho, j \coloneqq O(m^{2/3}))$-decomposition of the original graph and dynamically maintain them.
By lemma~\ref{lemma:jtree_preserve_all_min_st_cut}, we maintain only $O(\log{n})$ of these $j$-trees instead of $k$.
For every $st$-cut query, we ran the algorithm from lemma~\ref{lemma:flow_solver}~\cite{Peng16} on these $O(\log{n})$ core graphs.

To dynamically maintain these core graphs, dynamic cut sparsifiers from lemma~\ref{lemma:DynamicCutSparsifier} are used.
In addition to that, we also present a data structure for maintaining $j$-tree under the operation of adding a vertex to the core.

To prove the theorem, we need a dynamic data structure for maintaining $j$-trees.
The tools are formalized as the following lemma.

\begin{lemma} \label{lemma:dynamic_JTree}
    Given $j \ge 1$ and a graph $G=(V, E, c)$ with weight ratio $U=\poly(n)$ and a $j$-tree $H$ of $G$ such that $H \preceq G \preceq_{\alpha} H$.
    Let $n=|V|, m=|E|$.
    We can dynamically maintain a $O(j)$-tree $\Tilde{H}$ such that $\Tilde{H} \preceq G \preceq_{O(\alpha)} \Tilde{H}$ under up to $j$ of following operations:
    \begin{enumerate}
        \item $\textsc{Initialize}(G)$:
              Build data structures for maintaining $H$ in $O(\frac{mn}{j}\log{n})$-time.
        \item $\textsc{AddTerminal}(u)$:
              Move vertex $u$ to the core of $H$.
              Such operation can be done in amortized $O(\frac{mn}{j^2}\log^6{n})$-time.
        \item $\textsc{Insert}(u, v, c)$:
              Insert the edge $(u, v)$ to $G$ in amortized $O(\frac{mn}{j^2}\log^6{n})$-time.
        \item $\textsc{Delete}(e)$:
              Delete the edge $e$ from $G$ in amortized $O(\frac{mn}{j^2}\log^6{n})$-time.
    \end{enumerate}
    The total number of edge change in the core is $O(\frac{mn}{j})$.
    Hence the amortized number of edge changes per operation is $O(\frac{mn}{j^2})$.
    Also, $C(\Htil)$ (core of $\Htil$) is sparse, i.e., it has $O(j\log^4{j})$ edges.
\end{lemma}

\subsection{Tree-terminal path and edge moving}
To prove Lemma~\ref{lemma:dynamic_JTree}, we have to open the black box of $j$-tree construction.
It creates a graph by first select a proper spanning tree/forest and then route off-tree edges by tree paths and set of edges restricted on a small subset of vertices.
To well-understand and formalize the construction, we introduce some notations about spanning forests and trees.

\begin{definition}
\label{definition:SkeletonTree}
    Given a forest $F$ and a subset of vertices $C \subseteq V(F)$.
    Add at most $|C|$ vertices to $C$ so that every pairwise lowest common ancestor is in $C$.
    Then iteratively remove vertices from $V(F) \setminus C$ of degree 1 until no such vertices remain.
    For each path with endpoints in $C$ and no internal vertices in $C$, replace the whole path with a single edge.
    We define the resulting forest as \emph{Skeleton Tree} of $F$ with respect to $C$, denoted by $S(F, C)$.
\end{definition}

\begin{definition}
\label{definition:TreePath}
    Given a forest $F$, define $F[u, v]$ as the unique $uv$-path in $F$ if they are connected.
    Given any edge $e=uv$, we use $F_e$ to denote the path $F[u,v]$.
\end{definition}

\begin{definition}
\label{definition:TreePartition}
    Given a forest $T$ and a subset of vertices $C$, a subset of edges $F \subseteq T$ is a \emph{tree partition} of $T$ with respect to $C$ if every component of $T \setminus F$ has exactly one vertex in $C$.
    For every vertex $u$ in $T$, we define $u$'s \emph{representative} with respect to a \emph{tree partition} $F$ and $C$ as the only vertex of $C$ in the component containing $u$ in $T \setminus F$.
    Denoted as $T_{C, F}(u)$.
    
    For any edge $e=uv \in T$, we define $e$'s \emph{tree-representative moving} as
    \begin{align*}
        \mathrm{Repr}_{T,C,F}(e=uv) \coloneqq
        \begin{cases}
            e, &\text{for } e \in T \setminus F \\
            T_{C,F}(u)T_{C,F}(v), &\text{for } e \in F
        \end{cases}.
    \end{align*}
    Use this $\mathrm{Repr}_{T,C,F}$, we define $e$'s \emph{tree-representative path} as
        \begin{align*}
            Q_{T,C,F}(e=uv) \coloneqq
            \begin{cases}
                e, &\text{for } e \in T \setminus F \\
                T[u, T_{C,F}(u)] + \mathrm{Repr}_{T,C,F}(e) + T[T_{C,F}(v), v], &\text{for } e \in F
            \end{cases}.
    \end{align*}
\end{definition}

Here we introduce notations from \cite{KPSW19}, which defines the so-called tree-portal paths.
Portals are terminals in our terminology.

\begin{definition}
\label{definition:TreeTerminalEdgeMoving}
    Given a graph $G=(V,E)$, a spanning forest $T$ and a subset of vertices $C$ (terminals).
    For any 2 vertices $u, v \in V$, define $T_C(u, v)$ to be the vertex in $C$ closest to $u$ in $T[u, v]$.
    If no such vertex exists, $T_C(u, v) \coloneqq \perp$.

    For any edge $e=uv \in E \setminus T$, first, we can orient arbitrarily.
    Then $e$'s \emph{tree-terminal edge moving} can be defined as \begin{align*}
        \mathrm{Move}_{T, C}(e=uv) \coloneqq
        \begin{cases}
            vv, &\text{for } T_C(u,v) = \perp \\
            T_C(u,v)T_C(v,u), &\text{otherwise} \\
        \end{cases}.
    \end{align*}
    Use this $\mathrm{Move}_{T, C}$, we can define $e$'s \emph{tree-terminal path} as
    \begin{align*}
        P_{T, C}(e=uv) \coloneqq
        \begin{cases}
            T[u, v] + \mathrm{Move}_{T, C}(e), &\text{for } T_C(u,v) = \perp \\
            T[u, T_C(u, v)] + \mathrm{Move}_{T, C}(e) + T[T_C(v, u), v], &\text{otherwise} \\
        \end{cases}.
    \end{align*}
\end{definition}

\subsection{Initializing a $J$-Tree}
In this subsection, we review the static construction of a $j$-tree.
It is summarized in Algorithm~\ref{algo:JTreeInitialize}.
For the dynamical purpose, we slightly modify the construction from ~\cite{Madry10}.

Briefly, the procedure first computes (1) $T$, some spanning tree of $G$, (2) $C$, an $O(j)$-sized subset of vertices, and (3) $F$, a subset of edges in $T$.
Then a $O(j)$-tree, $H$, is constructed by moving endpoints into $C$ for edges not in forest $T \setminus F$.
The construction is summarized in Algorithm~\ref{algo:RouteWithTCF}.
Hence, it is easy to see that the core graph of $H$ is $H[C]$, the subgraph induced by $C$.
Such moving is defined using either $\mathrm{Repr}_{T,C,F}(e)$ for $F$ or $\mathrm{Move}_{T, C}(e)$ otherwise.

Such edge moving corresponds to an embedding of $G$ into $H$.
Each edge of $G$ is routed in $H$ using either one tree path or 2 tree paths concatenated by an edge in the core.

$T$, $C$ and $F$ are computed by Algorithm~\ref{algo:JTreeInitializeTCF}.
First, we try to embed $G$ into some spanning tree $T$ of $G$ by routing each edge $e$ of $G$ using the unique tree path $T_e$.
Heuristically, to minimize the congestion incurred in each tree edge, a low stretch spanning tree (LSST)~\cite{ABN08} is used as $T$.
LSST guarantees low "total" congestion on tree edges.

But to ensure low congestion on "every edge", we remove tree edges with the highest congestion (relative to its capacity) and route impacted edges alternatively.
The removed tree edges are collected as set $F$ and endpoints of them are collected as set $C$.
Ideally, we move every edge not in $T$ using $\mathrm{Move}_{T, C}(e)$.
And for data structural purposes, we add $O(j)$ more vertices to $C$ to make sure we route each edge using a short tree path, i.e., of size $O(n/j)$.
As discussed in \cite{Madry10}, such edge moving does not guarantee a $j$-tree.

Identical to \cite{Madry10}, we add vertices in $S(T, C)$, skeleton tree of $C$, to $C$.
And add $O(|C|)$ more edges to $F$ so that $F$ is a tree partition of $T$ with respect to the new $C$.

The main difference from \cite{Madry10} is that we add more terminal vertices ($C$) and route off-tree edges after we determine $C$.
An argument from ~\cite{GhaffariKKLP18} shows that the more terminal we add, the better the congestion approximation.

\begin{algorithm2e}
    \caption{$\textsc{ComputeTCF}(G=(V, E, c), j, l: E \to \Real_{\ge 0})$}
    \label{algo:JTreeInitializeTCF}
    Compute a spanning tree $T$ of $G$ with average stretch $\Otil(\log{n})$ with respect to $l$ in $\Tilde{O}(m)$-time by \cite{ABN08}. \\

    For each edge $e=uv \in E$, let $f^e$ be the $uv$-flow in $T$ that routes $c(e)$ units of flow along $T_e$.
    Let $\mathbf{f}$ be the collection of $f^e$ for every $e \in E$.
    $\mathbf{f}$ is therefore a embedding of $G$ into $T$.
    $|\mathbf{f}| \in \Real^{E(T)}$, the vector of amount of flow crossing each edge of $T$, can be computed in $\Tilde{O}(m)$-time. \\
    
    For $e \in T$, define its relative loading, $rload(e) \coloneqq |\mathbf{f}(e)|/c(e)$. \\
    
    Decompose $E(T)$ into $O(\log{n})$ subsets $F_i, i \in \{1, \ldots, \lceil\log{\norm{f}_\infty} + 1 \rceil\}$, for $e \in F_i$ if $rload(e) \in (R/2^i, R/2^{i-1}]$ where $R = \max_{e \in T}{rload(e)}$. \\
    
    Let $i_0$ be the minimal index such that $|F_{i_0}| = \Omega(j/\log{n})$.
    Define $F = \bigcup_{i=1}^{i_0}{F_i}$.
    Note that $|F| \le j$ and contains edges with the largest relative load. \\
    
    Define $C$ consists of terminal vertices and all endpoints of edges of $F$. \\
    
    Add $O(j)$ more vertices to $C$ such that every path of length $O(n/j)$ on $T$ contains at least one vertex in $C$.\\
    
    Add vertices appeared in $S(T, C)$ to $C$ as well. \\
    
    For every adjacent $uv \in E(S(T, C))$, add the edge $e \in T[u, v]$ with largest $rload(e)$ to $F$.
    Note that endpoints of such an edge are not added to $C$. \\
    
    By the construction of $F$, we know every tree in $T \setminus F$ contains exactly one vertex in $C$. \\
    
    \Return $(T, C, F)$\\
\end{algorithm2e}

\begin{algorithm2e}
\label{algo:RouteWithTCF}
\caption{$\textsc{Route}(G=(V, E, c), T, C, F)$}
    $E_{\mathrm{tree}} \coloneqq \{\mathrm{Repr}_{T, C, F}(e) \mid e \in T\}.$ \\
    $E_{\mathrm{off-tree}} \coloneqq \{\mathrm{Move}_{T, C}(e) \mid e \not\in T\}.$ \\
    Define an embedding $\mathbf{f}$ of $G$ into $H$ as follows.\\
    For every $e \in T$, $f^e$ routes $c(e)$ units through path $Q_{T,C,F}(e)$.\\
    For every $e \not\in T$, $f^e$ routes $c(e)$ units through path $P_{T,C}(e)$.\\
    Define $c_H(e) \coloneqq |\mathbf{f}(e)|$, the amount of flow crossing $e$ in the embedding.\\
    $H \coloneqq (V, E_H \coloneqq E_{\mathrm{tree}} \cup E_{\mathrm{off-tree}}, c_H).$ \\
    Remove self-loops from $H$. \\
    \Return $H$. \\
    \tcp{$H$ is a $|C|$-tree with core $C(H) = H[C]$, the subgraph induced by subset of vertices $C$.}
    \tcp{The embedding $\mathbf{f}$ of $G$ into $H$ is referred as the \emph{canonical $j$-tree embedding}.}
\end{algorithm2e}

\begin{algorithm2e}
\label{algo:JTreeInitialize}
\caption{$\textsc{JTree}(G=(V, E, c), j, l: E \to \Real_{\ge 0})$}
    $(T, C, F) \coloneqq \textsc{ComputeTCF}(G, j, l)$. \\
    $H \coloneqq \textsc{Route}(G, T, C, F)$. \\
    \Return $H$. \\
\end{algorithm2e}

\subsection{Data Structure for dynamical maintenance}

Here we present the data structure for maintaining a $j$-tree, i.e., proving Lemma~\ref{lemma:dynamic_JTree}.

\subsubsection{Structural arguments for $j$-tree maintenance}

To prove Lemma~\ref{lemma:dynamic_JTree}, one has to make sure adding terminals does not increase the congestion.
The argument is formalized as the following lemma:
\begin{lemma}
\label{lemma:AddTerminalPreserveCongestion}
    Given a graph $G=(V,E,c)$, a spanning forest $T$ of $G$, a subset of vertices $C$ and $F$, a tree partition of $T$ with respect to $C$.
    For any vertex $u \in V \setminus C$, there is an edge $e_u \in T \setminus F$ such that $F + e_u$ is tree partition of $T$ with respect to $C + u$ (every component of $T \setminus (F + e_u)$ has exactly one vertex in $C$).
    
    Furthermore, let $H \coloneqq \textsc{Route}(G, T, C, F)$.
    If $H \preceq_\alpha G$, then the graph $\overline{H} \coloneqq \textsc{Route}(G,T,C + u,F + e_u) \preceq_\alpha G$.
\end{lemma}

\begin{proof}
    Let $x = T_{C,F}(u)$, the only vertex in $C$ in $u$'s component in $T \setminus F$.
    Since edges in $T \setminus F$ appear in both $G$ and $H$, let $e_u$ be the edge with minimum $c_H(e)$ in $T[u, x]$.
    
    Clearly, $F + e_u$ is tree partition of $T$ with respect to $C + u$.
    Since we delete one edge in $u$'s component which is a tree, it is split into 2 components such that $x$ and $u$ are in different components.
    
    Observation from \cite{GhaffariKKLP18} that adding more vertices into $C$ does not increase congestion immediately gives us that, $\overline{H} \coloneqq \textsc{Route}(G,T,C + u,F + e_u) \preceq_\alpha G$.
\end{proof}
To maintain $O(j)$-tree $H$ under dynamic edge updates in $G$, we first add both endpoints of the updating edge to the terminal and then perform the edge update in $C(H)$, core of $H$, directly.
One has to make sure such behavior does not increase the congestion when routing $G$ in $H$.
The following lemma gives such promise:
\begin{lemma}
\label{lemma:CoreEdgeUpdatePreserveCongestion}
    Given a graph $G=(V,E,c)$, a spanning forest $T$ of $G$, a subset of vertices $C$ and $F$, a tree partition of $T$ with respect to $C$.
    Let $H \coloneqq \textsc{Route}(G, T, C, F)$, and $e=uv$ be any edge with $u,v \in C$ ($e$ might not be in $G$) with capacity $c_e$.
    First note that $G[C] \subseteq H[C]$.
    
    If $H \preceq_\alpha G$ holds via the \emph{canonical $j$-tree embedding}, then both $(H+e) \preceq_\alpha (G+e)$ and $(H-e) \preceq_\alpha (G-e)$ holds.
\end{lemma}
\begin{proof}
    Let $\mathbf{f}$ be the \emph{canonical $j$-tree embedding} of $G$ into $H$.
    Let $\mathbf{f}^+$ be an embedding of $G+e$ into $H+e$ defined by routing $e \in (G+e)$ using $e \in (H+e)$ and routing other edges using the one defined in $\mathbf{f}$.
    To bound congestion incurred in $H+e$ using $\mathbf{f}^+$, observe that $|\mathbf{f}^+(e_H)| = |\mathbf{f}(e_H)|$ for any $e_H \in H$ and $|\mathbf{f}^+(e)| = c_e$.
    The observation comes from the fact routing $e \in (G+e)$ only affects $e \in (H+e)$.
    Thus, $(H+e) \preceq_\alpha (G+e)$ holds.
    If $e \in G$, observe that $\mathbf{f}$ routes $e$ using only its counterpart in $H$.
    It is because $\mathbf{f}$ routes $e$ using the path $P_{T,C,F}(e)$ containing only $e \in H$.
    Thus, define $\mathbf{f}^-$, an embedding of $G-e$ into $H-e$, by routing any other edge than $e$ via $\mathbf{f}$.
    For any edge $e_H \in H-e$, we have $|\mathbf{f}^-(e_H)| = |\mathbf{f}(e_H)|$.
    Therefore, no edge has congestion increased and $(H-e) \preceq_\alpha (G-e)$ holds.
\end{proof}

By the above 2 lemmas, we can guarantee low congestion if we maintain the $j$-tree correctly.

We need the following dynamic tree data structure.

\begin{lemma}
\label{lemma:PathQueryLinkCutTree}
    Given a rooted forest $T$ edge weight $w:E(T) \to \Real$, there is a deterministic data structure $D(T)$ supports following operations in $O(\log{n})$ amortized time.
    \begin{enumerate}
        \item $\mathtt{root}(u)$: Return the root of the tree containing $u$.
        \item $\mathtt{makeRoot}(u)$: Make $u$ as the root of the tree containing it.
        \item $\mathtt{pathMax}(u, v)$: Return the edge with maximum weight in the unique $uv$ path. Or $-\infty$ if $u, v$ are not connected in $T$.
        \item $\mathtt{cut}(e)$: Remove the edge $e$ from $T$.
        \item $\mathtt{link}(u, v, c)$: Add a new edge $e=uv$ with weight $c$. It is guaranteed that no cycle is formed after adding this new edge.
    \end{enumerate}
\end{lemma}

\begin{lemma}
\label{lemma:DynamicSkeletonTree}
    Given a rooted forest $T$ and a subset of vertices $C \subseteq V(T)$, there is a deterministic data structure $S(T)$ maintaining $S(T, C)$, the skeleton tree of $T$ with respect to $C$, under following operations in $O(\log{n})$ amortized time.
    \begin{enumerate}
        \item $\mathtt{AddC}(u)$:
            Return the set $V(S(F, C \cup \{u\})) \setminus V(S(F, C))$, which has size at most 2.
            Then add $u$ to $C$.
        \item $\mathtt{Neighbor}(u \in C)$:
            Return the neighboring vertices of $u$ in $S(F, C)$.
    \end{enumerate}
\end{lemma}

\subsubsection{Proof sketch of Lemma~\ref{lemma:dynamic_JTree}} 
Intuitively, we maintain the $j$-tree $H$ by mimicking the static procedure.
To make the resulting graph sparse, a dynamic cut sparsifier is used for the core.
Worth noticing, we maintain both the whole $j$-tree $H$ and the one with sparsified core, $\Htil$.
The reason for not applying sparsifier to the whole graph is because edges not in the core form a forest.
And by Lemma~\ref{lemma:only_core_cuts_matter}, we only care cuts in the core graph.

The most important part of our data structure is to support the $\textsc{AddTerminal}$ operation.
Initially, the $j$-tree $H$ is constructed by $\textsc{Route}(G,T,C,F)$.
When adding some vertex $u$ to $C$, we have to (1) find the edge $e_u$ in $T \setminus F$ and (2) update $H$ as $\textsc{Route}(G,T,C + u,F + e_u)$.

Thus, We maintain the following data structures:
\begin{enumerate}
    \item A dynamic 2-cut sparsifier $\Ctil(H)$ from Lemma~\ref{lemma:DynamicCutSparsifier} for maintaining a sparsified core graph.
    \item A dynamic tree data structure $D(T)$ from Lemma~\ref{lemma:PathQueryLinkCutTree} for finding such $e_u$.
          $D(T)$ is also used to find $\textsc{Repr}_{T,C,F}(e_u)$, the corresponding edge of $e_u$ in the core.
    \item For every off-tree edge $e=uv$, maintain both $T[u, T_{uv}(C)]$ and $T[T_{vu}(C), v]$ walks using doubly linked list.
          Maintain $\WW$ as a collection of all such walks.
    \item For every $x \in C$, maintain a set $P(x)$ consisting of walks in $\WW$ ending up at $x$.
    \item For every vertex $u \in V$, maintain a set $\mathtt{RI}(u)$ consisting of walks in $\WW$ containing $u$.
\end{enumerate}
The last 3 data structure is for maintaining $\textsc{Move}_{T,C}(e), e \not\in T$ with $C$ increasing.

\subsubsection{Formal proof of Lemma~\ref{lemma:dynamic_JTree}} 

\begin{proof}[Proof of Lemma~\ref{lemma:dynamic_JTree}]
    We may assume the $j$-tree, $H$, is constructed using the static procedure stated previously.
    Recall that $T$ is the low-stretch spanning tree of $G$.
    $C$ is the set of terminal (vertices in the core) of $H$.
    $F$ is the set of tree edges chopped off.
    \begin{algorithm2e}
    \caption{$\textsc{Initialize}(G, H)$} \label{algo:DynamicJTreeInitialize}
        Let $T$ be the low stretch spanning tree used in constructing $H$. \\
        Initialize $D(T)$ for $T$ from Lemma~\ref{lemma:PathQueryLinkCutTree} \\
        Initialize $S(T)$ for $T$ from Lemma~\ref{lemma:DynamicSkeletonTree} \\
        Initialize a dynamic cut sparsifier from Lemma~\ref{lemma:DynamicCutSparsifier} for $C(H)$, say $\Ctil(H)$. \\
        Let $C$ be the initial terminal set, i.e., $V(C(H))$. \\
        
        $\WW \coloneqq \phi$. \\
        $\mathtt{RI} \coloneqq \phi$. \\
        $P \coloneqq \phi$. \\
        
        \For{$e=uv \in E(G) \setminus E(T)$}{ Let $w_u$ be the $T[u, T_{uv}(C)]$ walk. \\
            Add $w_u$ to $\mathtt{RI}(a)$ for every $a \in V(w_u)$. \\
            \If{$T_{uv}(C)$ exists}{
                Add $w_u$ to $P(T_{uv}(C))$. \\
            }
            Add $w_u$ to $\WW$. \\
            Same for the $w_v \coloneqq T[T_{vu}(C),v]$ walk. \\
        }
        \Return $(D(T), S(T), \Ctil(H), \WW, \mathtt{RI}, P)$
    \end{algorithm2e}
    The procedure for $\textsc{AddTerminal(u)}$ is presented in Algorithm~\ref{algo:JTreeAddTerminal}.
    When adding a vertex $u$ to the terminal set $C$, we first have to find a tree edge $e_u$ via $D(T)$.
    Then we update both $C$, the terminal set, and $F$, tree partition with respect to new $C$.
    As shown in Lemma~\ref{lemma:AddTerminalPreserveCongestion}, finding such $e_u$ reduces to a path query in a dynamic tree.
    After that, we have to update $H$ to $\textsc{Route}(G,T,C,F)$.
    $\textsc{Route}(G,T,C,F)$ maps edges in $T$ using $\textsc{Repr}_{T,C,F}$ and $\textsc{Move}_{T,C}$ otherwise.
    For edges in $T$ having different $\textsc{Repr}_{T,C,F}(e)$, they corresponds to edges incident to $u$'s component in $T \setminus (F + e_u)$.
    This step can be made efficient by only move the smaller part out and relabel if necessary.
    This ensures a $O(|T|\log{n}) = O(n\log{n})$ total time complexity.
    By amortizing them across $j$ operations, this step has amortized $O((n/j)\log{n})$-time.
    To update $\textsc{Move}_{T,C}(e)$ for edges $e=vw \not\in T$, we explicitly maintain both $T[v,T_C(v,w)]$ and $T[T_C(w,v),w]$ walks.
    Observe that $T_{C+u}(v, w)$ is either $T_C(v, w)$ or $u$ depending on whether $u \in T[v,T_C(v,w)]$.
    This observation tells us that $T[v,T_C(v,w)]$ only gets shorter with prefix unchanged.
    When adding $u$ to the terminal, we simply find all $T[v,T_C(v,w)]$ walks containing $u$ and shortcut them at $u$.
    Using doubly linked list and pointers, we can find these walks and shortcut them with $O(1)$ overhead.
    The total time complexity on maintaining $\textsc{Move}_{T,C}(e)$ can be bounded by the total length of $T[v,T_C(v,w)]$ walks.
    From the construction of $\textsc{JTree}$, we know every such $T[v,T_C(v,w)]$ has length $O(n/j)$.
    Hence the total time complexity is $O(mn/j)$, and amortized time complexity per operation is $O(mn/j^2)$.
    Note that each of these edge change incurs a $O(\log^6{n})$ for edge updates in $\Ctil(H)$, sparsifier of the core graph.
    Combining above bounds, we know $\textsc{AddTerminal(u)}$ has amortized time complexity of
    \begin{align*}
        O\left(\log^6{n}\left(\frac{mn}{j^2}+\frac{n}{j}\log{n}\right)\right) =
        O\left(\frac{mn}{j^2}\log^6{n}\right).
    \end{align*}
    
    \begin{algorithm2e}[!t]
    \caption{$\textsc{AddTerminal}(u)$}
    \label{algo:JTreeAddTerminal}
        Let $(D(T), S(T), \Ctil(H), \WW, \mathtt{RI}, P)$ be the data structures defined in Algorithm~\ref{algo:DynamicJTreeInitialize} \\
        $E^+ \coloneqq \phi, E^- \coloneqq \phi$.
        \tcp*{$E^+$ and $E^-$ are the sets of edges to be inserted or deleted in $C(H)$ respectively}
        $x \coloneqq D(T).\mathtt{root}(u)$
        \tcp*{$x$ is the vertex in $C$ in $u$'s component in $T \setminus F$.}
        $e_u \coloneqq D(T).\mathtt{pathMax}(t, u)$
        \tcp*{$e_u$ is the edge with highest $c_H(e)$ in $T[x, u]$}
        $D(T).\mathtt{cut}(e_u)$ \\
        $D(T).\mathtt{makeRoot}(u)$. \\
        \tcp{Update $C \coloneqq C + u$ and $F \coloneqq F + e_u$.} Add $(x, u, c_H(e_u))$ to $E^+$. \\
        Add edges in $T$ incident to $u$'s component in $T \setminus F$ to $E^-$. \\
        For edges visited in previous step, add them to $E^+$ by replacing 1 endpoint $x$ to $u$. \\
        \tcp{Maintain $\textsc{Move}_{T,C}(e)$ for edges $e \not\in T$.}
        \For{$W \in \WW$ such that $u \in W$}{
            Add the core edge for $W$ to $E^-$. \\
            Shortcut $W$ at $u$. \\
            Add the core edge for the shortened $W$ to $E^+$. \\
        }
        \tcp{Update edges in the core.}
        \For{$e \in E^+$}{
            $\Ctil(H).\mathtt{Insert}(e)$ \\
        }
        \For{$e \in E^-$}{
            $\Ctil(H).\mathtt{Delete}(e)$ \\
        }
        \Return $(D(T), S(T), \Ctil(H), \WW, \mathtt{RI}, P)$ \\
    \end{algorithm2e}
    
    \begin{algorithm2e}
    \caption{$\textsc{Insert}(u, v, c)$}
    \label{algo:JTreeEdgeInsert}
    $\mathrm{AddTerminal(u)}$ \\
    $\mathrm{AddTerminal(v)}$ \\
    $\Ctil(H).\mathtt{Insert}(u, v, c)$ \\
    \end{algorithm2e}
    
    \begin{algorithm2e}
    \caption{$\textsc{Delete}(e=uv)$}
    \label{algo:JTreeEdgeDelete}
    $\mathrm{AddTerminal(u)}$ \\
    $\mathrm{AddTerminal(v)}$ \\
    $\Ctil(H).\mathtt{Delete}(u, v, c(e))$ \\
    \end{algorithm2e}
    
    For $\textsc{Insert}(u, v, c)$ and $\textsc{Delete}(e)$, we first add both endpoints to terminal as presented in Algorithm~\ref{algo:JTreeEdgeInsert} and Algorithm~\ref{algo:JTreeEdgeDelete}.
    Then we directly insert/delete the interested edge from the core graph.
    The time complexity is occupied by the cost of adding terminal vertex.
    Since we correctly maintain a unsparsifier $O(j)$-tree $H$, we have $H \preceq G \preceq_\alpha H$ by Lemma~\ref{lemma:AddTerminalPreserveCongestion} and Lemma~\ref{lemma:CoreEdgeUpdatePreserveCongestion}.
    Also, by the decomposability of cut approximation and the core of $\Htil$ is a 2-cut sparsifier, we have $H \preceq G \preceq_{O(\alpha)} H$.
\end{proof}

\subsection{Put everything together}

\begin{proof}[Proof of Theorem~\ref{theorem:DynamicMinCut}]
    Let $j$ be some parameter determined later and $k = \Theta(\frac{m\log{U}\log^2{n}}{j})$.
    \begin{algorithm2e}[h]
    \caption{$\textsc{Initialize}(G)$}
    \label{algo:DynamicMinCut}
        $n \coloneqq |V(G)|$, $m \coloneqq |E(G)|$, $j \coloneqq m^{2/3}$, $k \coloneqq \Theta(\frac{m\log{U}\log^2{n}}{j})$, $t \coloneqq \Theta(\log{n})$. \\
        Let $\mathbf{G} = \{G_1, \ldots, G_k\}$ be a $(k, \Tilde{O}(\log{n}), \Theta(j))$-decomposition of $G$ by Lemma~\ref{lemma:static_jtree_decomposition}. \\
        Sample $t$ graphs with repetition from $\mathbf{G}$, say, $G_1, \ldots, G_t$. \\
        \For{$i = 1, \ldots, t$}{
            $D_i \coloneqq \mathrm{Initialize}(G, G_i)$ by Lemma~\ref{lemma:dynamic_JTree}. \\
        }
        \Return $\mathbf{D} \coloneqq \{D_1, \ldots, D_t\}$ \\
    \end{algorithm2e}
    Initialization of the data structure is summarized as Algorithm~\ref{algo:DynamicMinCut}.
    First apply Lemma~\ref{lemma:static_jtree_decomposition} to acquire a $(k, \Tilde{O}(\log{n}), \Theta(j))$-decomposition of $G$, say $G_1, \ldots, G_k$.
    Then we apply Lemma~\ref{lemma:jtree_preserve_all_min_st_cut} to sample $t=O(\log{n})$ of them, say $G_1, \ldots, G_t$.
    For each of $G_i$, we incur Lemma~\ref{lemma:dynamic_JTree} to build data structures for dynamical operations.
    Let $D_1, \ldots, D_t$ be the data structures for each of $G_1, \ldots, G_t$.
    2-approximated dynamic cut sparsifers from Lemma~\ref{lemma:DynamicCutSparsifier} is also built for the cores of $D_1, \ldots, D_t$.
    Note that each $D_i$ supports up to $j$ operations, we rebuild $G_1, \ldots, G_k$ and $D_1, \ldots, D_t$ every $j$ operations.
    To deal with the query $\mathrm{mincut}(s,t)$, we run the algorithm from Lemma~\ref{lemma:flow_solver} on each sparsified core of $D_1, \ldots, D_t$.
    The running time is $\Tilde{O}(t \times j) = \Tilde{O}(j)$.
    Among results, the one with the smallest cut value is returned.
    The correctness comes from Lemma~\ref{lemma:jtree_preserve_all_min_st_cut} and Lemma~\ref{lemma:only_core_cuts_matter} with high probability.
    The quality of the result is within $\Tilde{O}(\log{n})$-factor with the optimal solution.
    For edge updates, we propagate them to $D_1, \ldots, D_t$ in amortized time $O(t \cdot \frac{mn}{j^2}\log^6{n}) = O(\frac{mn}{j^2}\log^7{n})$.
    As guaranteed by Lemma~\ref{lemma:dynamic_JTree}, each operation corresponds to $O(\frac{mn}{j^2})$ changes to the core.
    Each of the edge change is handled by the cut sparsifier in $O(\log^6{n})$-time.
    So the update time is
    \begin{align*}
        O\left(\frac{mn}{j^2}\log^2{n} + t\cdot\frac{mn}{j^2}\log^6{n}\right) = O\left(\frac{mn}{j^2}\log^7{n}\right).
    \end{align*}
    The cost for rebuild consists of 2 parts, $O(km\log{m})$-time for building decomposition of $G$ and $O(tm\log^6{n})$-time for initializing $D_1, \ldots, D_t$ and cut sparsifiers for cores.
    By charging the cost among $j$ operations, the runtime cost charged with each operation is
    \begin{align*}
        O\left(\frac{km\log{n}+tm\log^6{n}}{j}\right) = O\left(\frac{m\log{U}\log^2{n} \cdot m\log{n}}{j^2}\right) = O\left(\frac{m^2}{j^2}\log^4{n}\right).
    \end{align*}
    To balance the query cost and update cost, $j$ is set to $m^{2/3}$.
    So time complexity per operation is now $\Tilde{O}(m^{2/3})$.
\end{proof}

\subsection{Dynamic Max-flow Against an Adaptive Adversary}

We next show how to modify our $j$-tree based data-structure to obtain a randomized dynamic algorithm that works against an adaptive adversary. 

\begin{theorem} \label{thm:dynamicMaxFlowAdaptive}
	Given a graph $G=(V,E,c)$ with polynomially bounded capacities, there is a dynamic data structure that maintains $G$ against an \emph{adaptive} adversary subject to the following operations:
	\begin{enumerate}
		\item \textsc{Insert}$(u,v,c)$: Insert the edge $(u,v)$ to $G$ in $\tilde{O}(m^{3/4})$ amortized time.
		\item \textsc{Delete}$(e)$: Delete the edge $e$ from $G$ in $\tilde{O}(m^{3/4})$ amortized time.
		\item \textsc{MinCut}$(s,t)$: Output an $\tilde{O}(\log n)$-approximation to the $st$-min-cut value of $G$ in $~O(m^{3/4})$ time w.h.p. 
		The cut set $S$ can be obtained on demand with linear overhead in $|S|$.
	\end{enumerate}
\end{theorem}

Our previous construction used the oblivious adversary assumption in two places. (1) First, when building a decomposition of the original graph into $O(j)$-trees, we sampled only a logarithmic number of them during the preprocessing phase and dynamically maintained these sampled graphs. Note that an adaptive adversary could use the query operation to reveal information about the random bits used by our algorithm and which graphs we sampled, and this is why we needed to assume that adversary is oblivious. To circumvent this assumption, we instead maintain all the $O(j)$-trees in the decomposition and sample a small number of them only when handling queries. (2) Second, the dynamic cut sparsifier from Lemma~\ref{lemma:DynamicCutSparsifier} works only against an oblivious adversary, so we need a dynamic cut sparsifier that works against an adaptive advesray. In fact, because we explicitly compute a sparsifier on the core vertices, it suffices to have a data structure that outputs the sparsifier in time proportional to the number of core vertices. This allows to use fresh random bits when sampling a sparsifier during the query operation. Such an algorithm can be inferred from previous literature by combining expander decomposition based construction of graph sparsifiers~\cite{SpielmanT11} with recent works on decremental maintenances of expanders~\cite{NanongkaiSW17,SaranurakW19}. Slightly more formally, given an $n$-vertex graph $G$, using the pruning procedure from~\cite{SaranurakW19}, we can maintain an expander decomposition under edge deletions and recurse on the edges between expander clusters. To handle edge insertions, we employ a well-known reduction from decremental to full-dynamic algorithms~(see e.g., Lemma~4.17 from~\cite{AbrahamDKKP16}) , which in turn leads to a fully-dynamic algorithm for maintaining a hierarchy of expander decompositions. Since cut/spectral sparsifiers are decomposable, and constructing them on expanders amounts to sampling $O(\log n \epsilon^{-2})$ random edges per vertex~\cite{SpielmanT11,PengS14}, it follows that constructing a sparsifier from the current hierarchy of expanders can be done in $\tilde{O}(n \epsilon^{-2})$ time. The above idea is explicitly implemented in the recent work by Bernstein et al.~\cite{BernsteinCutSpAdapt} and we formally state their result below.


\begin{lemma}\cite[Theorem~10.5]{BernsteinCutSpAdapt} \label{lem:adaptiveCutSparsifier}
	There exists a fully dynamic algorithm that maintains for any weighted graph with an $(1+\epsilon)$-approximate cut sparsifier against an \emph{adaptive} adversary. The algorithm’s pre-processing time is bounded by $O(m)$, amortized update time is $\tilde{O}(1)$ and query time is $\tilde{O}(n \epsilon^{-3} \log U)$. The query operation returns an $(1+\epsilon)$-approximate cut sparsifier of $G$.
\end{lemma}

We now explain the necessary modifications to our data-structure. Similarly to Algorithm~\ref{algo:DynamicMinCut}, given a graph $G$,  we compute a $(k,\tilde{O}(\log n), \Theta(j))$-decomposition $\mathbf{G} = \{G_1,\ldots G_k\}$ of $G$ using Lemma~\ref{lemma:static_jtree_decomposition}, where $k=\Theta\left(\frac{m \log^2 n \log U}{j}\right)$. We maintain each $G_i$ from $\mathbf{G}$ using the data-structure $D_i$ from  Lemma~\ref{lemma:dynamic_JTree}, where for each core in $G_i$ we maintain an adaptive dynamic cut sparsifier using Lemma~\ref{lem:adaptiveCutSparsifier} (recall that previously we sampled $O(\log n)$ $G_i$'s from $\mathbf{G}$ and maintained a dynamic sparsifier against an \emph{oblivious} adversary for each of them). These data-structures are rebuilt from scratch every $j$ operations. Upon receiving an edge insertion or deletion, we pass the corresponding update to each $D_i$. When receiving an $st$-min cut query, we first add $s$ and $t$ to the core of each $G_i$ and then sample $G_1,\ldots,G_{t}$ with repetition from $\mathbf{G}$, where $t= \Theta(\log n)$. For each $i=1,\ldots, t$, we construct a cut sparsifier for the core of $G_i$ using the query operation from Lemma~\ref{lem:adaptiveCutSparsifier}. On each sparsified core of $G_i$'s we compute an $st$-min cut from scratch and then return the smallest value among those min cuts as an estimate. Note that sampling $\Theta(\log n)$ graphs whenever we receive a query ensures that the adversary cannot learn anything useful about our algorithm. 

\begin{proof}[Proof of Lemma~\ref{thm:dynamicMaxFlowAdaptive}]
The correctness proof is exactly the same as in Theorem~\ref{theorem:DynamicMinCut}. We next study the running time. The preprocessing cost consists of (1) the cost for computing the decomposition $\mathbf{G}$ and (2) and the cost for initializing the data-structure $D_1,\ldots, D_k$. By Lemma~\ref{lemma:static_jtree_decomposition}, (1) is bounded by $\tilde{O}(km)$ while (2) is bounded by $\tilde{O}(kmn/j)$ by Lemma~\ref{lemma:dynamic_JTree}. Since we rebuild our data-structure from scratch every $j$ operations, the cost of the rebuild charged to each operation is
\[
	\tilde{O}\left(\frac{km + kmn/j}{j}\right) = \tilde{O}\left(\left(\frac{m^2}{j^2} + \frac{m^2 n}{j^3}\right) \log U \right) = \tilde{O}\left(\frac{m^3 \log U}{j^3}\right),
\]
where the last inequality uses that $j \leq m$.

Next, by Lemma~\ref{lemma:dynamic_JTree}, the amortized time to support an edge insertion or deletion in $D_i$ is $\tilde{O}(mn/j^2)$. Since we maintain $k$ different $D_i$'s, it follows that the amortized time per edge insertion or deletion is bounded by
\[
	\tilde{O}\left(k \cdot \frac{mn}{j^2}\right) = \tilde{O}\left( \frac{m^2n \log U}{j^3}\right).
\]
Combining the above bounds, it follows that the amortized update time is $\tilde{O}\left(\frac{m^3 \log U}{j^3}\right)$.

Up to a logarithmic factor, the query cost is dominated by (1) the time to construct a cut sparsifier for the core and (2) the time to compute an $st$-min cut on a graph of size $\tilde{O}(j)$. As both can be implemented in $\tilde{O}(j)$ time, it follows that the query time is also $\tilde{O}(j)$. To balance the update and query time, we set $j=m^{3/4}$, which proves the lemma.
\end{proof}

\section{Fully-Dynamic All-Pairs Shortest Paths}




In this section, we once again demonstrate the power of \emph{Fully-Dynamic Vertex Sparsifier} by designing a $\Otil(\log{n})$-approximate dynamic APSP oracle with sublinear update and query time.
The main result of this section is formalized as the following theorem:

\begin{theorem}
\label{theorem:DynamicAPSP}
    Given a graph $G=(V,E,l)$, we have a fully-dynamic data structure that maintains all pair distance up to $\Otil(\log{n})$-factor and supports following operations:
    \begin{enumerate}
        \item $\textsc{Insert}(u, v, c)$:
              Insert the edge $(u, v)$ to $G$ in amortized $O(m^{2/3}\log^4{n})$-time.
        \item $\textsc{Delete}(e)$:
              Delete the edge $e$ from $G$ in amortized $O(m^{2/3}\log^4{n})$-time.
        \item $\textsc{Distance}(s, t)$:
              Output a $\Tilde{O}(\log{n})$-approximation to the $st$-distance value of $G$ in $O(m^{2/3+o(1)})$-time w.h.p..
    \end{enumerate}
\end{theorem}

Apply the dynamic spanner from \cite{FG19} on the input graph $G$, we can reduce the number of edges from $m$ to $O(n^{1+o(1)})$ while preserving distance up to $O(1)$-factor.
\begin{corollary}
\label{corollary:DynamicAPSP}
    Given a graph $G=(V,E,l)$, we have a fully-dynamic data structure that maintains all pair distance up to $\Otil(\log{n})$-factor and supports following operations:
    \begin{enumerate}
        \item $\textsc{Insert}(u, v, c)$:
              Insert the edge $(u, v)$ to $G$ in amortized $O(n^{2/3+o(1)})$-time.
        \item $\textsc{Delete}(e)$:
              Delete the edge $e$ from $G$ in amortized $O(n^{2/3+o(1)})$-time.
        \item $\textsc{Distance}(s, t)$:
              Output a $\Tilde{O}(\log{n})$-approximation to the $st$-distance value of $G$ in $O(n^{2/3+o(1)})$-time w.h.p..
    \end{enumerate}
\end{corollary}


\subsection{Path, Distance and $L_1$-embeddability}



In this subsection, we present concepts regarding preserving distance.

\begin{definition}
\label{def:routing}
Given 2 graphs on the same vertex set $G=(V,E,l), H=(V,E_H,l_H)$.
A routing of $H$ into $G$ is a set of paths $\CP = \{P^e \text{ an uv-path in }H\mid e=uv \in E\}$.
We say $G$ is $t$-routable in $H$ (or $H$ $t$-routes $G$), denoted by $H \preceq^1_{t} G$, if for every edge $e \in G$, $l_H(P^e) \le t \cdot l(e)$ holds.
The subscript is often omitted when $t=1$.

When the routing $\CP$ is clear in the context, we often write $l_H(e)$ to denote $l_H(P^e)$.
\end{definition}

\begin{lemma}
Given 2 graphs on the same vertex set $G=(V,E,l), H=(V,E_H,l_H)$.
If $G$ is $t$-routable in $H$, $H \preceq^1_\beta G$, then $d_H(s,t) \le \beta \cdot d_G(s,t), \forall s, t \in V$.
\end{lemma}

\begin{corollary}
If $G \preceq^1 H \preceq^1_\beta G$, we have
\[
d_G(s,t) \le d_H(s,t) \le \beta \cdot d_G(s,t), \forall s, t \in V.
\]
\end{corollary}

\begin{corollary}
If $H \subseteq G$, we have $G \preceq^1 H$.
\end{corollary}

\subsection{Metric-$J$-Trees as Vertex Sparsifier}

Here we introduce the notion of \emph{Metric-$J$-Tree}, whose name adopts from $j$-tree in \cite{Madry10}.
Given a graph, a \emph{Metric-$J$-Tree} is built from a spanning tree with additional $O(j)$ edges.
In a graph of this family, distance computation can be speed up by transforming such graph into one with much less vertices.
This property is exactly what we need for vertex sparsifier.

\begin{definition}
\label{def:J1}
Given a graph $G=(V,E,l)$, positive integer $j$, a subset of edges $F$ of size at most $j$, and a spanning tree $T$ of $G$.
The \emph{Metric-$j$-tree} of $G$ with respect to $T$ and $F$, denoted by $J^1_G(T,F)$, is a subgraph of $G$ with vertex set $V$ and edge set
\[
E_H = E(T) \cup F.
\]
That is, $J^1_G(T,F)$ keeps only spanning tree $T$ and edges in $F$.
Also, when we speak of routing of $J^1_G(T,F)$ into $G$, we route edge $e=uv$ not in $F$ or $T$ via $T_e$, the unique $uv$-path in $T$.
And route other edges using their identical counterpart.
\end{definition}

Such structure is interested because of the following theorem.
It suggest that we can approximate distance within several $J^1_G(T,F)$.
First, we define metric-decomposition of $G$.

\begin{definition}
\label{defn:MetricDecomp}
Given a graph $G=(V,E,l)$ and a family of graphs $\CG$.
A collection of graphs $H_1, \ldots, H_k$ and $k$ real numbers $\lambda_1, \ldots, \lambda_k$ is a \emph{$(k, \rho, \CG)$-metric-decomposition of $G$} if
\begin{enumerate}
    \item $\lambda_i \ge 0, \forall i$ and $\sum_i\lambda_i = 1$.
    \item $H_i \in \CG, \forall i$.
    \item $G \preceq^1 H_i, \forall i$.
    \item $\sum_i\lambda_iH_i \preceq^1_\rho G$.
\end{enumerate}
When $\CG$ is the family of all metric-$j$-tree of $G$, we denote $(k, \rho, \CG)$-metric-decomposition of $G$ by \emph{$(k, \rho, j)$-metric-decomposition of $G$}.
\end{definition}

Such decomposition can be computed via MWU-like method, just like \cite{Madry10}.
It is formalized as the following:

\begin{theorem}
\label{thm:J1MetricDecomp}
Let $G=(V,E,l)$ be a graph with weight ratio $U=\poly(n)$.
In $O(km\log^4{n})$-time, we can find a
\[
    \left(k, \alpha = \Otil(\log{n}), O\left(\frac{m\log^3{n}}{k}\right)\right)
\]-metric-decomposition of $G$.
Or conversely, given positive integer $j$, we can compute a $(O((m/j)\log^3{n}),\Otil(\log{n}),j)$-metric-decomposition of $G$ in $O((m^2/j)\log^7{n})$ time.
\end{theorem}

The proof of Theorem~\ref{thm:J1MetricDecomp} uses same MWU-like approach as \cite{Madry10}.
Briefly, we compute one component of the decomposition one at a time.
In addition to edge length, we also incur edge weight on the graph.
Intuitively, such edge weight regularizes stretch of over-stretched edges in the current decomposition.

We formally define stretch and how this edge weight interact with the graph.
For a given edge weight function $\bw$ on $G$, and a graph $H=(V,E_H,l_H)$ that routes $G$.
Define the \emph{volume $\bw(H)$ of $H$ (with respect to $\bw$)} to be $\bw(H) \coloneqq \sum_{e \in G}{\bw(e)l_H(e)}$.
Also, for any edge $e \in G$, define the \emph{stretch of $e$ in $H$} to be $\eta_H(e) \coloneqq \frac{l_H(e)}{l(e)}$ and denote by $\eta(H)$ the maximum value of $\eta_H(e)$, i.e. $\eta(H) \coloneqq \max_{e \in G}{\eta_H(e)}$.
Furthermore, define a set $\psi(H)$ as
\[
    \psi(H) \coloneqq \{e \in G \mid 0.5 \eta(H) \le \eta_H(e)\}.
\]
Intuitively, the set $\psi(H)$ contains edges of $G$ that suffer stretch at least as the half of the maximum stretch.

Given these definitions, the MWU method is formally stated as the following with proof deferred in the Appendix~\ref{sec:J1MetricDecompMWUProof}:

\begin{restatable}{lemma}{KTree}
\label{lemma:J1MetricDecompMWU}
Let $\alpha \geq \log{m}$ and a family of graphs $\CG$ such that for any edge weight function $\bw$ on $G$, we can find in $O(f(m))$ time a subgraph $H_\bw = (V, E_{H_\bw}, l_{H_\bw})$ of $G$ that belongs to $\CG$ and:
\begin{enumerate}
\label{MWUCondition}
    \item\label{MWU1} $\bw(H_\bw) \le \alpha \bw(G)$,
    \item\label{MWU2} $G \preceq^1 H_\bw$, and
    \item\label{MWU3} $|\psi(H_\bw)| \geq \frac{4\alpha m}{k}$
\end{enumerate}
then a $(k, 2\alpha, \CG)$-metric-decompostion of $G$ can be computed in $O(k \cdot f(m))$ time.
\end{restatable}

Next we will discuss how to construct $J^1_G(T,F)$ given edge weight function $\bw$ that satisfies all 3 conditions Theorem~\ref{lemma:J1MetricDecompMWU} needs.
First we present the following lemma that derives from low stretch spanning tree construction in \cite{ABN08}.

\begin{lemma}
\label{lemma:GeneralizedLSST}
Let $G=(V,E)$ be a graph with non-negative edge length $l$ and edge weight $\bw$.
We can find a spanning tree $T$ of $G$ (with edge length $l$) s.t.
\[
    l_{uv} \le d_T(u, v), \forall e=uv
\]
and
\[
    \sum_{e=uv}{d_T(u,v)\bw_e} \le \alpha \sum_{e=uv}{l_e\bw_e}
\]
for some $\alpha = O(\log{n}\log\log{n}) = \Otil(\log{n})$.
Such tree $T$ can be computed in $\Otil(m)$-time.
\end{lemma}

From now on, we use $\alpha$ to denote the ratio in Lemma~\ref{lemma:GeneralizedLSST}.

Given a graph $G=(V,E,l)$, positive integer $k$ for sparsity of \emph{metric-decomposition}, and a edge weight function $\bw$.
Also let $U=\poly(n)$ be the weight ratio of $G$.
The construction works like the following
\begin{enumerate}
    \item Compute LSST $T$ using Lemma~\ref{lemma:GeneralizedLSST} with respect to $l$ and $\bw$.
          Route $G$ using $T$, i.e., route each edge $e=uv$ using the unique $uv$-path in $T$.
    \item For every edge $e=uv \in G$, compute its stretch, $\eta(e) = d_T(u,v)/l_e$, in $O(m)$ time.
          Observe that $\eta(e) \le mU$.
    \item Partition edges of $G$ into $\log(mU)$ sets $F_1, F_2, \ldots F_{\log(mU)}$, with $F_i \coloneqq \{e \in G \mid 2^{i-1} \le \eta(e) < 2^i\}$.
          Define $F_{\ge i} = \bigcup_{j=i}^{\log{mU}}{F_j}$.
          Add edges with $\eta(e) < 1$ to $F_1$.
    \item \label{findOptimalJ}
          Find smallest $j^* \le \log(mU)$ such that
          \[
              |F_{\ge j^*}| \le \frac{4(2\alpha + 1)m\log{mU}}{k},
              |F_{\ge j^*-1}| > \frac{4(2\alpha + 1)m\log{mU}}{k}
          \]
          If no such $j^*$ exists, i.e., $|F_{\log(mU)}| > 4(2\alpha + 1)m\log{mU}/k$, let $F=\phi$ and output $H_\bw=J^1_G(T,\phi)=T$.
    \item By Pigeonhole principle, there is some $\bar{j}, j^* \le \bar{j} \le \log(mU)$ such that
          \[
              |F_{\bar{j} - 1}| \ge \frac{4(2\alpha + 1)m}{k}.
          \]
    \item Let $F$ be the set $F_{\ge \bar{j}}$ and output $H_\bw=J^1_G(T,F)$.
\end{enumerate}

For condition~\ref{MWU1} of Lemma~\ref{lemma:J1MetricDecompMWU}, we know $\bw(T) \le \alpha \bw(G)$.
Since $T$ is a subgraph of $H_\bw$, $d_{H_\bw}(u, v) \le d_T(u, v), \forall u, v$ holds for every pair of vertex $u, v$.
We have $\bw(H_\bw) \le \bw(T) \le \alpha \bw(G)$.

For condition~\ref{MWU2} of Theorem~\ref{lemma:J1MetricDecompMWU}, $G \preceq^1 H_\bw$ holds trivially since $H_\bw$ is a subgraph of $G$.

For condition~\ref{MWU3} of Theorem~\ref{lemma:J1MetricDecompMWU}, if the procedure ends at Step~\ref{findOptimalJ} and outputs $G$, we know $\psi(H_\bw) \supseteq F_{\log(mU)}$ and $|F_{\log(mU)}| > 4(2\alpha + 1)m\log{mU}/k$.
Thus,
\[
    |\psi(H_\bw)| \ge |F_{\log(mU)}| > \frac{4(2\alpha + 1)m\log{mU}}{k} \ge \frac{4\alpha m}{k}
\]
Otherwise, we observe that in $J^1_G(T,F_{\ge \bar{j}})$, edges in $F_{\ge \bar{j}}$ have stretch $1$ and therefore $\psi(H_\bw) \supseteq F_{\bar{j} - 1}$ and $|F_{\bar{j} - 1}| \ge 4(2\alpha + 1)m/k$.

Also notice that the set $F$ has size $\le 4(2\alpha + 1)m\log{mU}/k = O(m\log^3{n}/k)$.

The above procedure works in $\Otil(m)$-time which is occupied by the LSST construction from Lemma~\ref{lemma:GeneralizedLSST}.
It is summarized as the following lemma:
\begin{lemma}
\label{lemma:ComputeGoodJ1}
Given a graph $G=(V,E,w)$, positve integer $k$, and edge weight function $\bw$.
We can compute a spanning tree $T$ and a subset of edges $F$ in $\Otil(m)$ time such that
\begin{enumerate}
    \item $\bw(J^1_G(T, F)) \le \alpha \bw(G)$,
    \item $G \preceq^1 J^1_G(T, F)$,
    \item $|\psi(J^1_G(T, F))| \ge \frac{4\alpha m}{k}$, and
    \item $|F| = O(\frac{m\log^3{n}}{k})$.
\end{enumerate}
\end{lemma}

\begin{proof}[Proof of Theorem~\ref{thm:J1MetricDecomp}]
It comes directly from Lemma~\ref{lemma:J1MetricDecompMWU} and Lemma~\ref{lemma:ComputeGoodJ1}.
\end{proof}

\subsection{Tranform $J^1_G(T,F)$ into a vertex sparsifer}

In this section, we show how to construct vertex sparsifier, $H_C$, that preserve distance within interested terminal set $C$.
Let $H_C$ be an empty graph with vertex set $C$ initially.
The construction works as follows:
\begin{enumerate}
    \item Given $J^1_G(T, F)$, add endpoints of edges in $F$ to $C$.
    \item Add vertices in $S(T,C)$, degree 3 vertices in the Steiner tree of $C$ in $T$, to $C$ as well.
    \item Since vertices of $S(T,C)$ are in $C$ as well, for every edge $e=uv \in S(T,C)$, add edge $uv$ with length $l(T_e)$ to $H_C$.
          That is, edge $uv$ corresponds to the unique tree $uv$-path $T$.
    \item For any edge $e \in \not\in T$, add $\textsc{Move}_{T,C}(e)$ to with length $l(P_{T,C}(e))$ to $H_C$.
\end{enumerate}

Suppose $T, F, C$ are given, denote the construction of $H_C$ by $\textsc{Route}^1(G,T,C,F)$.

It works almost the same as the vertex sparsifier construction for the dynamic max flow problem.


Such $H_C$ is a preserve terminal-wise distance in $J^1_G(T, F)$, which is formalized and proved in the following lemma:

\begin{lemma}
\label{lemma:DistanceVertexSparsifier}
\[
\forall s, t \in C, d_{H_C}(s,t) \le d_{J^1_G(T, F)}(s, t).
\]
\end{lemma}
\begin{proof}

Let $P^*$ be the shortest $st$-path in $J^1_G(T, F)$.
We break $P^*$ into maximal segments of paths, $P_1, \ldots, P_k$, such that each of them intersects with $C$ only at endpoints.
By construction, each $P_i$ is either i) a tree path, or ii) an edge comes from $F$.
Both possibilities has their mapped edge in $H_C$.
Let $e_1, \ldots e_k$ be corresponding edges of $P_1, \ldots, P_k$.
Clearly, $e_1, \ldots e_k$ forms a $st$-path in $H_C$.
hence
\[
d_{J^1_G(T, F)}(s, t) = \sum_{i=1}^{k}{l_{J^1_G(T, F)}(P_i)} = \sum_{i=1}^{k}{l_{H_C}(e_i)} \ge d_{H_C}(s,t).
\]
\end{proof}

If $J^1_G(T, F)$ preserve distance of $G$ for vertex set $C$ within factor of $t$, so does $H_C$.
\begin{corollary}
\label{lemma:AlphaDistanceVertexSparsifier}
If $J^1_G(T, F) \preceq^1_t G$,
\[
\forall s, t \in C, d_{H_C}(s,t) \le t \cdot d_G(s, t).
\]
\end{corollary}

By dynamically maintain this construction, we have the following dynamic data structure:

\begin{lemma}
\label{lemma:DynamicMetricJTree}
    Given a graph $G=(V,E,w)$, positive integer $j$, spanning tree $T$ and a subset of edges $F$ of size $O(j)$.
    Suppose $J^1_G(T,F) \preceq^1_t G$ for some $t > 0$,
    we can maintain a vertex sparsifier $\Htil$ that maintains terminal-wise distance up-to $t$-factor that supports up to $O(j)$ of following operations:
    \begin{enumerate}
        \item $\textsc{Initialize}(G, T, F)$:
              Build data structures for maintaining $H$ in $O((mn/j)\log{n})$-time.
        \item $\textsc{AddTerminal}(u)$:
              Add $u$ to the terminal set of $C=C_G(T,F)$.
              Such operation can be done in amortized $O((mn/j^2)\log^3{n})$-time.
        \item $\textsc{Insert}(u, v, c)$:
              Insert the edge $(u, v)$ to $G$ in amortized $O((mn/j^2)\log^3{n})$-time.
        \item $\textsc{Delete}(e)$:
              Delete the edge $e$ from $G$ in amortized $O((mn/j^2)\log^3{n})$-time.
    \end{enumerate}
    The total number of edge changes in $\Htil$ is $O(mn/j)$, hence each operation has amortized recourse of $O(mn/j^2)$.
    Also, $\Htil$ has $O(n^{1+o(1)})$ edges.
\end{lemma}

\subsection{Dynamic Metric $J$-Tree}

In this section, we present tools and arguments that help us prove Lemma~\ref{lemma:DynamicMetricJTree}.

\subsubsection{Structural arguments}

To prove Lemma~\ref{lemma:DynamicMetricJTree}, one has to make sure adding terminals does not increase the stretch.
The argument is formalized as the following lemma:
\begin{lemma}
\label{lemma:AddTerminalPreserveDistance}
    Given a graph $G=(V,E,l)$, a spanning forest $T$ of $G$, a subset of vertices $C$ and $F$, a subset of edges.
    Suppose there is no branch vertex in $T$ with respect to $C$, i.e. $V(S(T,C)) = C$.
    Let $u$ be a vertex.
    We have $V(S(T,C+u)) \supseteq C$ and $|V(S(T,C+u)) \setminus C| \le 2$.
    
    Furthermore, let $H \coloneqq \textsc{Route}^1(G, T, C, F)$.
    If $H \preceq_\alpha G$, then the graph $\overline{H} \coloneqq \textsc{Route}^1(G,T,V(S(T,C+u)),F) \preceq_\alpha G$.
\end{lemma}

\begin{proof}
    Suppose we root $T$ at $u$.
    Since $V(S(T,C)) = C$, then either i) all vertex in $C$ lies in one subtree of $T$ or ii) $u$ lies in a path connecting 2 vertices of $C$.
    If i) happens, $V(S(T,C+u)) \setminus C$ has $u$ and $x$, the lowest common ancestor of all vertices of $C$.
    If ii) happens, $V(S(T,C+u)) \setminus C = \{u\}$.
    
    Observe that we route every edge of $G$ in $\overline{H}$ using a shorter path.
    Thus the stretch does not increase.
\end{proof}

To maintain metric-$O(j)$-tree $H$ under dynamic edge updates in $G$, we first add both endpoints of the updating edge to the terminal and then perform the edge update in $H_C$, the vertex sparsifier constructed from $H$.
One has to make sure such behavior does not increase the stretch when routing $G$ in $H$.
The following lemma gives such promise:
\begin{lemma}
\label{lemma:CoreEdgeUpdatePreserveDistance}
    Given a graph $G=(V,E,l)$, a spanning forest $T$ of $G$, a subset of vertices $C$ and $F$, a subset of edges.
    Let $H \coloneqq \textsc{Route}^1(G, T, C, F)$, and $e=uv$ be any edge with $u,v \in C$ ($e$ might not be in $G$) with length $l_e$.
    First note that $G[C] \subseteq H[C]$.
    
    If $H \preceq^1_\alpha G$, then both $(H+e) \preceq^1_\alpha (G+e)$ and $(H-e) \preceq^1_\alpha (G-e)$ holds.
\end{lemma}
\begin{proof}
    For $(H+e) \preceq^1_\alpha (G+e)$, $H+e$ can route edges in $G$ via the routing that implements $H \preceq^1_\alpha G$.
    For newly added edge $e$, $H+e$ route it using the $e$.
    
    Since $H$ routes $G$ by tree-terminal path, edges not in $H[C]$ are routed without using $e \in H[C]$.
    For $(H-e) \preceq^1_\alpha (G-e)$, all edges of $G-e$ can be routed by $H-e$ using the old routing.
\end{proof}

\subsubsection{Data structure toolbox}



To keep the resulting vertex sparsifier has small number of edges, we use the dynamic spanner data structure from \cite{FG19}.

\begin{lemma}[\cite{FG19}]
\label{lemma:DynamicSpanner}
Given a graph $G=(V,E,c)$ with weight ratio $U=\poly(n)$,
there is a randomized fully dynamic data structure on maintaining a spanner of $G$ with stretch $(1+\epsilon)(2k-1)$ and expected size $O(n^{1+1/k}\log^2{n}\epsilon^{-1})$ with expected amortized update time $O(k\log^3{n})$.
\end{lemma}

\subsubsection{Proof of Lemma~\ref{lemma:DynamicMetricJTree}}

\begin{proof}[Proof of Lemma~\ref{lemma:DynamicMetricJTree}]
The data structure is almost the same as Lemma~\ref{lemma:dynamic_JTree}.
Except we use Lemma~\ref{DynamicSkeletonTree} to update edges in $\Htil$ that corresponds to edge changes from $S(T,C)$ to $S(T, C+u)$.
Also, instead of dynamic cut sparsifier, we use dynamic spanner from Lemma~\ref{lemma:DynamicSpanner} to reduce the number of edges in the resulting vertex sparsifier.
\end{proof}

\subsection{Put everything together}

From Theorem~\ref{thm:J1MetricDecomp}, we know with probability 0.5, some $G_i$ sampled from distribution defined by $\lambda_i$'s preserve distance up to $4\alpha$-factor.
Given this, our dynamic data structures first sampled $t=O(\log{n})$ of them, say, $G_1, \ldots, G_t$.
Given any query $s, t$, we have
\begin{enumerate}
    \item $d_G(s, t) \le d_{G_i}(s, t), \forall s, t$ deterministically, and
    \item $\min_{1 \le i \le t}d_{G_i}(s, t) \le 4\alpha d_G(s, t)$ with high probability.
\end{enumerate}

In order to compute $d_{G_i}(s, t)$ efficiently, we have to reduce both number of edges and vertices of $G_i$'s.

\begin{proof}[Proof of Theorem~\ref{theorem:DynamicAPSP}]
Let $j = m^{2/3}, k = \Theta((m\log^3{n})/ j)$.
Just like Theorem~\ref{theorem:DynamicMinCut} for dynamic max flow, we first apply Theorem~\ref{thm:J1MetricDecomp} to acquire a $(k, \Otil(\log{n}), \Theta(j))$-metric-decompostion of $G$, say, $\{(\lambda_i, H_i)\}_{1 \le i \le k}$.
Then we sample $t=O(\log{n})$ of them, say, $H_1, \ldots, H_t$ with probability $\Pr(\text{draw }H_j) = \lambda_j$.

For each of $H_i, i \le t$, we incur Lemma~\ref{lemma:DynamicMetricJTree} to construct a data structure, say $D_i$, that support dynamic operations.

Every $j$ operations, we rebuild the whole thing from scratch.
Rebuild takes $O(km\log^{n})=O((m^2\log^4{n})/j)$-time.
We charge the rebuild cost to these $j$ operations, each of them now is charged with $O((m^2\log^4{n})/j^2)=O(m^{2/3}\log^4{n})$-time.
Each edge update, we propagate them to these $t$ data structures.
Each of the $D_i$ can handle edge update in amortized $O((mn\log^3{n})/j^2)=O(m^{2/3}\log^3{n})$-time.
Since there are $t=O(\log{n})$ of them, each edge update can be handled in amortized $O(m^{2/3}\log^4{n})$-time.

Given a query $s, t$, we compute $st$-distance in each of the $t$ vertex sparsifiers of size $O(j^{1+o(1)})=O(m^{2/3+o(1)})$.
With high probability, $\min_{1 \le i \le t}{d_{D_i}(s, t)} \le \Otil(\log{n})d_G(s, t)$.
We output the minimum of these $t$ distances.
Since we can compute exact distance in $O(|E|+|V|\log|V|)$-time, the $st$ query can be handled in $O(t \times m^{2/3+o(1)}) = O(m^{2/3+o(1)})$-time.
\end{proof}

\section{Fully-Dynamic Effective Resistance}
In this section, we utilize the local sparsifier construction from \cite{DurfeeGGP19}.
By apply the construction recursively, we can speed-up the data structure.
It is formalized as the following theorem.

\begin{theorem}
\label{theorem:DynammicER}
    Given a positive integer $d$, error term $\epsilon \in (0, 1)$, and a graph $G=(V, E, c)$ with weight ratio $U=\poly(n)$.
    There is a dynamic data structure maintaining $G$ subject to the following operations:
    \begin{enumerate}
        \item $\textsc{Insert}(u, v, c)$: Insert the edge $(u, v)$ to $G$ in amortized $O(n^{2/3+1/(3d+3)}\epsilon^{-(2d+4)}\log^{2d+11}{n})$-time.
        \item $\textsc{Delete}(e)$: Delete the edge $e$ from $G$ in amortized $O(n^{2/3+1/(3d+3)}\epsilon^{-(2d+4)}\log^{2d+11}{n})$-time.
        \item $\textsc{ER}(s, t)$: Output a $(1+2d\epsilon)$-approximation to $\ER^G(s, t)$ in $O\left(n^{2/3+1/(3d+3)}\epsilon^{-(2d+2)}\log^{9d+10}{n}\right)$-time.
    \end{enumerate}
    All of above guarantees hold with high probability.
\end{theorem}

\subsection{Effective Resistance, Random Walks and $L_2$-Embeddability} 

In this subsection, we define notions related to the \emph{Laplacian} of a graph.
Notions and properties of Laplacians are critical in building the desired data structure.
Boldface is used to indicate that the variable is a vector.
We use $\CHI_i$ to denote the vector with $i$-th coordinate being 1 and 0 elsewhere.
Also define $\CHI_{i, j} = \CHI_i - \CHI_j$.

\begin{definition}
\label{definition:Laplacian}
    Given a graph $G=(V, E, c)$.
    Let $\B \in \Real^{V \times E}$ be the edge-incidence matrix of $G$, i.e., $\B_e = \CHI_{u, v}$, $\forall e=uv \in E$ with arbitrary orientation.
    Let $\C \in \Real^{E \times E}$ be the diagonal matrix with $\C_{e,e}=c(e)$, $\forall e \in E$.
    The \emph{Laplacian} $\L_G \in \Real^{V \times V}$ of $G$ is defined as
    \begin{align*}
        \L_G \coloneqq \T{\B}\C\B.
    \end{align*}
    That is, $\L_{G, u, u}=\sum_{u \in e \in E}{c(e)}$, is the weighted degree of $u$.
    And $\L_{G, u, v}=-\sum_{e=uv \in E}{c(e)}$, the minus sum of edge weights between $u, v$.
    
    For edge weight always being non-negative in our discussion, we define the Laplacian norm with respect to $x$ by
    \begin{align*}
        \norm{x}_{\L_G} \coloneqq \sqrt{\T{x}\L_Gx}.
    \end{align*}
\end{definition}
Also, Moore-Penrose pseudo-inverse of $\L_G$ is defined as $\pinv{\L_G}$.
Using this definition, we can define \emph{effective resistance} between 2 vertices.

\begin{definition}
\label{definition:EffectiveResistance}
Given $G=(V, E, c)$, the \emph{effective resistance} between 2 vertices $u$ and $v$ is defined as
\[
    \ER^G(u, v) \coloneqq \T{\CHI_{u,v}}\pinv{\L_G}\CHI_{u, v}.
\]
\end{definition}
Laplacian system solver is used to compute effective resistance.
As computing $\pinv{\L_G}\CHI_{u, v}$ is essentially solving for $x$ under $\L_Gx = \CHI_{u, v}$.
Therefore, $\ER^G(u, v) = x_u - x_v$.
The fastest solver is due to ~\cite{CKMPPRX14} which formalized as follows:
\begin{lemma}~\cite{CKMPPRX14}
\label{lemma:LaplacianSolver}
    Given a graph $G=(V, E, c)$, $b=\L_Gx^*$ and error term $\epsilon > 0$.
    There is an algorithm that finds $x$ w.h.p.
    such that
    \[
        \norm{x^*-x}_{\L_G} \le \epsilon\norm{x^*}_{\L_G}
    \]
    in $O(m\sqrt{\log{n}}\log{\frac{1}{\epsilon}}\cdot (\log\log{n})^{3 + \delta})$-time for any constant $\delta$.
\end{lemma}
Using this fast solver, effective resistance between vertices can be computed efficiently.

\begin{corollary}
\label{corollary:ERSolver}
    Given a graph $G=(V, E, c)$, $u, v \in V$ and error term $\epsilon > 0$.
    We can compute a value $\phi$ w.h.p.
    such that
    \[
        (1-\epsilon)\ER^G(u, v) \le \phi \le (1+\epsilon)\ER^G(u, v)
    \]
    in $\Otil(m)$-time.
\end{corollary}

\begin{definition}
\label{definition:SpectralApproximation}
    Given a graph $G=(V, E, c)$ and $\epsilon \in (0, 1)$, we say a graph $H=(V, E_H \subseteq E, c_H)$ is a \emph{$(1+\epsilon)$-spectral-sparsifier} of $G$ if $\forall x \in \Real^{V}$,
    \[
        (1-\epsilon)\T{x}\L_Gx \le \T{x}\L_Hx \le (1+\epsilon)\T{x}\L_Gx.
    \]
    Denoted by $H \approx_\epsilon G$.
\end{definition}

\begin{fact}
    $H \approx_\epsilon G$ implies $H \sim_\epsilon G$.
\end{fact}

To speed up, we would like to compute effective resistance in the sparsified graph instead of the original one.
The following fact supports motivation.

\begin{fact}
    If $H \approx_\epsilon G$, $\forall u, v \in V$ we have
    \[
        (1-\epsilon)\ER^G(u, v) \le \ER^H(u, v) \le (1+\epsilon)\ER^G(u, v).
    \]
\end{fact}
For a graph $G=(V, E, c)$ with non-negative weight, we can define \emph{random walk} in $G$ using the distribution proportional to edge weight.
That is, given a walk $u_0, u_1, \ldots, u_k$, the probability
\[
    \Pr_G\left(u_{k+1}=v \mid u_0, \ldots, u_k\right) = \frac{c(u_kv)}{c(u_k)}.
\]
Alternatively, by fixing a starting vertex $u_0$, a walk $u_0, \ldots, u_k$ is sampled with probability
\[
    \Pr_G\left(u_0, \ldots, u_k\right) = \prod_{i=0}^{k-1}{\frac{c(u_iu_{i+1})}{c(u_i)}}.
\]

\subsection{Dynamic Spectral Sparsifier}
\begin{lemma}~\cite{AbrahamDKKP16}
\label{lemma:DynamicSpectralSparsifier}
    Given a graph $G=(V, E, c)$ with weight ratio $U$.
    There is an $(1 + \epsilon)$-spectral sparsifier $H$ of $G$ w.h.p..
    Such $H$ supports the following operations:
    \begin{itemize}
        \item $\mathrm{Insert}(u, v, c)$: Insert the edge $(u, v)$ to $G$ in amortized $O(\log^9{n}\epsilon^{-2})$-time.
        \item $\mathrm{Delete}(e)$: Delete the edge $e$ from $G$ in amortized $O(\log^9{n}\epsilon^{-2})$-time.
    \end{itemize} The weight ratio of $H$ is $O(nU)$.
    Moreover, the size of $H$ is $O(n\log^9{n}\epsilon^{-2})$.
\end{lemma}

\subsection{Schur Complement as Vertex Sparsifier}
To design an efficient data structure for computing effective resistance, we need a smaller graph preserving desired information.

\begin{definition}
\label{definition:SchurComplement}
    Given a graph $G=(V, E, c)$ and $C \subseteq V$.
    Write $D = V \setminus C$ We write the Laplacian of $G$ as
    \[
        \L =
        \begin{bmatrix}
            \L_{[C, C]} & \L_{[C, D]} \\
            \L_{[D, C]} & \L_{[D, D]} \\
        \end{bmatrix}.
    \]
    The \emph{Schur Complement} of $G$ onto $C$, denoted by $\SC(G, C)$, is the matrix obtained after performing Gaussian Elimination on variables corresponds to $D$.
    The closed form is given by
    \[
        \SC(G, C) = \L_{[C, C]} - \L_{[C, D]}\L_{[D, D]}^{-1}\L_{[D, C]}.
    \]
\end{definition}

\begin{fact}
\label{fact:RecursiveSchurComplement}
    Given $C_1 \subseteq C_2 \subseteq V$.
    We have $\SC(G, C_1) = \SC(\SC(G, C_2), C_1)$.
\end{fact}

\begin{fact}
    $\SC(G, C)$ is a Laplacian matrix of some graph with vertex set $C$.
\end{fact}

Using this fact, we abuse the notation by using $\SC(G, C)$ to denote both the Laplacian and the corresponding graph.
An important property about $\SC(G, C)$ is that it preserves effective resistance.

\begin{fact}
\label{fact:SchurComplementPreserveER}
    \[
        \forall u, v \in C, \ER^G(u, v) = \ER^{\SC(G, C)}(u, v).
    \]
\end{fact}
If we can efficiently maintain the Schur Complement for some $C$ small enough together with edge sparsification scheme, we can compute effective resistance in a much smaller graph using Lemma~\ref{lemma:LaplacianSolver}.

Building the $\SC(G, C)$ naively is a sequential process, which does not fit in our dynamic data structure paradigm.
In ~\cite{DurfeeGGP19}, they provide an alternative way of constructing $\SC(G, C)$ approximately using random walks.
The following lemma justifies this approach.

\begin{lemma}~\cite{DPPR17}
Given any graph $G=(V, E, c)$ and a subset of vertices $C \subseteq V$.
Given any walk $w=v_0, \ldots v_l$, we say $w$ is \emph{terminal-free} if $w \cap T = \{v_0, v_l\}$.
The Schur complement $\SC(G, C)$ is given as an union over all multi-edges corresponding to \emph{terminal-free} walks $v_0, \ldots v_l$ with weight
\[
    c(v_0v_1)\prod_{i=1}^{l-1}{\frac{c(v_iv_{i+1})}{c(v_i)}}.
\]
\end{lemma}

\subsection{Dynamic Schur Complement}

In this section, we introduce a dynamic data structure for maintaining $\SC(G, C)$ ~\cite{DurfeeGGP19}.
We adapt this data structure for faster effective resistance computation in the dynamic graph.
To design a faster data structure, we deploy a recursive routine and slightly modify the tool.
The result is stated as the following lemma: \begin{lemma}~\cite{DurfeeGGP19}
\label{lemma:DynamicSchurComplement}
    Given $\beta \in (0, 1)$, error term $\epsilon \in (0, 1)$, a graph $G=(V, E, c)$ with weight ratio $U$, and a terminal set $T \subseteq V$ with $|T| = \beta m$.
    There is a dynamic data structure $\mathcal{D}$ that maintains $\Htil \approx_{\epsilon} \SC(G, C)$ for some $T \subseteq C$ with $|C| = \Theta(m\beta)$.
    Such $\Htil$ is a subgraph of $\SC(G, C)$ with $O(m\beta\epsilon^{-2}\log^9{n})$ edges and weight ratio $O()$.
    Such data structure supports up to $O(m\beta)$ of the following operations:
    \begin{enumerate}
        \item $\textsc{Initialize}(G, T, \beta)$: Initialize $\mathcal{D}$ in $O(m\beta^{-4}\epsilon^{-4}\log^{4}{n})$ time.
        \item $\textsc{Insert}(u, v, w)$: Insert an new edge $uv$ with weight $w$ to $G$ in amortized $O(\beta^{-2}\epsilon^{-2}\log^3{n})$ time.
        \item $\textsc{Delete}(e)$: Delete the edge $e$ from $G$ in amortized $O(\beta^{-2}\epsilon^{-2}\log^3{n})$ time.
        \item $\textsc{AddTerminal}(u)$: Move vertex $u$ to $T$ in amortized $O(\beta^{-2}\epsilon^{-2}\log^3{n})$ time.
    \end{enumerate} $\mathcal{D}$ maintains $O(m\epsilon^{-2}\log{n})$ random walks each with $O(\beta^{-1}\log{n})$ distinct vertices.
    For every vertex $u \not \in C$, total occurrence of $u$ in these random walks is $O(\beta^{-2}\epsilon^{-2}\log^2{n})$.
    The total number of edge change in $H$ is $O\left(m\beta^{-1}\epsilon^{-2}\log^2{n}\right)$.
    All the above bounds hold with high probability.
\end{lemma}
Here we briefly outlined the process for $\mathrm{Initialize}(G, T, \beta)$ and point out a small modification from the original construction in ~\cite{DurfeeGGP19}.

\begin{enumerate}
    \item Initialize $C$ with $T$.
          For each edge $uv$, add both $u, v$ to $C$ with probability $\beta$.
          Let $H$ be an empty graph.
    \item For each edge $e=uv$, sample 2 random walks $w_u, w_v$ starting from $u$ and $v$ respectively.
          Such random walk is generated until either a vertex in $C$ is hitted or $\Theta(\beta^{-1}\log{n})$ distinct vertices is visited.
          For both walks hit vertices $a, b \in C$, Let $r$ be the resistance between $ab$ in the path $w_u+e+w_v$.
          An edge $ab$ with weight $1/r$ is added to $H$.
    \item Repeat above step for $\rho \coloneqq \epsilon^{-2}\log{n}$ times and scale edge weights in $H$ by $1/\rho$.
          Therefore, there are $O(m\epsilon^{-2}\log{n})$ walks generated and their total size is $O(m\beta^{-1}\epsilon^{-2}\log^2{n})$.
          Balanced binary search trees are used to maintain these walks.
          Also, a reverse index for each vertex to locate the position in every walk containing it.
    \item Identify the top $\beta m$ vertices with most occurrence in these $O(m\epsilon^{-2}\log{n})$ walks, add them to $C$ as well.
          Meanwhile, shortcut those walks involving these vertices.
          By a Markov-type argument, vertex not in $C$ has occurrence at most $O(\beta^{-2}\epsilon^{-2}\log^2{n})$.
    \item Incur a dynamic $(1+\epsilon)$-spaectral sparsifier from Lemma~\ref{lemma:DynamicSpectralSparsifier} on $H$.
          The resulting sparse graph $\Htil$ is the desired approximator for $\SC(G, C)$.
\end{enumerate}
The only difference from ~\cite{DurfeeGGP19} is the 4th step, which is crucial for efficiency.
As this step can be viewed as performing $\beta m$ $\mathrm{AddTerminal}(u)$ operations, which does not affect the approximation guarantee.

\subsection{The Main Result}
In this section, we prove the main theorem for dynamic effective resistance.
The high-level idea is to utilize the Fact~\ref{fact:RecursiveSchurComplement} and build layers of Schur Complements.
The effective resistance computation is performed in the last layer, which has the smallest size.

\begin{proof}[Proof of Theorem~\ref{theorem:DynammicER}]
    Let $\beta = n^{-1/(3d+3)}$.
    First we incur a dynamic $(1 + \epsilon)$-spectral-sparsifier on $G$.
    Let $G_0$ be the sparsified $G$.
    Define a chain of graphs $G_1, G_2, \ldots G_d$, where $G_{i+1}$ is the sparsified Schur Complment on $\beta|E(G_i)|$ vertices ,i.e.  $G_{i+1} \approx_\epsilon \SC(G_i, C_i)$ for some $C_i$ of size $\beta|E(G_i)|$.
    Each $G_{i+1}$ is maintained using Lemma~\ref{lemma:DynamicSchurComplement} on $G_i$ with $\beta$ and error term $\epsilon$.
    Also, we need to rebuild $G_{i+1}$ every $\beta|E(G_i)|$ steps.
    For every edge updates, it is propagated down to $G_d$.
    For $\mathrm{ER}(s, t)$ query, both $s, t$ are added to the terminal set of every $G_i$s.
    And then we incur Lemma~\ref{lemma:LaplacianSolver} to compute $\ER^{G_d}(s, t)$, which gives a $(1+\epsilon)^d=1+O(d\epsilon)$-approximation.
    Now we analyze the running time for this data structure.
    Let $n_i = |V(G_i)|$ and $m_i = |E(G_i)|$, it is clear that $S \coloneqq \log^9{n}\epsilon^{-2} = \frac{m_i}{n_i}, \forall i$.
    Also, one edge update or terminal add in $G_i$ creates amortized $\gamma \coloneqq \beta^{-2}\epsilon^{-2}\log^{-2}{n}$ changes to $G_{i+1}$.
    For a single $G_i$, the rebuild cost is spread across $\beta m_i$ operations.
    That is,
    \[
        \frac{m_i\beta^{-4}\epsilon^{-4}\log^4{n}}{\beta m_i} = \beta^{-5}\epsilon^{-4}\log^4{n}
    \]
    which dominates the actual time complexity for each dynamic operations.
    Also, every rebuild creates $m_i$ changes to $G_{i+1}$.
    By distributing them across $\beta m_i$ operations before next rebuild, every change in $G_i$ creates amortized $\beta^{-1}$ number of changes to $G_{i+1}$.
    But $\beta^{-1} = o(\gamma)$, therefore we can still bound the changes in $G_{i+1}$ per change in $G_i$ by $O(\gamma)$.
    For every change in $G$, it creates $S \cdot \gamma^{i-1}$ changes to $G_i$.
    And each change costs $\beta^{-5}\epsilon^{-4}\log^4{n}$ to handle.
    So for every change, the amortized time can be expressed as
    \begin{align*}
        O\left(\sum_{i=1}^{d}{S \cdot \gamma^{i-1} \cdot \beta^{-5}\epsilon^{-4}\log^4{n}}\right) &=
        O\left(\beta^{-(2d+3)}\epsilon^{-(2d+4)}\log^{2d+11}{n}\right) \\
        &= O\left(n^{2/3+1/(3d+3)}\epsilon^{-(2d+4)}\log^{2d+11}{n}\right).
    \end{align*}
    And for $\mathrm{ER}(s, t)$ query, Lemma~\ref{lemma:LaplacianSolver} is ran on $G_d$, which has $m_d = nS(\beta S)^d$ edges.
    So each query, we can bound the time by
    \begin{align*}
        O\left(m_d\log{n}\right) &= O\left(n\beta^d\epsilon^{-(2d+2)}\log^{9d+10}{n}\right) \\ &= O\left(n^{2/3+1/(3d+3)}\epsilon^{-(2d+2)}\log^{9d+10}{n}\right).
    \end{align*}
\end{proof}

\appendix

\section{Proof of Lemma~\ref{lemma:GeneralizedLSST}}

To prove Lemma~\ref{lemma:GeneralizedLSST}, we use the following algorithm from \cite{ABN08} that computes a low-stretch-spanning tree in $O(m\log n)$-time.

\begin{lemma}[\cite{ABN08}]
\label{lemma:LSST}
Given a graph $G=(V,E,l)$, there is an algorithm that computes a spanning tree $T$ such that
\[
    \frac{1}{|E|}\sum_{e=uv \in E}{\frac{d_T(u,v)}{l(e)}} \le \alpha,
\]
where $\alpha = O(\log n \log\log n) = \Otil(\log n)$.
The algorithm runs in $O(m\log n)$-time.
\end{lemma}

\begin{proof}[Proof of Lemma~\ref{lemma:GeneralizedLSST}]
In order to use Lemma~\ref{lemma:LSST}, we have to convert the graph into the one without edge weights.
For every $e \in E$, define
\[
    r(e) = 1 + \left\lfloor \frac{l(e) w(e) |E|}{\bw(G)} \right\rfloor,
\]
and create a graph $\overline{G}$ identical to $G$ except we add $r(e) - 1$ more parallel edges for every edge $e \in E$.
Then we apply Lemma~\ref{lemma:LSST} on $\overline{G}$ and return the resulting spanning tree $T$.
Such $T$ satisfies
\begin{align*}
    \sum_{e=uv \in E}\frac{d_T(u, v)r(e)}{l(e)} \le \alpha |E(\overline{G})|
\end{align*}

First, we show the running time by bound the size of $\overline{G}$.
Note that
\begin{align*}
    |E(\overline{G}| = \sum_{e \in E}r(e) \le
    \sum_{e \in E}\left(1 + \frac{l(e) w(e) |E|}{\bw(G)} \right) \le
    |E| + |E| = 2|E|.
\end{align*}
Thus, the algorithm runs in $O(m\log n)$-time.

Next, we show the approximation guarantee.
Observe that for every edge $e$,
\begin{align*}
    r(e) \ge \frac{l(e) w(e) |E|}{\bw(G)} \ge
    \frac{l(e) w(e)}{\bw(G)}\frac{\sum_{f \in E}r(f)}{2}.
\end{align*}
Plug in this lower bound and we have
\begin{align*}
    \alpha \sum_{e \in E}r(e) \ge
    \sum_{e=uv \in E}\frac{d_T(u, v)r(e)}{l(e)} \ge
    \sum_{e=uv \in E}\frac{d_T(u, v)w(e) \sum_{f \in E}r(f)}{2\bw(G)},
\end{align*}
and hence,
\begin{align*}
    2 \alpha \bw(G) \ge \sum_{e=uv \in E}d_T(u, v)w(e).
\end{align*}

Such $T$ satisfies the requirement of this lemma.

\end{proof}

\section{Multiplicative-Weight-Update methods}

\subsection{Proof of Theorem~\ref{lemma:J1MetricDecompMWU}}
\label{sec:J1MetricDecompMWUProof}

\KTree*

Let $G=(V,E,l)$ be a graph.
Let $\{H_i\}_i$ be the set of graphs in $\CG$ such that $G \preceq^1 H_i \preceq^1_{t_i} G$.
Introduce a coefficient $\lambda_i$ for each $H_i$.
These $\lambda_i$'s are initially zero and in the end only a small number of them will become nonzero.
For a given graph $H_i = (V, E_i, l_i)$ that $t_i$-routes $G$, recall the definition of $\eta_{H_i}(e) = l_{H_i}(e) / l(e), \forall e \in G$, stretch of $e$ routed by $H_i$.

Following approach from \cite{Madry10} and \cite{Racke08}, let $M$ be an $|E| \times N$ matrix, $N$ being the cardinality of $\{H_i\}_i$, with $M_{e, i} = \eta_{H_i}(e)$.
Let $\blbd = (\lambda_1, \ldots, \lambda_N)$ be a vector corresponding to a convex combination of $\{H_i\}_i$.
If $\max{M\blbd} \le \alpha$ for some $\alpha > 0$, we have
\begin{align*}
    \forall e \in E, \sum_i{\lambda_i \eta_{H_i}(e)} &\le \alpha \text{, and thus} \\
    \forall e \in E, \sum_i{\lambda_i l_{H_i}(e)} &\le \alpha l(e) \\
\end{align*}

Consider the following constrains:
\begin{align*}
    \texttt{lmax}(M\blbd) &\le 3\alpha \\
    \sum_i{\lambda_i} &= 1 \\
    \lambda_i &\ge 0, \forall i \\
\end{align*}
where $\texttt{lmax}(\bx) = \ln{\sum_{e \in E}{\exp(x_e)}} \ge \max_{e \in E}{x_e}$.
Any valid $\blbd$ would corresponds to a $(3\alpha, \CG)$-distance-decompostion of $G$ for $\alpha \ge \ln m$.

Here we present some known fact about $\texttt{lmax}$ function:
\begin{fact}
\[
    \forall \bx \ge 0, \max_{e \in E}{x_e} \le \texttt{lmax}(\bx) \le \max_{e \in E}{x_e} + \ln m
\]
\end{fact}
\begin{fact}
\label{fact:lmaxConvex1}
Define
\[
    \texttt{partial}_e(\bx) \coloneqq
    \frac{\partial\texttt{lmax}(\bx)}{\partial x_e} = 
    \frac{\exp(x_e)}{\sum_f\exp(x_f)}.
\]
We have
\[
    \forall \bx, \beps \ge 0, \beps \le 1,
    \texttt{lmax}(\bx+\beps) \le
    \texttt{lmax}(\bx) + 2\sum_e{\epsilon_e\texttt{partial}_e(\bx)}
\]
\end{fact}

\begin{fact}
\label{fact:lmaxConvex2}
Define 
\[
    \texttt{partial}_i(\blbd) \coloneqq
    \frac{\partial\texttt{lmax}(M\blbd)}{\partial \lambda_i} = 
    \sum_e{\eta_{H_i}(e) \cdot \texttt{partial}_e(M\blbd)} = 
    \sum_e{\frac{l_{H_i}(e)}{l(e)} \cdot \texttt{partial}_e(M\blbd)}.
\]
Let $1_i$ be the all zero vector except $i$-th coordinate being 1.
Recall $\eta(H_i)$ being the maximum stretch when $H_i$ routes $G$.
For any $0 \le \delta_i \le 1 / \eta(H_i)$,
we have
\[
    \texttt{lmax}(M(\blbd + \delta_i 1_i)) \le
    \texttt{lmax}(M\blbd) + 2 \delta_i \texttt{partial}_i(\blbd)
\]
\end{fact}

\begin{proof}[Proof of Theorem~\ref{lemma:J1MetricDecompMWU}]
The vector $\blbd$ is found as follows.
Starting with $\blbd = 0$.
As long as $\sum_i\lambda_i < 1$, we define edge weights $\bw$ with $w(e) = \texttt{partial}_e(M\blbd) / l(e)$.
By condition~\label{MWU1}, we can compute $H_{i(\bw)}$ such that
\[
    \bw(H_{i(\bw)}) = \sum_e{l_{H_{i(\bw)}}(e)\frac{\texttt{partial}_e(M\blbd)}{l(e)}} \le
    \alpha \sum_e{l(e)\frac{\texttt{partial}_e(M\blbd)}{l(e)}} =
    \alpha \bw(G).
\]
Next we increase $\lambda_{i(\bw)}$ by $\min\{1/\eta(H_{i(\bw)}), 1 - \sum_i\lambda_i\}$.

Let $\blbd$ be the resulting solution.
We have to make sure $\texttt{lmax}(M\blbd) \le 3\alpha$.
First, we observe that
\begin{align*}
    \bw(H_{i(\bw)}) =
    \sum_e{l_{H_{i(\bw)}}(e)\frac{\texttt{partial}_e(M\blbd)}{l(e)}} = \texttt{partial}_i(\blbd),
\end{align*}
and
\begin{align*}
    \bw(G) =
    \sum_e{l(e)\frac{\texttt{partial}_e(M\blbd)}{l(e)}} = \sum_e\texttt{partial}_e(M\blbd) = 1.
\end{align*}
Therefore, we have $\texttt{partial}_i(\blbd) = \bw(H_{i(\bw)}) \le \alpha \bw(G) = \alpha$.

By Fact~\ref{fact:lmaxConvex2} and $\lambda_{i(\bw)} \le 1/\eta(H_{i(\bw)})$,
at each iteration, we have
\begin{align*}
    \texttt{lmax}(M\blbd) &\le
    \texttt{lmax}(M\mathbf{0}) + 2\sum_i{\lambda_i\bw(H_{i(\bw)})} \\
    &\le \texttt{lmax}(M\mathbf{0}) + 2\sum_i{\lambda_i\alpha} \\
    &\le \ln m + 2\alpha \le 3 \alpha
\end{align*}

Next, we have to show an upper-bound on the number of iterations.
Define the potential function $\Phi(\blbd) \coloneqq \sum_e\sum_i{\lambda_i\eta_{H_i}(e)}$.
Initially, $\Phi(\blbd) = \Phi(\mathbf{0}) = 0$.
The potential is only increasing throughout the algorithm.
At the end, we have $\Phi(\blbd) \coloneqq \sum_e\sum_i{\lambda_i\eta_{H_i}(e)} \le 3\alpha m$ since $\sum_i{\lambda_i\eta_{H_i}(e)} \le 3\alpha, \forall e$.
Observe every time we update $\blbd$, $\Phi(\blbd)$ increases by $|\psi(H_{i(\bw)})|/2 \ge 2\alpha m/k$,
since for every $e \in \psi(H_{i(\bw)})$,
\begin{align*}
    \lambda_{i(\bw)} \eta_{H_{i(\bw)}}(e) \ge
    \frac{\eta_{H_{i(\bw)}}(e)}{\eta(H_{i(\bw)})} \ge
    \frac{1}{2}
\end{align*}

Therefore, by Condition~\ref{MWU3} on lower-bounding the size of $|\psi(H_{i(\bw)})|$, we have at most $1.5k$ iterations and the theorem follows.
\end{proof}

\bibliographystyle{alpha}
\bibliography{literature_ISO}

\end{document}